\documentclass[12pt,a4paper]{article}
\usepackage[english]{babel}
\usepackage{float}
\usepackage{amsmath, amssymb, amsthm, epsfig, color, mathtools, bbm, dsfont, mathrsfs, bm}
\usepackage[T1]{fontenc}
\usepackage{indentfirst}
\usepackage[shortlabels]{enumitem}
\usepackage[title]{appendix}
\usepackage{lmodern}       
\usepackage{setspace}
\usepackage{soul}
\usepackage{accents}
\usepackage{hyperref} 
\usepackage{geometry}
\usepackage{xcolor}
\usepackage{relsize}
\usepackage{scalerel}

\usepackage{xr}    
\usepackage{graphicx} 
\usepackage[export]{adjustbox} 
\usepackage{subcaption} 
\usepackage{siunitx}
\usepackage{bm} 
\usepackage{comment}

\usepackage{addfont}
\addfont{OT1}{rsfs10}{\rsfs}

\usepackage{tikz, tkz-euclide}
\usetikzlibrary{patterns}
\usetikzlibrary{arrows.meta}
\usetikzlibrary{calc}
\usetikzlibrary{decorations.markings}
\usepackage{fancybox}

\newcommand{\be}{\begin{equation}}
\newcommand{\ee}{\end{equation}}

\newcommand{\ctr}{\gamma}

\newcommand{\supp}{\mathrm{supp}}

\newcommand{\diam}{\mathrm{diam}}
\newcommand{\dis}{\mathrm{dist}}

\newcommand{\I}{\mathrm{I}}
\newcommand{\Sp}{\mathrm{sp}}
\newcommand{\Z}{\mathbb{Z}}
\newcommand{\lab}{\mathrm{lab}}
\newcommand{\inB}{\partial_{\text{in}}}
\numberwithin{equation}{section}

\newcommand{\es}{\mathcal{E}}
\newcommand{\is}{\mathscr{I}}
\newcommand{\ctrb}{\Gamma}
\newcommand{\extb}{\Gamma^e}

\newcommand{\bc}{+}
\newcommand{\Ztil}{\widetilde{Z}}
\newcommand{\polymer}{\Gamma}
\newcommand{\polymerb}{X}

\newtheorem*{theorem*}        {Theorem}
\newtheorem*{conjecture*}   {Conjecture}
\newtheorem{theorem}{Theorem}[section]
\newtheorem{lemma}[theorem]{Lemma}
\newtheorem*{lemma*}          {Lemma}

\newtheorem{corollary}[theorem]{Corollary}
\newtheorem{proposition}[theorem]{Proposition}

\newtheorem{remark}          {Remark}[section]

\theoremstyle{definition}
\newtheorem{definition}[theorem]{Definition}

\makeatletter

\makeatletter
\def\moverlay{\mathpalette\mov@rlay}
\def\mov@rlay#1#2{\leavevmode\vtop{%
		\baselineskip\z@skip \lineskiplimit-\maxdimen
		\ialign{\hfil$\m@th#1##$\hfil\cr#2\crcr}}}
\newcommand{\charfusion}[3][\mathord]{
	#1{\ifx#1\mathop\vphantom{#2}\fi
		\mathpalette\mov@rlay{#2\cr#3}
	}
	\ifx#1\mathop\expandafter\displaylimits\fi}
\makeatother

\newcommand{\cupdot}{\charfusion[\mathbin]{\cup}{\cdot}}
\newcommand{\bigcupdot}{\charfusion[\mathop]{\bigcup}{\cdot}}
\newcommand{\rctr}{\polymerb_{\ctr}^{(n+1)}}

\newgeometry{vmargin={15mm}, hmargin={14.5mm,18.5mm}}

\allowdisplaybreaks
\begin{document}

\begin{center}
{\LARGE Cluster Expansion and Decay of Correlations for Multidimensional Long-Range Ising Models}
\vskip.5cm
Lucas Affonso$^{1}$, Rodrigo Bissacot$^{1}$, Jo{\~a}o Maia$^{1,2}$, Jo\~{a}o F. Rodrigues$^{1}$, Kelvyn Welsch$^{1}$
\vskip.3cm
\begin{footnotesize}
$^{1}$ Institute of Mathematics and Statistics (IME-USP), University of S\~{a}o Paulo, Brazil\\
$^{2}$ Beijing International Center for Mathematical Research (BICMR), Peking University, China.
\end{footnotesize}
\vskip.1cm
\begin{scriptsize}
emails: lucas.affonso.pereira@gmail.com, rodrigo.bissacot@gmail.com, maia.joaovt@gmail.com, joao.felipe.rodriguesp@gmail.com, kelvyn.emanuel@gmail.com
\end{scriptsize}

\end{center}

\begin{abstract}
    We develop the cluster expansion for the multidimensional multiscaled contours defined by three of us. These contours are suitable for long-range Ising models with interaction $J_{xy}=J(|x-y|)= J/|x-y|^\alpha$, $J>0$, and $\alpha>d$. As an application of the convergence of the cluster expansion at low temperatures, we study the decay of the truncated two-point correlation functions, showing that the decay is algebraic with coefficient $\alpha$.  
\end{abstract}

\section{Introduction}

The cluster expansion is one of the most powerful tools in statistical mechanics and mathematical physics \cite{Malyshev1980}, providing valuable information about the models to which it is applicable. Although it was historically first designed for high-temperature systems, its usefulness in more general contexts was soon recognized, having today abstract representations, see \cite{Bissacot2010, Fernndez2007, Jansen2022, Koteck1986, procacci2023cluster, Temmel2014}, where one studies the convergence of a ``gas'' of abstract objects, usually called polymers. Each polymer has a weight, which depends on a set of parameters dictated by the specific model of interest, for example, the temperature, chemical potential, size, etc. The main goal is to prove the convergence of a series corresponding to the free energy (or pressure), which typically can only be done in a small region of the parameter space. The convergence of such series can be applied in the study of correlations, large deviations, and others. The usage of contours as polymers to develop the expansion in low-temperature allowed the derivation of strong results since the onset, like those in the seminal work of Minlos and Sinai \cite{Minlos} and when Gallavotti, Martin-L\"{o}f and Miracle-Sole \cite{GMM} managed to prove results about coexisting phases for the Ising model.

 The Ising model \cite{Ising1925, Klske2025, Niss2004} is the most important model in statistical mechanics for studying phase transitions and critical phenomena in ferromagnetic systems. In its traditional form, the Ising model considers interactions only between nearest-neighbor spins on a lattice. When the model is extended to include long-range interactions --- where each spin interacts with all others with a coupling strength that decays as a power law --- new and complex behaviors emerge \cite{PhysRevE.89.062120, Iagolnitzer1977}. The first reader-friendly article proving the convergence of the cluster expansion for the Ising model seems to be \cite{Pfister1991LargeDA}, by Pfister.
 
 The development of low-temperature cluster expansions for short-range interactions highly benefits from the fact that, in this framework, contours do not interact with each other. In \cite{Park.88.I}, Park proved the convergence of the cluster expansion by dealing with the long-range interaction as a perturbation of a short-range one; the strategy forces the interactions to be weak, so the results for Ising models were obtained only for $\alpha > 3d + 1$. In addition, Park's argument is inspired by Pirogov-Sinai theory \cite{Pirogov.Sinai.75, Pirogov1976, Zahradnik.84} where the contours are connected objects. However, since the breakthrough of Fr\"{o}hlich-Spencer in \cite{Frohlich.Spencer.82}, it is well known that connected contours are not suitable for general long-range systems \cite{Affonso2024, affonso2024phasetransitionferromagneticqstate, Johanes, Bissacot_Corsini_cluster, Cassandro.05, Cassandro.Merola.Picco.17, Cassandro.Merola.Picco.Rozikov.14, Cassandro.Picco.09}. The results of Ginibre, Grossmann, and Ruelle \cite{Ginibre1966}, proving phase transition for $\alpha > d+1$ using a Peierls argument with connected contours (specifically, plaquettes as $(d-1)-$dimensional objects), suggest that a cluster expansion should also converge in this regime, but such a proof has never appeared in the literature to the best of our knowledge.

For one-dimensional long-range Ising systems, when $\alpha >2$ we have uniqueness of the Gibbs measure, and for $\alpha \in (1,2)$ Dyson in \cite{Dyson1969} proved the phase transition. The critical case $\alpha=2$ was solved in the seminal paper of Fr\"{o}hlich and Spencer \cite{Frohlich.Spencer.82} where they proved the phase transition and introduced for the first time the powerful idea of disconnected contours in long-range Ising models, inspired by their previous work \cite{FS81}. Imbrie \cite{Imbrie.82} then developed a cluster expansion for the same model in order to estimate the decay of the truncated two-point correlation function. After this, a good amount of literature was produced for one-dimensional long-range Ising models, taking inspiration from the contours defined by Fr\"{o}hlich and Spencer. In \cite{Cassandro.05}, Cassandro, Ferrari, Merola, and Presutti introduced a geometric notion of one-dimensional contour and proved the existence of a phase transition. Their definition was then used by Cassandro, Merola, Picco, and Rozikov \cite{Cassandro.Merola.Picco.Rozikov.14} to prove convergence of the cluster expansion for $\alpha \in \left( 3 - \frac{\ln 3}{\ln 2}, 2 \right]$, and to study the random field long-range model \cite{Cassandro.Picco.09}. These results, however, do not hold on the entire interval $(1,2]$ and also require that the nearest-neighbors interaction $J(1)$ is large enough; see also Littin and Picco \cite{Littin2017}, a restriction that is not present in the original work of Fr\"{o}hlich and Spencer. Later, Bissacot, Endo, van Enter, Kimura, and Ruszel  \cite{Bissacot.Endo.18} showed that this restriction can be removed for decays $\alpha$ even closer to $2$.

Only recently, Affonso, Bissacot, Corsini, and Welsch \cite{corsini} revisited this problem and proved the phase transition by a Peierls argument for arbitrary $J(1)$ and $\alpha \in (1,2]$. Under the same hypotheses, Bissacot and Corsini \cite{Bissacot_Corsini_cluster} proved the convergence of the cluster expansion at low temperatures. When $\alpha=2$, Imbrie and Newman \cite{Imbrie.Newman.88}, using the nearest neighbor interaction $J(1)$ as a perturbative parameter, found that $\langle\sigma_{x_1};\sigma_{x_2}\rangle^+\approx |x_1-x_2|^{-\theta(\beta)}$, where $0\leq \theta(\beta)\leq 2$. 

Inspired by the ideas from Fr\"{o}hlich-Spencer \cite{Frohlich.Spencer.82}, a sequence of papers and theses  \cite{Pereira, Affonso2024, affonso2024phasetransitionferromagneticqstate, Johanes, maia2024phase} from the Brazilian school of mathematical physics introduced multiscale contours in the multidimensional setting. First, Affonso, Bissacot, Endo, and Satoshi \cite{Affonso2024} introduced multiscale contours close in definition to the original, using diameters as in Fr\"{o}hlich and Spencer \cite{FS81}, and proved the phase transition in the whole region $\alpha > d$ even in the presence of a decaying field. After that, Affonso, Bissacot, and Maia \cite{Johanes} refined the definition, now using volumes instead of diameters; see also \cite{maia2024phase}, to prove the phase transition in the long-range random field Ising model in $d\geq 3$. This partially answered a conjecture made in the seminal paper by Bricmont and Kupiainen \cite{Bricmont.Kupiainen.88}. This refined definition of multidimensional Fr\"{o}hlich-Spencer contours using volumes seems to be the definitive one, and it was already used to complete the phase diagram of the bidimensional long-range random field Ising model by Ding, Huang and Maia \cite{johanes_china}, and also to prove the phase transition in ferromagnetic $q-$state systems, such as the Potts and Clock model by Affonso, Bissacot, Faria, and Welsch \cite{affonso2024phasetransitionferromagneticqstate}.

Despite more than four decades after the ideas of Fr\"{o}hlich and Spencer of disconnected contours for one-dimensional long-range systems, and almost six decades after the proof of the phase transition via contours for long-range multidimensional Ising models when $\alpha > d+1$ by Ginibre, Grossmann, and Ruelle, the main result of our paper is the first proof for the convergence of the cluster expansion at low temperatures for $\alpha > d$. The lack of a suitable definition of contours for multidimensional models partially explains the delay for this result; another more heuristic reason is the existence of three different regimes for the \textit{surface energy terms} of the long-range Ising models: denoting by $B_{R}$ the ball of radius $R$ and center $0$, we have

\[ F_{B_{R}} \coloneqq \displaystyle\sum_{\substack{x\in B_{R} \\ y\notin  B_{R}}}J_{xy} \approx
\begin{cases}
R^{2d - \alpha} & \text{if } d < \alpha < d + 1, \\
R^{d - 1} \log(R) & \text{if } \alpha = d + 1, \\
R^{d - 1} & \text{if } \alpha > d + 1.
\end{cases}
\]

More precisely, given two functions $f, g: \mathbb{Z}^d \to \mathbb{R}^{+}$, we say that $f$ is asymptotic to $g$, denoting by $f\approx g$, when there exist $N > 0$ and positive constants $A, A'$ such that $A'f(x) \leq g(x) \leq Af(x)$ for every $x$ such that $|x| > N$. Although we have to consider more complex regions than balls, this behavior justifies, in some sense, why the connected contours made by plaquettes should work only for $\alpha > d+1$. For a proof, see \cite{Pereira, Affonso2024, Biskup_Chayes_Kivelson_07}.

Our main application of the convergence of the cluster expansion is the proof that, in the low-temperature regime, the truncated two-point correlation function has the same algebraic decay as the interaction $J_{xy}$. For multidimensional Ising models with algebraically decaying interactions, this problem has remained open for more than four decades, since Iagolnitzer and Souillard proved that the truncated correlation cannot decay faster than the interaction in \cite{Iagolnitzer1977}.

The asymptotic behavior of correlations is well understood for the nearest-neighbor Ising model. In a recent breakthrough, Duminil-Copin, Goswami, and Raoufi \cite{Duminil_Copin_2019} demonstrated that the two-point correlation function \eqref{correlation} decays as $\exp(-c |x_1-x_2|)$, for an appropriate constant $c>0$ depending only on $\beta$ and $d$ for every $\beta$ different from the critical $\beta_c$ (the infimum of all $\beta$ such that $\langle\sigma_0\rangle_\beta^+>0$). They also claim that their results extend to general ferromagnetic short-range Ising models. 
Combined with earlier work by Lebowitz and Penrose \cite{Lebowitz1968} on ferromagnetic Ising models with a constant magnetic field $h$, this result provides a complete picture of the decay of correlations for these models outside the critical temperature $\beta_c$. Further refinements are also available: polynomial corrections to the decay, known as Ornstein-Zernike asymptotics, can be rigorously established \cite{Aoun+Ott+Velenik-2024, Ott+Velenik-2023}.

In contrast to the short-range case, the understanding of long-range systems remains incomplete. In \cite{Iagolnitzer1977}, it is proved that for ferromagnetic long-range Ising spin systems: (for precise definitions of the Hamiltonian and truncated correlation functions see Sections \ref{hamiltonian} and \ref{correlations}, respectively)
\[
\langle \sigma_{x_1}; \sigma_{x_2} \rangle_{\beta,\mathbf{h}}^+\geq C(\beta,\mathbf{h}) J_{x_1x_2},
\] 
where $C(\beta,\mathbf{h})>0$, for any $\beta$ and $\mathbf{h} \geq 0$. This result implies that, for interactions decaying polynomially, correlations cannot decay faster than polynomially. Notice that the result also holds, in particular, for any system with constant magnetic field $h$. This contrasts sharply with the short-range Ising model, where polynomial decay of correlations is observed only at the critical temperature $\beta_c$.  Our main result regarding the decay of correlations is Theorem \ref{main_decay}, proved in Section 3, and it can be stated as follows:
\begin{theorem*}
    For $\beta$ large enough, there exists a constant $c_4(\alpha,d,\beta) \coloneqq c_4>0$ such that for any distinct points $x_1, x_2 \in \Z^d$ it holds
 \[
 \langle\sigma_{x_1};\sigma_{x_2}\rangle^+_{\beta} \leq c_4 J_{x_1 x_2}. 
 \]
\end{theorem*}

The theorem above has an important corollary, concerning the decay of the correlation between local functions.

\begin{corollary}\label{cor: truncated correlation of functions}
    Let $f,g: \Omega \rightarrow \mathbb{R}$ two local functions with $\supp(f) \cap \supp(g) =\emptyset$ and let $\beta$ be large enough. Then there exists $C_{f,g}(\alpha,d,\beta) \coloneqq C_{f,g}$ such that
    \[
    \langle f; g \rangle_\beta^+ \leq \frac{J C_{f,g}}{\dis(\supp(f),\supp(g))^\alpha}, 
    \]
    where $\displaystyle \dis(\supp(f),\supp(g)) = \min_{\substack{x \in \supp(f) \\ y \in \supp(g)}}|x-y|$.
\end{corollary}

The proof of Corollary~\ref{cor: truncated correlation of functions} will be provided in Section~\ref{correlations}, following the proof of Theorem \ref{main_decay}. Previous results cover different regions of the $(\beta,h)$-plane. In the high-temperature regime, when $\beta<\beta_c$ and $h=0$, the Gibbs measure is unique, and hence the truncated and non-truncated two-point functions coincide. Newman and Spohn \cite{newman_spohn_shiba_relation} proved, for the Ising model, that the two-point function is asymptotically proportional to $J_{x_1x_2}$ (see also Aizenman and Fernández in \cite{Aizenman1988}). This result was extended by Aoun \cite{Aoun2021} to a class of long-range random-cluster and Potts models, where a sharp asymptotic formula, including the proportionality constant, is obtained.

For nonzero external field, Iagolnitzer and Souillard \cite{Iagolnitzer1977} conjectured that the same asymptotic behavior should hold whenever $h\neq0$, although they only proved it when $h$ is sufficiently large. For $h>0$, they obtained the upper bound
\[
\langle\sigma_{x_1};\sigma_{x_2}\rangle_{\beta,h} \leq C(\beta,h,\varepsilon)|x_1-x_2|^{2d-\alpha+\varepsilon},
\]
for every $\varepsilon>0$.

A related low-temperature result for infinite-range Ising models with exponentially decaying interactions was obtained by Aoun and Khettabi \cite{Aoun2023}. They consider interactions of the form
\[
J_x=\psi(x)e^{-\rho(x)},
\]
where $\rho$ is a norm and $\psi$ is a subexponential correction. Under suitable assumptions, they prove that, for $d\geq2$ and sufficiently large $\beta$, the truncated two-point function along a fixed direction is comparable to the interaction. Their proof relies on the random-current representation and on a comparison with a killed random walk. The exponential-moment assumption used in their work is not satisfied by the power-law interactions considered here.

Our work complements these results by treating the low-temperature regime for multidimensional Ising models with interactions $J_{xy}=J|x-y|^{-\alpha}$, $\alpha>d$. We prove an upper bound of the same order as the interaction by means of a convergent cluster expansion based on multiscale contours. Combined with the lower bound of \cite{Iagolnitzer1977}, this yields the algebraic comparability of the truncated two-point function with $J_{x_1x_2}$.

We also mention similar results to those of Iagoniltzer and Soulliard, proved by Klein and Masooman in the case of random $J_{xy}$ and $h_x$ \cite{Klein.Masooman.97}. Our work complements this picture for the deterministic case by covering the region where $\beta$ is much larger than $\beta_c$, being, to the best of our knowledge, the first result of this type for multidimensional Ising models with algebraically decaying interactions in the low-temperature regime. (see Figure \ref{fig:decay}). 

\begin{figure}[H]
    \centering
    \input{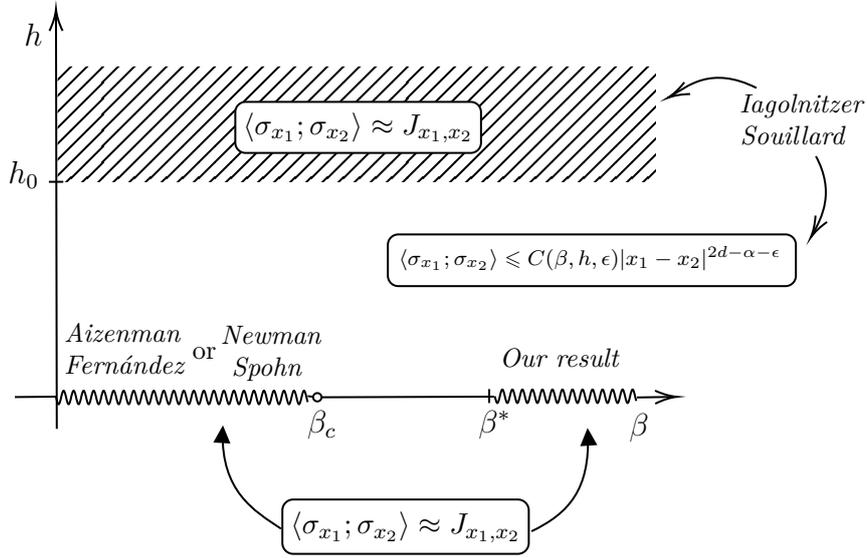}
    \caption{In the figure above, $h_0$ and $\beta^*$ are constants much larger than $0$ and $\beta_c$, respectively. This is due to the fact that the results depend on the convergence of the cluster expansion.}
    \label{fig:decay}
\end{figure}

In this paper, we are going to show that it is possible to develop a convergent cluster expansion in all the region $\alpha > d \geq 2$ at low enough temperatures, making use of a refined version of the contours in \cite{Affonso2024} as defined in \cite{Johanes}. As an application, we study the decay of the two-point correlation function, where we show the algebraic decay of the correlations at low temperatures. The paper is organized as follows. In Section 2, we review the definitions of the contours and develop the cluster expansion, proving the relevant estimates and the convergence of the cluster expansion at low enough temperatures. In Section 3, we show how to control the truncated correlation functions using the results in the previous section.

\section{Cluster Expansion}\label{cluster}

\subsection{The Model and Contours}

The configuration space is given by $\Omega\coloneqq \{-1,+1\}^{\Z^d}$. We write $\Lambda\Subset \Z^d$ to denote a finite subset of $\Z^d$. Fixed such $\Lambda$, the set of \textit{local configurations} is $\Omega_\Lambda\coloneqq \{-1,1\}^\Lambda$. Moreover, given ${\eta\in\Omega}$, the set of \textit{local configurations with $\eta$ boundary condition} is ${\Omega_\Lambda^\eta\coloneqq \{\omega\in\Omega : \omega_x=\eta_x, \text{ }\forall x\in\Lambda^c\}}$.  The \textit{local Hamiltonian of the ferromagnetic long-range Ising model} in $\Lambda\Subset\Z^d$ with $\eta$-boundary condition is a function $H_{\Lambda,\mathbf{h}}^{\eta}:\Omega_\Lambda^\eta \to \mathbb{R}$, given by

\begin{equation}\label{hamiltonian_spin}
    H_{\Lambda,\mathbf{h}}^{\eta}(\sigma)\coloneqq -\sum_{x,y\in\Lambda} J_{xy}\sigma_x\sigma_y - \sum_{x\in \Lambda, y\in\Lambda^c} J_{xy}\sigma_x\eta_y-\sum_{x\in \Lambda}h_x \sigma_x,
\end{equation}
where $\mathbf{h}=\{h_x\}_{x\in \Lambda}$ is a collection of real numbers (external field), and the interaction $\{J_{xy}\}_{x,y\in\Z^d}$ is defined as

\begin{equation}\label{Long-Range Interaction}
J_{xy} \coloneqq \begin{dcases}
\frac{J}{|x-y|^\alpha} &\text{ if }x\neq y,\\[0.2cm]
0 &\text{otherwise},
\end{dcases}
\end{equation}
with $J >0$, $\alpha>d$. Let $\mathscr{F}$ be the $\sigma$-algebra on $\Omega$ generated by the cylinder sets. The \textit{finite volume Gibbs measures} are probability measures on $(\Omega, \mathscr{F})$ under which the integration of bounded measurable functions $f:\Omega\rightarrow \mathbb{R}$ is given by 
    \begin{equation}
        \langle f \rangle_{\Lambda, \beta, \mathbf{h}}^\eta \coloneqq \frac{1}{Z_{\Lambda, \beta,\mathbf{h}}^{\eta}}\sum_{\sigma \in \Omega_\Lambda^\eta}f(\sigma)e^{-\beta H_{\Lambda,\mathbf{h}}^{\eta}(\sigma)}.
    \end{equation}
Here, $\beta>0$ is the inverse temperature and $Z_{\Lambda,\beta,\mathbf{h}}^{\eta}$ is the \textit{partition function}, defined as 

\begin{equation}\label{partitionfunction_spin}
    Z_{\Lambda, \beta,\mathbf{h}}^{\eta}\coloneqq \sum_{\sigma\in\Omega_\Lambda^\eta} e^{-\beta H_{\Lambda,\mathbf{h}}^{\eta}(\sigma)}.
\end{equation}
When $\mathbf{h} \equiv 0$, we will omit the subscript $\mathbf{h}$ and write $H_{\Lambda}^{\eta}$, $Z_{\Lambda, \beta}^{\eta}$ and $\langle\cdot\rangle_{\Lambda,\beta}^\eta$. Two particularly important boundary conditions are the ground state configurations $\sigma_{+} \equiv +1$ and $\sigma_{-} \equiv -1$, and are called, respectively, $+$ and $-$ boundary conditions. It is well known that the measures $\langle\cdot\rangle_{\Lambda,\beta,\mathbf{h}}^+$ and $\langle\cdot\rangle_{\Lambda,\beta,\mathbf{h}}^-$ converge in the weak* topology when $\Lambda$ invades $\Z^d$ to $\langle\cdot\rangle_{\beta,\mathbf{h}}^+$ and $\langle\cdot\rangle_{\beta,\mathbf{h}}^-$, respectively. In this paper, we will only consider the $+$-boundary condition.

One of the most prolific ideas for studying the low-temperature regime of lattice spin systems with discrete state space, since the development of Peierls' argument, is to map each spin configuration to a configuration of geometric objects called contours. To do this, we use the definition of incorrect points. Unless stated otherwise, the distances will be given by the $\ell_1$-norm. 

\begin{definition}\label{def:incorrect_points}
     Given a configuration $\sigma$, a point $x \in \Z^d$ is $+$\textit{-correct} \textit{(resp.} $-$\textit{-correct)} for $\sigma$ if $\sigma_y = +1$ \textit{(resp.} \textit{$\sigma_y = -1$)} for every $y \in B_1(x)$, where $B_1(x)$ is the unit ball centered at $x \in \Z^d$. A point is called \textit{incorrect} for $\sigma$ if it is neither $+$-correct nor $-$-correct. The \textit{boundary} of a configuration $\sigma$ is the set $\partial \sigma$ of all incorrect points for $\sigma$.
\end{definition}

Since we are working with $+$-boundary conditions, the boundary of every configuration of interest is finite. The contours correspond to some partition of the boundary. In Pirogov-Sinai theory, for example, the contours are the connected components of the boundary paired with a label. However, this construction is not suitable for general long-range systems, as it restricts the rate at which interactions can decay; see \cite{Park.88.I, Park.88.II}. Building on the ideas present in Fr\"{o}hlich-Spencer \cite{Frohlich.Spencer.82}, Affonso, Bissacot, Endo, and Satoshi \cite{affonso2024phasetransitionferromagneticqstate} introduced multiscale contours that are more appropriate for long-range Ising models, allowing decay in the whole region $\alpha > d$ of exponents. A notable feature of these contours is that they are potentially disconnected sets. The following definition, due to \cite{Johanes}, is a refined version of the referred multiscale contours, whose main advantage is that it is simpler, despite still allowing the contours to be disconnected.

\begin{definition}\label{Ma}
    Let $M>1$ and $a > d$. For each $A\Subset\Z^d$, a set $\Gamma(A) \coloneqq \{\overline{\gamma} : \overline{\gamma} \subset A\}$ is called an $(M,a)$-\emph{partition} when the following two conditions are satisfied.
	\begin{enumerate}[label=\textbf{(\Alph*)}, series=l_after] 
		\item They form a partition of $A$, i.e.,  $\bigcupdot_{\overline{\gamma} \in \Gamma(A)}\overline{\gamma}=A$.
		
		\item For all $\overline{\gamma}, \overline{\gamma}^\prime \in \Gamma(A)$, 
			\be\label{B_distance_2}
			\dis(\overline{\gamma},\overline{\gamma}') > M\min\left \{|V(\overline{\gamma})|,|V(\overline{\gamma}')|\right\}^\frac{a}{d+1},
 			\ee
	\end{enumerate}
 where $V(\Lambda)$ denotes the \textit{volume} of $\Lambda \Subset \mathbb{Z}^d$, and is given by $V(\Lambda) \coloneqq \Z^d \setminus \Lambda^{(0)}$ with $\Lambda^{(0)}$ being the unique unbounded connected component of $\Lambda^c$. Notice that, although the term ``volume'' is often associated with a numerical quantity, $V(\Lambda)$ denotes a set. Its cardinality is denoted by $|V(\Lambda)|$.
\end{definition}

In this paper, we will use $a \coloneqq a(\alpha,d) = \frac{3(d+1)}{(\alpha-d)\wedge 1}$. The constant $M$ will be appropriately chosen later. Even after fixing the parameters $M$ and $a$, there may still be multiple partitions of a set that are $(M, a)$-partitions. However, there is always a \textit{finest} $(M,a)$-partition, and we pick this one in the definition of the map $A \mapsto \Gamma(A)$ (see \cite{Johanes} for details). 

The contours will be defined by means of the $(M, a)$-partition, and the following two propositions guarantee desired properties for them. Property \textbf{(A1)} below will be crucial once we define external contours (see Proposition \ref{external}). It implies in particular that $\overline{\gamma}^\prime$ is contained in the unbounded component of $\overline{\gamma}^c$ if and only if $V(\overline{\gamma})\cap V(\overline{\gamma}^\prime) = \emptyset$. Furthermore, it will be important for the development of the cluster expansion later on that the compatibility between contours can be checked pairwise. This is the subject of Proposition \ref{Comp}.

\begin{proposition}
    The finest $(M,a)$-partition of any $A\Subset \Z^d$ satisfies the following property:

\begin{itemize}
    \item[\textbf{\emph{(A1)}}] For any $\overline{\gamma},\overline{\gamma}^\prime\in \Gamma(A)$, $\overline{\gamma}'$ is contained in only one connected component of $(\overline{\gamma})^c$.
\end{itemize}
\end{proposition}
\begin{proof}
     We will show that, if $\Gamma$ satisfies \eqref{B_distance_2} and there is a pair $(\overline{\gamma}', \overline{\gamma}'')$ such that $\overline{\gamma}''$ is contained in more than one connected component of $(\overline{\gamma}')^c$, then we can break $\overline{\gamma}''$ into two pieces, yielding a partition that still satisfies \eqref{B_distance_2}, but is strictly finer than the original. The conclusion follows from the fact that $\Gamma$ is the finest $(M,a)$-partition satisfying \eqref{B_distance_2}.

    Let $A$ be one of the bounded connected components of $(\overline{\gamma}')^c$ such that $A \cap \overline{\gamma}'' := \overline{\gamma}_1 \neq \emptyset$ and take $\overline{\gamma}_2 := \overline{\gamma}''\backslash \overline{\gamma}_1$. Then $\Gamma' := (\Gamma\backslash \{\overline{\gamma}''\} ) \cup \{\overline{\gamma}_1, \overline{\gamma}_2\}$ is clearly a finer partition than $\Gamma$ of the same set. We will show that $\Gamma'$ satisfies \eqref{B_distance_2}. 

    Let $\overline{\gamma} \in \Gamma\backslash \{\overline{\gamma}''\}$ be any. Then, for $i \in \{1, 2\}$, we have
    \begin{equation*}
        d(\overline{\gamma}, \overline{\gamma}_i) \geq d(\overline{\gamma}, \overline{\gamma}'') > M\min\left \{|V(\overline{\gamma})|,|V(\overline{\gamma}'')|\right\}^{\frac{a}{d+1}} > M\min\left \{|V(\overline{\gamma})|,|V(\overline{\gamma}_i)|\right\}^{\frac{a}{d+1}}.
    \end{equation*}

    Thus, we only need to check the condition for the pair $(\overline{\gamma}_1, \overline{\gamma}_2)$.  Notice that $d(\overline{\gamma}_1, \overline{\gamma}_2) > d(\overline{\gamma_1}, \overline{\gamma}')$. Since $d(\overline{\gamma}_1, \overline{\gamma}') \geq d(\overline{\gamma}', \overline{\gamma}'')$, we have
    \begin{equation*}
        d(\overline{\gamma}_1, \overline{\gamma}_2) > d(\overline{\gamma}', \overline{\gamma}'') > M\min\left \{|V(\overline{\gamma}')|,|V(\overline{\gamma}'')|\right\}^{\frac{a}{d+1}} >  M\min\left \{|V(\overline{\gamma}_1)|,|V(\overline{\gamma}_2)|\right\}^{\frac{a}{d+1}}.
    \end{equation*}
    The last inequality follows from the fact that $V(\overline{\gamma}_i)\subseteq V(\overline{\gamma}'')$, because $\overline{\gamma}_i\subseteq\overline{\gamma}''$ and the fact that $V(\overline{\gamma}_1)\subseteq V(\overline{\gamma}')$, because $\overline{\gamma}_1$ is contained in a bounded connected component of $(\overline{\gamma}')^c$.
\end{proof}

\begin{proposition}\label{Comp}
    Let $\overline{\gamma}_1,\dots,\overline{\gamma}_n$ be a family of subsets such that $\{\overline{\gamma}_i, \overline{\gamma}_j\}$ is the finest $(M, a)$-partition of $\overline{\gamma_i} \cup \overline{\gamma_j}$, for every $i \neq j$. Then the family $\Gamma=\{\overline{\gamma}_1,\dots,\overline{\gamma}_n\}$ is the finest $(M, a)$-partition of $\bigcup_{i = 1}^n \overline{\gamma_i}$.
\end{proposition}
\begin{proof}
Since Condition \textbf{(B)} depends only on pairs of elements, we just need to check that the set $\Gamma$ is the finest partition. Furthermore, we claim that $\Gamma$ is the finest $(M,a)$-partition if and only if any $\Gamma'\subsetneq \Gamma$ is also the finest $(M,a)$-partition of its elements. The necessity condition is straightforward: if it were not the case, one could replace some $\Gamma' \subsetneq \Gamma$ with a strictly finer partition, $\Gamma'_2$, and, clearly, $(\Gamma \backslash \Gamma') \cup \Gamma'_2$ would be strictly finer than $\Gamma$.

Now suppose that every proper subset $\Gamma' \subsetneq \Gamma$ is the finest one but $\Gamma$ is not. Let then $\Gamma''$ be the finest one. This means that, for each $\overline{\gamma} \in \Gamma$ and $\overline{\gamma}'' \in \Gamma''$, there are only two possibilities: either $\overline{\gamma} \cap \overline{\gamma}'' = \varnothing$ or $\overline{\gamma}'' \subseteq \overline{\gamma}$. Since $\Gamma \neq \Gamma''$, there must exist $\overline{\gamma}''_1 \in \Gamma''$ such that $\overline{\gamma}''_1 \notin \Gamma$. In this case, the remaining possibilities are $\overline{\gamma}''_1 \cap \overline{\gamma}= \varnothing$ and $\overline{\gamma}''_1 \subsetneq \overline{\gamma}$ for each $\overline{\gamma} \in \Gamma$. Now, take any point $x \in \overline{\gamma}''_1$ and let $\overline{\gamma}_\ast$ be the element of $\Gamma$ such that $x \in \overline{\gamma}_\ast$. From the above possibilities, we must have $\overline{\gamma}''_1 \subsetneq \overline{\gamma}_\ast$. Since $\Gamma''$ is the finest $(M, a)$-partition, $\overline{\gamma}''_2 := \overline{\gamma}_\ast \backslash \overline{\gamma}''_1$ must be decomposable in elements of $\Gamma''$ as well, so $\{\overline{\gamma}_\ast\} \subset \Gamma$ is not the finest one, contradicting our hypothesis.

Finally, if each pair $\{\overline{\gamma}_i, \overline{\gamma}_j\}$ is the finest $(M, a)$-partition, we can proceed by induction using the fact just proven to show that $\Gamma$ is the finest $(M, a)$-partition.
\end{proof}

\begin{definition}[\textbf{Contours}]\label{def:d_contours}
A pair $\gamma \coloneqq (\overline{\gamma},\omega)$, with $\overline{\gamma} \Subset \mathbb{Z}^d$ and $\omega\in\Omega_{\overline{\gamma}}$, is called a \textit{contour} if there is some configuration $\sigma \in \Omega$ such that $\overline{\gamma} \in \Gamma(\partial \sigma)$ and $\omega = \sigma\vert_{\overline{\gamma}}$, that is, it is the restriction of $\sigma$ to $\overline{\gamma}$. The \emph{support of the contour} $\gamma$ is defined as $\Sp(\gamma)\coloneqq \overline{\gamma}$, and its \emph{size} is given by $|\gamma| \coloneqq |\Sp(\gamma)|$.
\end{definition}

With this definition, every configuration $\sigma \in \Omega_\Lambda^{\pm}$ is naturally associated to the family of contours $\Gamma(\sigma) \coloneqq \{\gamma_1, \dots, \gamma_n\}$.

\begin{remark}
    The present definition of contours is slightly different from the previous works \cite{Affonso2024, Johanes}. There, the contours only have information about the support and the labels, while here we choose to put the whole configuration inside the contour to simplify some calculations. The consequence is that the map mentioned above is \emph{injective}.
\end{remark}

It is noteworthy the fact that the correspondence $\Omega_\Lambda^{\pm} \ni \sigma \mapsto \Gamma(\sigma)$ (with codomain being the set of all family of contours) is not surjective, since there can exist some family of contours that are not generated by any configuration. 

\begin{definition}
    We say that a family of contours $\Gamma$ is \textit{compatible} if there exists some configuration $\sigma$ such that $\Gamma = \Gamma(\sigma)$. We say that $\gamma$ and $\gamma'$ are compatible if $\Gamma = \{\gamma, \gamma'\}$ is compatible.
\end{definition}

For every finite $\Lambda\Subset\mathbb{Z}^d$, let $\mathcal{C}_\Lambda^{+}$ denote the collection of compatible families of contours associated with configurations in $\Omega_\Lambda^{+}$. By the definition of compatibility and the injectivity stated above, the map
\[
    \Omega_\Lambda^{+} \ni \sigma \longmapsto \Gamma(\sigma)\in\mathcal{C}_\Lambda^{+}
\]
is a bijection.

Given a subset $\Lambda \Subset \Z^d$ we define its \textit{interior} as $\I(\Lambda) \coloneqq V(\Lambda) \setminus \Lambda$. For the special case of a contour $\gamma$, we write $\I(\gamma)$ and $V(\gamma)$ instead of $\I(\Sp(\gamma))$ and $V(\Sp(\gamma))$. Also, denoting by $\I(\gamma)^{(k)}$ the connected components of $\I(\gamma)$, we can define the \emph{label} map ${\lab_{\overline{\gamma}}: \{\Sp(\gamma)^{(0)}, \I(\gamma)^{(1)},\dots, \I(\gamma)^{(n)}\} \rightarrow \{-1,+1\}}$ by taking the label of $\Sp(\gamma)^{(0)}$ as the sign of $\sigma$ in $\inB V(\gamma)$ and the label of $\I(\gamma)^{(k)}$ as the sign of $\sigma$ in $\partial_{\text{ex}} V(\I(\gamma)^{(k)})$. Notice that there can be connected components of a contour sitting inside its own interior. However, the labels are well-defined, since the sign of $\sigma$ is constant on the boundaries of $\Sp(\gamma)$. The following sets will be useful 
	\[
	\I_\pm(\gamma) \coloneqq \hspace{-1cm}\bigcup_{\substack{k \geq 1, \\ \lab_{\Sp(\gamma)}(\I(\gamma)^{(k)})=\pm 1}}\hspace{-1cm}\I(\Sp(\gamma))^{(k)} , \;\;\;
	\I(\gamma) = \I_+(\gamma) \cup \I_-(\gamma). \;\;\;
	\]

 For a family of contours $\Gamma$, we define $V(\Gamma) \coloneqq \bigcup_{\gamma \in \Gamma} V(\gamma)$. The sets $\I(\Gamma)$, \ $\I_-(\Gamma)$ and $\I_+(\Gamma)$ are defined analogously, by means of the union. 

One of the major steps in order to get a convergent cluster expansion is to define a suitable notion of \emph{external contour}. In the previous works \cite{Affonso2024, Johanes} the definition was a direct extension of the usual notion from Pirogov-Sinai theory --- a contour was said to be external if it's not in the interior of any other contour $\gamma'$. In our case, we will have to replace the interior by the minus interior $\I_-(\gamma')$. Thus, we introduce the \emph{modified volume}, $\widetilde{V}(\gamma) \coloneqq \Sp(\gamma) \cup \I_-(\gamma)$. As before, we define $\widetilde{V}(\Gamma)\coloneqq \cup_{\gamma\in\Gamma}\widetilde{V}(\gamma)$. We will use $\gamma \cup \Gamma$ instead of $\{\gamma\}\cup\Gamma$ in order to lighten the notation.

\begin{definition}[\textit{External and Internal Contours}]
    A contour $\ctr$ is \textit{external} with respect to a family $\ctrb$ if $\Sp(\ctr) \cap \widetilde{V}(\ctr') = \emptyset$ for every $\ctr' \in \ctrb\backslash \{\ctr\}$. We will denote by $\Gamma^e$ the family of all external contours from a given family of contours $\Gamma$. We define $\mathcal{E}^+_\Lambda$ as the collection of all compatible families $\Gamma$ of external contours in $\Lambda$ such that $V(\Gamma) \subset \Lambda$. When $\Lambda =\Z^d$, we write $\mathcal{E}_{\Z^d}^+ \coloneqq \mathcal{E}^+$.
    Moreover, we say that a family of contours $\Gamma$ is \textit{internal} to $\gamma$ if $\gamma \cup \Gamma$ is a compatible family of contours with $\gamma$ being the only external contour. We define $\mathscr{I}(\gamma)$ as the collection of all families of contours internal to $\gamma$. 
\end{definition}

\begin{remark}
    One of the properties of this definition is that if $\ctr$ has its support inside the plus interior $\I_+(\gamma')$ of an external contour $\gamma'$, then $\gamma$ is itself external. 
\end{remark}

\begin{figure}[h!]
    \centering
    \tikzset{every picture/.style={line width=0.75pt}} 

\begin{tikzpicture}[x=0.75pt,y=0.75pt,yscale=-1,xscale=1]

\draw  [fill={rgb, 255:red, 155; green, 155; blue, 155 }  ,fill opacity=1 ][line width=1.5]  (44,153.5) .. controls (44,78.11) and (139.66,17) .. (257.67,17) .. controls (375.67,17) and (471.33,78.11) .. (471.33,153.5) .. controls (471.33,228.89) and (375.67,290) .. (257.67,290) .. controls (139.66,290) and (44,228.89) .. (44,153.5) -- cycle ;
\draw  [fill={rgb, 255:red, 255; green, 255; blue, 255 }  ,fill opacity=1 ][line width=1.5]  (66.33,154) .. controls (66.33,89.93) and (152.07,38) .. (257.83,38) .. controls (363.6,38) and (449.33,89.93) .. (449.33,154) .. controls (449.33,218.07) and (363.6,270) .. (257.83,270) .. controls (152.07,270) and (66.33,218.07) .. (66.33,154) -- cycle ;
\draw  [fill={rgb, 255:red, 155; green, 155; blue, 155 }  ,fill opacity=1 ][line width=1.5]  (114.52,154) .. controls (114.52,113.41) and (178.68,80.5) .. (257.83,80.5) .. controls (336.98,80.5) and (401.15,113.41) .. (401.15,154) .. controls (401.15,194.59) and (336.98,227.5) .. (257.83,227.5) .. controls (178.68,227.5) and (114.52,194.59) .. (114.52,154) -- cycle ;
\draw  [color={rgb, 255:red, 0; green, 0; blue, 0 }  ,draw opacity=1 ][fill={rgb, 255:red, 255; green, 255; blue, 255 }  ,fill opacity=1 ][line width=1.5]  (165,122) .. controls (185,112) and (256.33,94) .. (255,122) .. controls (253.67,150) and (226.33,152) .. (255,182) .. controls (283.67,212) and (218.67,210) .. (165,182) .. controls (111.33,154) and (145,132) .. (165,122) -- cycle ;
\draw  [color={rgb, 255:red, 0; green, 0; blue, 0 }  ,draw opacity=1 ][fill={rgb, 255:red, 255; green, 255; blue, 255 }  ,fill opacity=1 ][line width=1.5]  (286.33,104) .. controls (297.67,100) and (304.33,92) .. (326.33,104) .. controls (348.33,116) and (338.33,100) .. (370.33,143) .. controls (402.33,186) and (302.33,225) .. (277,182) .. controls (251.67,139) and (275,108) .. (286.33,104) -- cycle ;
\draw  [fill={rgb, 255:red, 155; green, 155; blue, 155 }  ,fill opacity=1 ][line width=1.5]  (185,139.5) .. controls (185,132.04) and (191.04,126) .. (198.5,126) .. controls (205.96,126) and (212,132.04) .. (212,139.5) .. controls (212,146.96) and (205.96,153) .. (198.5,153) .. controls (191.04,153) and (185,146.96) .. (185,139.5) -- cycle ;
\draw  [fill={rgb, 255:red, 155; green, 155; blue, 155 }  ,fill opacity=1 ][line width=1.5]  (295,132.5) .. controls (295,125.04) and (301.04,119) .. (308.5,119) .. controls (315.96,119) and (322,125.04) .. (322,132.5) .. controls (322,139.96) and (315.96,146) .. (308.5,146) .. controls (301.04,146) and (295,139.96) .. (295,132.5) -- cycle ;

\draw (16,147.4) node [anchor=north west][inner sep=0.75pt]  [font=\large]  {$\gamma $};
\draw (91,140.4) node [anchor=north west][inner sep=0.75pt]  [font=\large]  {$\gamma '$};
\draw (326,118.4) node [anchor=north west][inner sep=0.75pt]  [font=\large]  {$\gamma ''$};
\draw (215,112.4) node [anchor=north west][inner sep=0.75pt]  [font=\large]  {$\gamma '''$};
\draw (132,77.4) node [anchor=north west][inner sep=0.75pt]  [font=\large]  {$+$};
\draw (169,61.4) node [anchor=north west][inner sep=0.75pt]  [font=\large]  {$+$};
\draw (272,48.4) node [anchor=north west][inner sep=0.75pt]  [font=\large]  {$+$};
\draw (321,56.4) node [anchor=north west][inner sep=0.75pt]  [font=\large]  {$+$};
\draw (362,71.4) node [anchor=north west][inner sep=0.75pt]  [font=\large]  {$+$};
\draw (397,100.4) node [anchor=north west][inner sep=0.75pt]  [font=\large]  {$+$};
\draw (412,149.4) node [anchor=north west][inner sep=0.75pt]  [font=\large]  {$+$};
\draw (392,191.4) node [anchor=north west][inner sep=0.75pt]  [font=\large]  {$+$};
\draw (355,216.4) node [anchor=north west][inner sep=0.75pt]  [font=\large]  {$+$};
\draw (313,230.4) node [anchor=north west][inner sep=0.75pt]  [font=\large]  {$+$};
\draw (267,236.4) node [anchor=north west][inner sep=0.75pt]  [font=\large]  {$+$};
\draw (219,236.4) node [anchor=north west][inner sep=0.75pt]  [font=\large]  {$+$};
\draw (177,228.4) node [anchor=north west][inner sep=0.75pt]  [font=\large]  {$+$};
\draw (140,213.4) node [anchor=north west][inner sep=0.75pt]  [font=\large]  {$+$};
\draw (102,180.4) node [anchor=north west][inner sep=0.75pt]  [font=\large]  {$+$};
\draw (99,103.4) node [anchor=north west][inner sep=0.75pt]  [font=\large]  {$+$};
\draw (212,50.4) node [anchor=north west][inner sep=0.75pt]  [font=\large]  {$I_{+}( \gamma )$};
\draw (167,125.4) node [anchor=north west][inner sep=0.75pt]  [font=\large]  {$+$};
\draw (224,131.4) node [anchor=north west][inner sep=0.75pt]  [font=\large]  {$+$};
\draw (162,152.4) node [anchor=north west][inner sep=0.75pt]  [font=\large]  {$+$};
\draw (233,175.4) node [anchor=north west][inner sep=0.75pt]  [font=\large]  {$+$};
\draw (182,169.4) node [anchor=north west][inner sep=0.75pt]  [font=\large]  {$I_{+}( \gamma ')$};
\draw (349,153.4) node [anchor=north west][inner sep=0.75pt]  [font=\large]  {$-$};
\draw (280,106.4) node [anchor=north west][inner sep=0.75pt]  [font=\large]  {$-$};
\draw (277,138.4) node [anchor=north west][inner sep=0.75pt]  [font=\large]  {$-$};
\draw (347,134.4) node [anchor=north west][inner sep=0.75pt]  [font=\large]  {$-$};
\draw (326,144.4) node [anchor=north west][inner sep=0.75pt]  [font=\large]  {$-$};
\draw (281,163.4) node [anchor=north west][inner sep=0.75pt]  [font=\large]  {$-$};
\draw (307,99.4) node [anchor=north west][inner sep=0.75pt]  [font=\large]  {$-$};
\draw (300,165.4) node [anchor=north west][inner sep=0.75pt]  [font=\large]  {$I_{-}( \gamma ')$};
\draw (144,137.4) node [anchor=north west][inner sep=0.75pt]  [font=\large]  {$+$};

\end{tikzpicture}
    \caption{The contours $\gamma$, $\gamma'$ and $\gamma'''$ are external. The contour $\gamma''$ is not external.}
    \label{fig:gammas_and_interior_plus_and_minus}
\end{figure}

 Notice that $\mathscr{I}(\gamma)$ depends only on $\gamma$, not on the other contours that can possibly be near $\gamma$. Also notice that the definition above is different from the ones in \cite{Affonso2024, Johanes}, but they are equivalent modulo the fact that we used the modified volume.

\begin{proposition}\label{external}
    Let $\ctrb$ be a family of compatible contours. For any $\ctr \in \ctrb \backslash \ctrb^e$ there exists a unique $\ctr' \in \ctrb^e$ such that $V(\ctr) \subset \I_-(\ctr')$.
\end{proposition}
\begin{proof}
By the definition of external contour, if $\ctr \in \ctrb \backslash \ctrb^e$ then there exists $\gamma'$ such that $\Sp(\gamma)\cap \widetilde{V}(\gamma')\neq \emptyset$. By Condition \textbf{(A1)}, $\Sp(\gamma)$ is a subset of one, and only one, connected component of $(\Sp(\gamma'))^c$. Since it has a nonempty intersection with the modified volume, $\Sp(\gamma)$ cannot be contained in $\widetilde{V}(\gamma')^c$; thus, it must be in $\I_-(\gamma')$.

Now, it is clear that $\Sp(\ctr) \subset \I_-(\ctr')$ implies that $V(\ctr) \subset \I_-(\ctr')$. Indeed, by Condition \textbf{(A1)}, we know that $\Sp(\ctr)$ is in some connected component $\I_-(\ctr')^{(k)}$ of the minus interior. Given any connected component $\overline{\gamma}^{(i)}$ of $\Sp(\ctr)$, if we had $V(\overline{\gamma}^{(i)}) \not\subset \I_-(\ctr')^{(k)}$, then we would be able to find some $x \in \I(\overline{\gamma}^{(i)})\backslash \I_-(\ctr')^{(k)}$. Taking also any $y \in \overline{\ctr}^{(i)}$, since $V(\overline{\ctr}^{(i)})$ is a connected set having both $x$ and $y$, we could take some path $\lambda$ in $V(\overline{\ctr}^{(i)})$ connecting these points and, hence, connecting the interior of $\gamma'$ with $V(\ctr')^c$. This path would, then, intercept $\Sp(\ctr')$ and, since $\lambda$ is in $V(\overline{\ctr}^{(i)})$, we would necessarily have some point of $\Sp(\ctr')$ in $\I(\overline{\gamma}^{(i)})$. Using again condition \textbf{(A1)}, we would get to the conclusion that $\Sp(\ctr') \subset \I(\overline{\gamma}^{(i)})$. If that were the case, we would have $V(\ctr') \subset V(\I(\overline{\gamma}^{(i)}))$, but since $\partial_{\text{ext}}V(\I(\overline{\gamma}^{(i)})) \subset \overline{\gamma}^{(i)}$, we would have

\[
\partial_{\text{ext}}V(\I(\overline{\gamma}^{(i)})) \subset \Sp(\gamma) \subset \I(\gamma')\subset V(\gamma')\subset V(\I(\overline{\gamma}^{(i)})), 
\]
which is absurd.

We concluded that, if $\ctr$ is not external, then $V(\ctr) \subset V(\ctr')$ for some other contour $\ctr'$. If $\ctr'$ is external, we are done. If it's not, we can iterate the procedure to find another $\ctr''$ such that $V(\ctr') \subset V(\ctr'')$ --- notice that $\ctr''$ cannot be $\ctr$. Since the number of contours is finite, we will eventually reach some external one by this procedure.

As for the uniqueness, suppose that $V(\ctr) \subset \I_-(\ctr')$ for an external $\ctr'$. Now, given any other external contour $\ctr''$, by condition \textbf{(A1)}, we know that $\Sp(\ctr'')$ is in some connected component of $\I_+(\ctr')$ or $(V(\ctr'))^c$, so it is impossible to have $V(\ctr) \subset \I_-(\ctr'')$.
\end{proof}

\begin{remark}
    Proposition \ref{external} implies that for any compatible family of contours $\Gamma$ and $\Gamma^{\text{e}}=\{\gamma_1, ..., \gamma_n\}$, there exists a unique partition of $\Gamma \backslash \Gamma^{\text{e}}$ into families $\Gamma_1, ..., \Gamma_n$ such that $\Gamma_i \in \mathscr{I}(\gamma_i)$ for each $i$.
\end{remark}

\begin{proposition}\label{leminha}
    Let $\ctrb$ be a compatible family of contours and $\ctrb^e$ the associated family of external contours. Then $\sigma_x = 1$ for all $x \in \widetilde{V}(\ctrb^e)^c$.
\end{proposition}
\begin{proof}
Each configuration defines a partition of the lattice $\Z^d$ with respect to the points being incorrect or $\pm$-correct. Then, let $\Gamma$ be a compatible family of contours, and $\sigma$ be the configuration such that $\Gamma(\sigma)=\Gamma$. Let $\Theta_x:\Omega\rightarrow \mathbb{R}$ be the function such that 
\[
\Theta_x(\sigma) = \begin{cases}
    +1, & \text{if } x \text{ is $+$-correct} \\
    -1, & \text{if } x \text{ is $-$-correct} \\
    \hspace{0.3cm} 0, & \text{if } x \text{ is incorrect}.
\end{cases}
\]
Then,
\[
\widetilde{V}(\ctrb^e)^c =\bigcup_{a=-1,0,+1}\{x\in \widetilde{V}(\ctrb^e)^c: \Theta_x(\sigma)=a\}.
\]
If the point $x$ is incorrect, then it must be in the support of some contour. If $x$ is $-$-correct, since we are in the $+$ boundary condition, it must be surrounded by incorrect points. In both cases, $x \in \widetilde{V}(\ctr)$ for some contour. If $\ctr$ is external, we are done. Otherwise, Proposition \ref{external} implies that $x\in \I_-(\gamma)$ for some $\gamma\in \Gamma^e$. 
\end{proof}

We finish this section by stating an important proposition that gives us an upper bound on the entropy of the contours introduced in this section. Let the set of all contours of size $n$ containing a given point $x$ be
\[
\mathcal{C}_x(n) \coloneqq \{\gamma \in \mathcal{E}_\Lambda^+: x \in V(\gamma), |\gamma|=n\}.
\]
\begin{proposition}\label{Bound_on_C_0_n}
	Let $d\ge 2$, $x\in \Z^d$ and $\Lambda\Subset \mathbb{Z}^d$. There exists $c_1\coloneqq c_1(d,M,\alpha)>0$ such that
	\begin{equation}\label{Eq: exp.bound.contours}
	|\mathcal{C}_x(n)| \leq e^{c_1 n}, \quad \forall\, n\geq 1.
	\end{equation}
\end{proposition}
For a proof of this proposition, see \cite[Corollary 3.28]{Johanes}.

\subsection{Hamiltonian via Contours}\label{hamiltonian}

In order to develop the cluster expansion, we first define the normalized Hamiltonian by the ground state energy as $$H^+(\Gamma) \coloneqq H^+_\Lambda(\sigma) - H^+_\Lambda(\sigma_+),$$ where $\Gamma = \Gamma(\sigma)$. We will rewrite it in terms of the contours introduced earlier. Since the interaction is long-range, the energy $H^+(\Gamma)$ does not decouple into the sum of energies $H^+(\gamma)$ for $\gamma\in\Gamma$; we must also consider the interactions between pairs of contours, giving rise to two terms. The first term will depend only on the contribution of one external contour $\gamma$ and its internal contours $\Gamma_\gamma \coloneqq \{\ctr' \in \Gamma; \ctr' \text{ is internal to } \ctr\}$,
\begin{equation}\label{contour_phi1}
    \Phi_1(\ctr \cup \Gamma_\gamma) \coloneqq 2\sum_{\{x,y\} \subset \widetilde{V}(\ctr)} J_{xy}\mathds{1}_{\{\sigma_x \neq \sigma_y\}} + 2 \sum_{\substack{x \in \widetilde{V}(\ctr) \\ y \in \widetilde{V}(\ctr)^c}}J_{xy}\mathds{1}_{\{\sigma_x \neq 1\}},
\end{equation}
while the other one encompasses the interaction between pairs of them
\begin{equation}\label{contour_phi2}
    \Phi_2(\ctr\cup \Gamma_\gamma,\ctr'\cup \Gamma_{\gamma'}) \coloneqq -4\sum_{\substack{x \in \widetilde{V}(\ctr) \\ y \in \widetilde{V}(\ctr')}}J_{xy}\mathds{1}_{\{\sigma_x = \sigma_y = -1\}}.
\end{equation}
Since $\Phi_2$ is negative, the interaction between external contours is attractive, which makes the proof of the convergence of the cluster expansion trickier than in the short-range case. We put the contour representation in the next proposition.
\begin{proposition}\label{prop_hamil_contour}
For any compatible family of contours $\Gamma$, let $\Gamma^e$ denote the family of all external contours with respect to $\Gamma$. Then,
\begin{equation}\label{hamil_contour}
    H^+(\Gamma) = \sum_{\gamma \in \Gamma^e} \Phi_1(\gamma \cup \Gamma_\gamma) + \sum_{\gamma\neq\gamma'\in \Gamma^e} \Phi_2(\ctr \cup \Gamma_\gamma,\ctr'\cup \Gamma_{\gamma'}),
\end{equation}
where $\Phi_1$ and $\Phi_2$ are defined respectively by equations \eqref{contour_phi1} and \eqref{contour_phi2}.
\end{proposition}
Although cumbersome, the proof of formula \eqref{hamil_contour} follows by straightforward computations, and we choose to omit it. We call attention to the fact that, had we defined the exterior contours in the usual way with $V(\ctr)$ instead of $\widetilde{V}(\ctr)$, the interaction energy between two external contours would have been smaller and hence more attractive. In our case, this weaker interaction energy plays an important role in the convergence of the cluster expansion. 

Let $\sigma$ be a configuration in $\Lambda \Subset \mathbb{Z}^d$ such that $\gamma \in \Gamma(\sigma)^e$, we define $\tau_{\gamma}(\sigma)$ as the configuration given by

\begin{equation}\label{tau}
\tau_{\gamma}(\sigma)_x =
\begin{cases}
\sigma_x & \text{if } x \in \I_+(\gamma) \cup V(\gamma)^c, \\
-\sigma_x & \text{if } x \in \I_-(\gamma), \\
+1 & \text{if } x \in \Sp(\gamma).
\end{cases}
\end{equation}

By the bijection between configurations in $\Omega_\Lambda^+$ and compatible families of contours, we also regard $\tau_\gamma$ as a map acting on compatible families. More precisely, let $\gamma\cup\Gamma$ be a compatible family of contours and let $\sigma$ be the unique configuration such that $\Gamma(\sigma)=\gamma\cup\Gamma$.
We define $\tau_\gamma(\gamma\cup\Gamma) \coloneqq \Gamma(\tau_\gamma(\sigma))$. Notice that, although the effect of $\tau_\gamma$ is to erase the contour $\gamma$ from the configuration, in general $\tau_\gamma(\gamma\cup\Gamma)\neq\Gamma$.
Indeed, $\tau_\gamma$ also flips the spins in $\I_-(\gamma)$ and may therefore alter the contours contained in this region.

Likewise, if $\Gamma$ is a compatible family of contours and $\sigma$ is the unique configuration such that $\Gamma(\sigma)=\Gamma$, we write $H^+(\Gamma) \coloneqq H_\Lambda^+(\sigma)-H_\Lambda^+(\sigma_+)$. With these conventions, we have $H^+(\tau_{\gamma}(\gamma \cup \Gamma)) = H_\Lambda^+(\tau_\gamma(\sigma)) - H_\Lambda^+(\sigma_+)$ for the energy one gets when erasing the contour $\gamma$ through the action of the map $\tau_\gamma$.

 In order to properly state the following proposition, we introduce for each $\Lambda \Subset \Z^d$ the \emph{surface energy term} as
\[
F_\Lambda  \coloneqq \sum_{\substack{x \in \Lambda \\ y \in \Lambda^c}}J_{xy}.
\]
The following proposition, proved in \cite{Johanes}, shows that the difference in energy when one erases a contour is positive and depends on quantities related to the erased contour.
\begin{proposition}\label{Prop: Cost_erasing_contour}
For $M$ large enough, there exists a constant $c_2(\alpha,d,J) \coloneqq c_2 > 0$, such that for  any $\Lambda \Subset \Z^d$, and $\gamma \cup \Gamma$ a family of contours such that $\gamma \in (\gamma \cup \Gamma)^e$, it holds
	\be
	H^+(\gamma \cup \Gamma)- H^+(\tau_{\gamma}(\gamma \cup \Gamma))\geq c_2||\gamma||,
	\ee
    where $||\gamma|| \coloneqq |\gamma|+F_{\I_-(\gamma)}+F_{\Sp(\gamma)}.$
\end{proposition}
\begin{remark}
    The constant $c_2$ can be taken as $\min\{1, J\}(2d+1)^{-1}2^{-\alpha-2}$, see the end of the proof of \cite[Proposition 2.10]{Johanes}.
\end{remark}

\subsection{Partition Function of a Polymer Gas}

\newcommand{\classe}{\mathscr{E}}

This section aims to rewrite $Z_{\Lambda, \beta}^+$ as the partition function of a polymer gas. To do so, we establish some notations related to graphs that we will extensively use. For a graph $G$, we denote its set of vertices as $V(G)$ and its edge set as $E(G)$. Given two graphs $G_1$ and $G_2$, $G_1$ is called a subgraph of $G_2$ if $V(G_1) \subset V(G_2)$ and $E(G_1) \subset E(G_2)$. Given a connected graph $G$ and $x, y \in V(G)$, we define $d(x, y)$ as the infimum of the number of edges of all paths joining $x$ and $y$. For a given set $X$, we define $\mathcal{G}_{X}$ as the set of all connected graphs $G$ with $V(G)=X$, $\mathcal{T}_{X}$ is the set of all trees $T$ with $V(T)=X$, and $\mathcal{T}^{x_0}_{X}$ is the set of all rooted trees $T$ with root $x_0$ and $V(T)=X$. We also write $\mathcal{T}^{0}_{n+1}$ when $X=\{0, 1,\dots,n\}$ and root equal $0$. For a tree, every vertex $v$ with $\deg(v)=1$ is called a \emph{leaf}, where $\deg(v)$ is the number of edges connected to $v$. For a rooted tree $T$, we say that a vertex $v$ is of \emph{generation} $k$ if the path in the tree connecting the root to $v$ has length $k$.

\begin{definition}\label{def_compatibility}
  A \emph{polymer} is a compatible family $\polymer$ of mutually external contours. Two polymers $\polymer$ and $\polymer'$ are \emph{compatible} when the following two conditions are satisfied:
    \begin{enumerate}[label=\textnormal{(\Roman*)}]
        \item For every $\gamma \in \polymer$ and $\gamma' \in \polymer'$, we have $\Sp(\gamma) \neq \Sp(\gamma')$ and $\{\Sp(\gamma), \Sp(\gamma')\}$ is a $(M, a)-$partition;
      \item exactly one of the following three conditions happens:
                \begin{enumerate}[label= (\roman*)]
                    \item $\widetilde{V}(\polymer) \cap \widetilde{V}(\polymer') = \emptyset$.
                    \item There is $\gamma \in \polymer$ such that $V(\polymer') \subset \I_-(\gamma)$.
                    \item There is $\gamma' \in \polymer'$ such that $V(\polymer) \subset \I_-(\gamma')$.
                \end{enumerate}
    \end{enumerate} 
\end{definition}

When two polymers are compatible, we write $\polymer \sim \polymer'$. When $\Gamma = \{\gamma\}$ and $\Gamma' = \{\gamma'\}$, we simply write $\gamma \sim \gamma'$. Notice that this boils down to checking the condition (I). Also, every polymer $\Gamma$ is incompatible with itself, in particular, $\gamma \not \sim \gamma$. Moreover, the set of all polymers in $\Lambda$ is precisely $\mathcal{E}^+_{\Lambda}$. For the next proposition, let us introduce the partition function \eqref{partitionfunction_spin} with the energy normalized by the ground state energy $H^+_\Lambda(\sigma_+)$,
 \[
 \Ztil^{+}_{\Lambda, \beta} \coloneqq e^{\beta H_\Lambda^+(\sigma_+)}Z^+_{\Lambda, \beta}.
\]

The next proposition is the multidimensional version of Theorem 4.11 of \cite{Cassandro.Merola.Picco.Rozikov.14}, we give a detailed proof for the sake of completeness.
\begin{proposition}\label{gas_polymer_partition}
Given any $\Lambda \Subset \Z^d$, 
\begin{equation}
    \Ztil^{+}_{\Lambda, \beta} = 1 + \sum_{ \emptyset\neq \polymerb \subset \mathcal{E}^+_{\Lambda}}\prod_{\polymer \in \polymerb}z_\beta^+(\polymer) \prod_{\left\{\polymer, \polymer'\right\}\subset X} \mathbbm{1}_{\polymer \sim \polymer'}.
\end{equation}
The quantity $\Ztil^{+}_{\Lambda, \beta}$ can be seen as the partition function of a gas of polymers with activity
\begin{equation}
    z_\beta^+(\polymer) = K(\polymer)\prod_{\ctr \in \polymer} W(\ctr),
\end{equation}
where the functions $W(\gamma)$ and $K(\Gamma)$ are defined by Equations \eqref{weight} and \eqref{functionA} respectively.
\end{proposition}
\begin{proof}
Notice that, thanks to Proposition \ref{external}, we can write all family of contours $\Gamma$ as an union $\cup_{\gamma \in \Gamma^e}\gamma \cup \Gamma_\gamma$, where $\Gamma_\gamma \in \mathscr{I}(\gamma)$. After further splitting the sum over contours as a sum over external and internal contours, one gets

\begin{equation*}
    \Ztil^+_{\Lambda, \beta} = 1 + \sum_{\substack{\emptyset \neq \extb \in \es^+_{\Lambda}}} \Ztil^+(\extb),
\end{equation*}

where $\widetilde{Z}^+(\Gamma^e)$, called the \emph{crystallic partition function} is given by

\begin{equation*}
    \Ztil^+(\extb) \coloneqq \sum_{\substack{\Gamma_\gamma \in \is(\gamma) \\ \gamma \in \Gamma^e}} e^{-\beta H^+(\Gamma)}.
\end{equation*}

 In particular, we write $\Ztil(\gamma)$ instead of $\Ztil(\{\gamma\})$ when $\Gamma=\{\gamma\}$. Then, using the spin-flip map introduced in \eqref{tau}, we define

\begin{equation*}
    \Ztil^+(\check{\ctr}) \coloneqq 1 + \sum_{\emptyset \neq \ctrb \in \is(\ctr)} e^{-\beta H^+(\tau_{\gamma}(\gamma \cup \Gamma))},
\end{equation*}
where the mark $\check{\gamma}$ indicates that the contour $\gamma$ is not included in the Hamiltonian. We get
\begin{align*}\label{preztil}
    \Ztil^+_{\Lambda} &= \sum_{\substack{\extb \in \es^+_{\Lambda}}} \frac{\Ztil^+(\extb)}{\prod_{\gamma\in \Gamma^e} \Ztil^+(\ctr)} \prod_{\gamma\in \Gamma^e} W(\ctr)  \Ztil^+(\check{\ctr}),
\end{align*}
where
\begin{equation}\label{weight}
    W(\ctr) \coloneqq \frac{\Ztil^+(\ctr)}{\Ztil^+(\check{\ctr})}.
\end{equation}
The next step is to expand $\Ztil^+(\extb)$ in terms of the two-body interactions between contours. We can write, by means of \eqref{hamil_contour} and using the Mayer trick,
\begin{equation*}
\begin{split}
\Ztil^+(\extb) = \sum_{P \in \mathcal{P}(\Gamma^e)} \prod_{Y \in P} \sum_{G \in \mathcal{G}_{Y}}\sum_{\substack{\ctrb_\gamma \in \is(\gamma) \\ \gamma \in Y}} \left(\prod_{\gamma \in Y} e^{-\beta \Phi_1(\gamma \cup \ctrb_\gamma)}\right) \varphi_G,
\end{split}
\end{equation*}
where $\mathcal{P}(\Gamma^e)$ is the set of all partitions of $\Gamma^e$ and $\varphi_G$ are the Ursell functions
\begin{equation*}
    \varphi_G =  \prod_{\{\gamma,\gamma'\} \in E(G)} \varphi_{\gamma,\gamma'} \quad \text{ with } \quad \varphi_{\gamma,\gamma'} = e^{-\beta \Phi_2(\gamma \cup \ctrb_\gamma, \gamma' \cup \ctrb_{\gamma'}) } - 1.
\end{equation*}
 In the case where $|Y|=1$, we set $\varphi_G=1$. Notice that the terms $\left(\prod_{\gamma \in Y} e^{-\beta \Phi_1(\gamma \cup \ctrb_\gamma)}\right) \varphi_G$ Depend only on those $\Gamma_\gamma$ such that $\gamma \in Y$. Analogously, for each partition $P \in \mathcal{P}(\Gamma^e)$

\begin{equation*}
    \prod_{\gamma \in \Gamma^e} \Ztil^+(\ctr)= \prod_{Y \in P} \sum_{\substack{\ctrb_\gamma \in \is(\gamma) \\ \gamma \in Y}} \prod_{\gamma \in Y}  e^{-\beta \Phi_1(\ctr \cup \ctrb_\gamma)}.
\end{equation*}
Hence, the ratio between the partition function with interacting contours and non-interacting contours, together with an exchange in summation, becomes
\begin{equation*}
     \frac{\Ztil^+(\extb)}{\prod_{\gamma \in \Gamma^e} \Ztil^+(\ctr)} = \sum_{P \in \mathcal{P}(\Gamma^e)} \prod_{Y\in P} K(Y),
\end{equation*}
with
\begin{equation}\label{functionA}
   K(Y) \coloneqq \sum_{G \in \mathcal{G}_{Y}} \frac{\sum_{\substack{\ctrb_\gamma \in \is(\gamma) \\ \gamma \in Y}}  \left(\prod_{\gamma \in Y} e^{-\beta \Phi_1(\gamma\cup \ctrb_\gamma)}\right)  \varphi_{G}}{\sum_{\substack{\ctrb_\gamma \in \is(\gamma) \\ \gamma \in Y}} \prod_{\gamma \in Y}  e^{-\beta \Phi_1(\ctr \cup \ctrb_\gamma)}}.
\end{equation}

In particular, 
\begin{equation}\label{zgammae}
    \Ztil^+(\extb) = \sum_{P \in \mathcal{P}(\Gamma^e)} \prod_{Y\in P}  \left[ K(Y)\prod_{\gamma \in Y}  W(\ctr)\Ztil^+(\check{\ctr})\right]
\end{equation}

so we can write the total partition function as
\begin{align*}
    \Ztil^+_{\Lambda} = 1 + \sum_{\substack{\emptyset \neq \extb \in \es^+_{\Lambda}}} \sum_{P \in \mathcal{P}(\Gamma^e)} \prod_{Y\in P}  \left[ K(Y)\prod_{\gamma \in Y}  W(\ctr)\Ztil^+(\check{\ctr})\right].
\end{align*}
Instead of summing over partitions of a family of external contours, we can sum over collections of families of external contours satisfying certain conditions. Denoting each such family by $\Gamma$ and the collections $X$, we get:
\begin{equation}\label{ztil}
\begin{split}
      \Ztil^+_{\Lambda} &= 1 + \sum_{\emptyset \neq \polymerb \subset \mathcal{E}_\Lambda^+} \mathbbm{1}_{\left\{\Gamma_X \in \mathcal{E}_\Lambda^+;\; \bigcap_{\Gamma \in X} \Gamma = \emptyset\right\}}\left[\prod_{\polymer \in \polymerb}K(\Gamma)\prod_{\ctr \in \polymer}  W(\ctr)\right]\prod_{\ctr \in \polymer_X} \Ztil^+(\check{\ctr}),
      \end{split}
\end{equation}
where $\Gamma_X = \bigcup_{\Gamma \in X}\Gamma$. Summarizing, the sum is over all collections of disjoint families of external contours such that the union of the families of this collection is also external and compatible. 

The previous expression gives rise to a recursion: we have the partition function in terms of $\Ztil^{\bc}(\check{\ctr})$, which is a partition function of the interior of the contour $\ctr$. To see that, given a family of contours $\Gamma = \{(\overline{\gamma}_i, \sigma_{\overline{\gamma}_i}), i = 1, \dots n\}$, define $\tau(\Gamma)$ as the family of contours with every spin flipped, $\tau(\Gamma) =  \{(\overline{\gamma}_i, -\sigma_{\overline{\gamma}_i}), i = 1, \dots n\}$. Also, denoting by $\mathscr{C}^\pm(\Lambda)$ the set of all families of compatible contour in $\Lambda$ with $\pm$-boundary condition, we can write

\begin{equation*}
    \Ztil^{\bc}(\check{\ctr}) = 1 +  \sum_{\emptyset \neq \Gamma \in \mathscr{C}^-(\I_-(\ctr))} \mathbbm{1}_{\{\gamma \cup \Gamma \text{ is compatible}\}} e^{-\beta H^+(\tau(\Gamma))}.
\end{equation*}

The contours are only in the minus components of the interior due to our definition of the external contour. Besides, we needed to put one more restriction because the contours in $\Ztil^{\bc}(\check{\ctr})$ are not allowed to get too close to $\gamma$ (otherwise they would have to be the same contour). Performing a change of variables, we have

\begin{equation*}
    \Ztil^{\bc}(\check{\ctr}) = 1 + \sum_{\emptyset \neq \Gamma \in \mathscr{C}^+(\I_-(\ctr))} \mathbbm{1}_{\{\gamma \cup \tau(\Gamma) \text{ is compatible}\}} e^{-\beta H^+(\Gamma)}.
\end{equation*}

Thus, $\Ztil^{\bc}(\check{\ctr})$ is equal to $ \Ztil^{\bc}_{\I_-(\ctr)}$ apart from a characteristic function. We will take advantage of this similarity to expand $\Ztil^{\bc}(\check{\ctr})$ analogously and insert into $ \Ztil^{\bc}_{\Lambda}$. 

In order to state the recursion in a clearer manner, we will need some definitions. Given a family of contours $\Gamma$, we define the level $\ell$ of a contour $\gamma \in \Gamma$ inductively as follows. If $\gamma$ is external, we put $\ell(\gamma) = 1$. Otherwise, $\ell(\gamma) = n + 1$ if $n$ is the highest level found in the set of contours $\gamma'$ such that $\gamma \in \is(\gamma')$. Finally, put $\ell(X) = \max \{\ell(\gamma); \gamma \in \Gamma_X\}$. We claim that, after $n$ steps, the partition function will be written as 

\begin{equation}\label{recur}
      \Ztil^{\bc}_{\Lambda} = 1 + \sum_{\substack{\emptyset \neq \polymerb \subset \mathcal{E}^+_{\Lambda}\\ \ell(X) \leq n}} \left(\prod_{\{\Gamma, \Gamma'\}} \mathbbm{1}_{\Gamma \sim \Gamma'} \right) \prod_{\polymer \in \polymerb} \left[ K^{\bc}(\polymer)\prod_{\ctr \in \polymer}  W^{\bc}(\ctr)\mathcal{Z}_{X, n}(\gamma)\right],
\end{equation}

where

\begin{equation*}
\mathcal{Z}_{X, n}(\gamma) =
\begin{cases}
    1, \quad \ell(\gamma) < n \\
     \Ztil^{\bc}(\check{\ctr}), \quad \ell(\gamma) = n.
\end{cases}
\end{equation*}

Here, the dependence on $X$ comes through the level of $\gamma$ in the family
$\Gamma_X$.

Notice that Equation \ref{recur} proves the desired formula. Indeed, taking $\ell(\Lambda) = \displaystyle \max_{X \subset \mathcal{E}^+_{\Lambda}} \ell(X)$ for each finite $\Lambda$, after $\ell(\Lambda) + 1$ steps, two things will happen: the constraint over the level of $X$ in the first summation will always be satisfied, so it may be dropped, and $\mathcal{Z}_{X, n}(\gamma)$ will always be $1$.

Thus, we only need to prove \eqref{recur}. It is trivial for $n = 1$, having \eqref{ztil} in mind. Now, suppose that the expression holds for some $n$. In order to obtain the expression for $n + 1$, we will need to expand $\Ztil^{\bc}(\check{\ctr})$, which will be present in the polymers $\Gamma$ in level $n$. We have:

\begin{equation*}
    \Ztil^{\bc}(\check{\ctr}) = 1 + \sum_{\substack{\emptyset \neq \extb \in \es^{\bc}_{\I_-(\gamma)}\\ \extb = \{\gamma_1, ..., \gamma_n\}}} \sum_{\substack{\ctrb_i \in \is(\gamma_i) \\ 1 \leq i \leq n}} \mathbbm{1}_{\{\gamma \cup \tau(\gamma_1 \cup \ctrb_1 \cup ... \cup \gamma_n \cup \ctrb_n)) \text{ is compatible}\}} e^{-\beta H^+(\gamma_1 \cup \ctrb_1 \cup \dots \cup \gamma_n \cup \ctrb_n))}
\end{equation*}

Now, we can use that $\tau\left(\bigcup_i \gamma_i \right) = \bigcup_i \tau(\gamma_i)$ and the following consequence of Proposition \ref{Comp}: given a family of contours $\Gamma \subset \I_-(\gamma)$, $\gamma \cup \Gamma$ is compatible if, and only if $\Gamma$ is compatible and the family $\gamma \cup \Gamma^e$ is compatible. In this case, we can write

\begin{equation*}
    \Ztil^{\bc}(\check{\ctr}) = 1 + \sum_{\emptyset \neq \Gamma^e \in \mathcal{E}_{\I_-(\gamma)}^+}\mathbbm{1}_{\{\gamma \cup \tau(\Gamma^e) \text{ compatible}\}}\widetilde{Z}^+(\Gamma^e)
\end{equation*}

and expand $\widetilde{Z}^+(\Gamma^e)$ as \eqref{zgammae}, yielding

\begin{equation*}
    \Ztil^{\bc}(\check{\ctr}) = \sum_{\Gamma^e \in \mathcal{E}_{\I_-(\gamma)}^+}\mathbbm{1}_{\{\gamma \cup \tau(\Gamma^e) \text{ compatible}\}} \sum_{\mathscr{P} \text{ partition of } \Gamma^e} \prod_{P \in \mathscr{P}} \left[ K^{\bc}(P)\prod_{\gamma \in P}  W^{\bc}(\gamma)\Ztil^{\bc}(\check{\ctr})\right]
\end{equation*}

Again, we will replace the summation over partitions by a summation over collections of polymers:

\begin{equation*}
    \sum_{\Gamma^e \in \mathcal{E}_{\I_-(\gamma)}^+}\mathbbm{1}_{\{\gamma \cup \tau(\Gamma^e) \text{ compatible}\}} \sum_{\mathscr{P} \text{ partition of } \Gamma^e} = \sum_{X \in \classe(\gamma)},
\end{equation*}

where $\classe(\gamma)$ is the set consisting of all families $X \subset \mathcal{E}^+_{\I_-(\gamma)}$ such that $\bigcap_{\Gamma \in X} \Gamma = \emptyset$, and $\tau(\Gamma_X)$ is compatible with $\gamma$. Notice that, if $X$ is a family of pairwise compatible polymers, by the definition of compatibility, the level is constant across each polymer; that is, for any $\Gamma \in X$, $\ell(\gamma) = \ell(\gamma'), \forall \gamma, \gamma' \in \Gamma$. Let $\polymerb_n$ be the subset of $X$ containing polymers $\Gamma$ with contours in the $n-$th level. Putting a superscript $(n+1)$ to remind ourselves of which polymers appear in the next iteration and the subscript $\gamma$ to remind that $X_{\gamma}$ is inside the interior of $\gamma$, we get

\begin{align*}
    &\prod_{\polymer \in \polymerb_n} \left[ K^{\bc}(\polymer)\prod_{\ctr \in \polymer}  W^{\bc}(\ctr)\mathcal{Z}_{X, n}(\gamma)\right]\\
      &=  \prod_{\polymer \in \polymerb_n} \left\{ K^{\bc}(\polymer)\prod_{\ctr \in \polymer}  W^{\bc}(\ctr) \sum_{\rctr \in \classe(\ctr)}  \prod_{\polymer_{(n+1)} \in \rctr} \left[ K^{\bc}(\polymer_{(n+1)})\prod_{\ctr_{(n+1)} \in \polymer_{(n+1)}}  W^{\bc}(\ctr_{\scriptscriptstyle (n+1)})\Ztil^{\bc}(\check{\ctr}_{\scriptscriptstyle (n+1)})\right]\right\}\\
      &=   \sum_{(\rctr)_{\ctr \in \polymer_{\polymerb_n}}}\prod_{\polymer \in \polymerb_n} K^{\bc}(\polymer) \prod_{\ctr \in \polymer} \Biggl\{  \mathbbm{1}_{\left\{\substack{\rctr \in \classe(\ctr)}\right\}} W^{\bc}(\ctr) \times \\
      & \hspace{5.5cm}\times \left. \prod_{\polymer_{(n+1)} \in \rctr} \left[ K^{\bc}(\polymer_{(n+1)})\prod_{\ctr_{(n+1)} \in \polymer_{(n+1)}}  W^{\bc}(\ctr_{\scriptscriptstyle (n+1)})\Ztil^{\bc}(\check{\ctr}_{\scriptscriptstyle (n+1)})\right]\right\},\\
\end{align*}
where we exchanged the summation and the product twice and used the fact that summing over $((\rctr)_{\ctr \in \polymer})_{\polymer \in \polymerb_n}$ and $(\rctr)_{\ctr \in \polymer_{\polymerb_n}}$ is effectively the same.

The term inside the summation above can be written as

\begin{align*}
    &\prod_{\polymer \in \polymerb_n}  \left\{ K^{\bc}(\polymer)  \left(\prod_{\ctr \in \polymer} W^{\bc}(\ctr)\right) \left( \prod_{\ctr \in \polymer}\mathbbm{1}_{\left\{\substack{\rctr \in \classe(\ctr)}\right\}} \right) \times \right. \\
 & \hspace{5cm} \times \left. \left(\prod_{\polymer_{(n+1)} \in \bigcup_{\ctr}\rctr} \left[ K^{\bc}(\polymer_{(n+1)})\prod_{\ctr_{(n+1)} \in \polymer_{(n+1)}}  W^{\bc}(\ctr_{\scriptscriptstyle(n+1)})\Ztil^{\bc}(\check{\ctr}_{\scriptscriptstyle (n+1)})\right]\right)\right\}\\[0.4cm]
    &=  \left( \prod_{\polymer \in \polymerb_n} \left[ K^{\bc}(\polymer) \prod_{\ctr \in \polymer} W^{\bc}(\ctr)\right]\right) \left( \prod_{\ctr \in\bigcup_{\Gamma \in \polymerb_n }\Gamma}\mathbbm{1}_{\left\{\substack{\rctr \in \classe(\ctr)}\right\}} \right) \times \\[0.4cm]
    & \hspace{4cm} \times \left(\prod_{\polymer_{(n+1)} \in \bigcup_{\ctr \in \Gamma_{X_n}}\rctr} \left[ K^{\bc}(\polymer_{(n+1)})\prod_{\ctr_{(n+1)} \in \polymer_{(n+1)}}  W^{\bc}(\ctr_{\scriptscriptstyle(n+1)})\Ztil^{\bc}(\check{\ctr}_{\scriptscriptstyle (n+1)})\right]\right) \\
    &  =\left(\prod_{\polymer \in \polymerb_n \cup \bigcup_{\gamma \in \polymer_{\polymerb_n}}\rctr} \left[ K^{\bc}(\polymer)\prod_{\ctr \in \polymer}  W^{\bc}(\ctr)\mathcal{Z}_{X, n+1}(\ctr)\right]\right)\left( \prod_{\ctr \in \bigcup_{\Gamma \in \polymerb_n }\Gamma}\mathbbm{1}_{\left\{\substack{\rctr \in \classe(\ctr)}\right\}} \right),
\end{align*}

Let $\polymerb_{<n}$ be the complement $\polymerb\backslash\polymerb_n$. For any contour $\ctr$ in the polymers of $\polymerb_{<n}$, we have that $\mathcal{Z}_{X,n+1}(\ctr) = \mathcal{Z}_{X,n}(\ctr) = 1$, which allows us to write

\begin{align*}
    \left(\prod_{\polymer \in \polymerb_{<n}} \left[ K^{\bc}(\polymer)\prod_{\ctr \in \polymer}  W^{\bc}(\ctr)\mathcal{Z}_{X,n}(\ctr)\right]\right)&\left(\prod_{\polymer \in \polymerb_n \cup \bigcup_{\gamma \in \polymer_{\polymerb_n}}\rctr}\left[ K^{\bc}(\polymer)\prod_{\ctr \in \polymer}  W^{\bc}(\ctr)\mathcal{Z}_{X, n+1}(\ctr)\right]\right) \\
    & \hspace{1cm} =\left(\prod_{\polymer \in \polymerb \cup \bigcup_{\gamma \in \polymer_{\polymerb_n}}\rctr} \left[ K^{\bc}(\polymer)\prod_{\ctr \in \polymer}  W^{\bc}(\ctr)\mathcal{Z}_{X, n+1}(\ctr)\right]\right)
\end{align*}

Thus, back to the partition function, it can be written as

\begin{align*}
     \Ztil^{\bc}_{\Lambda} = 1 + \sum_{\substack{\emptyset \neq \polymerb \in \mathcal{E}^+_{\Lambda}\\ \ell(X) \leq n}}  \sum_{(\rctr)_{\ctr \in \polymer_{\polymerb_n}}} \left(\prod_{\{\Gamma, \Gamma'\} \subset X} \mathbbm{1}_{\Gamma \sim \Gamma'} \right) &\left( \prod_{\ctr \in \bigcup_{\Gamma \in \polymerb_n }\Gamma}\mathbbm{1}_{\left\{\substack{\rctr \in \classe(\ctr)}\right\}} \right) \times \\
     &\times \prod_{\polymer \in \polymerb \cup \bigcup_{\ctr \in \polymer_{\polymerb_n}}\rctr} \left[ K^{\bc}(\polymer)\prod_{\ctr \in \polymer}  W^{\bc}(\ctr)\mathcal{Z}_{X, n+1}(\ctr)\right]
\end{align*}

This is almost what we want; we just need to rearrange the first summations. For this, we will perform the following change of variables.

\begin{equation*}
    \left(\polymerb, (\rctr)_{\ctr \in \polymer_{\polymerb_n}}\right) \mapsto \polymerb' = \polymerb \cup \bigcup_{\ctr \in \polymer_{\polymerb_n}} \rctr.
\end{equation*}
\begin{equation*}
    \sum_{\substack{\polymerb \in \mathcal{E}^+_{\Lambda}\\ \ell(X) \leq n}} \sum_{(\rctr)_{\ctr \in \polymer_{\polymerb_n}}} \left(\prod_{\{\Gamma, \Gamma'\} \subset X} \mathbbm{1}_{\Gamma \sim \Gamma'} \right) \left( \prod_{\ctr \in \bigcup_{\Gamma \in \polymerb_n }\Gamma}\mathbbm{1}_{\left\{\substack{\rctr \in \classe(\ctr)}\right\}} \right) \rightarrow \sum_{\substack{\polymerb' \in \mathcal{E}^+_{\Lambda}\\ \ell(X') \leq n+1}} \left(\prod_{\{\Gamma, \Gamma'\} \subset X'} \mathbbm{1}_{\Gamma \sim \Gamma'} \right)
\end{equation*}

We will prove that there is a bijection between the sets over which the two summations are being performed.

$(\to)$ Take some pair $\left(\polymerb, (\rctr)_{\ctr \in \polymer_{\polymerb_n}}\right)$ satisfying the conditions of the left-hand side. Put $X'_{n+1} := \bigcup_{\ctr \in \polymer_{\polymerb_n}} \rctr$. Clearly, $\polymerb' = \polymerb \cup X'_{n+1} \in \mathcal{E}^+_{\Lambda}$. If there were some contour of level $n+2$, then it would necessarily be inside a contour of level $n+1$. Notice that the levels of the contours in $X$ are preserved in $X'$, so contours of level $n+1$ and $n+2$ must come from the polymers in $X'_{n+1}$. Now, since $\rctr \in \classe(\gamma)$ for every $\gamma \in \Gamma_{X_n}$, the union $\bigcup_{\Gamma \in X'_{n+1}} \Gamma$ is a family of mutually external contours. The conclusion is that $\ell(X') \leq n+1$.

Now we are going to check that the polymers in $X'$ are pairwise compatible. Since we already know that this is true for pairs in $X$, we may restrict to the case where one of them is in $X'_{n+1}$. 

(I) The first step is to prove the compatibility between the \emph{contours} constituting the polymers. To do so, we are going to use heavily the following fact, which is a kind of transitivity for $\sim$ under special conditions.

\begin{center}
\emph{Fact: If $\gamma_1 \sim \gamma_2$ and $\gamma_2 \sim \gamma_3$, then $\gamma_1 \sim \gamma_3$ provided that $\gamma_1$ and $\gamma_3$ are \\ in two different connected components of $\Sp(\gamma_2)^c$.}
\end{center}

By hypothesis, $\gamma_n \sim \gamma_{n+1}$ for any $\gamma_n \in \Gamma_{X_n}$ and $\gamma_{n+1} \in X^{(n+1)}_{\gamma_n}$. Since we already know that $\gamma_n$ is compatible with all the contours of level $n$ or smaller, the fact above tells us that  $\gamma_{n+1} \sim \gamma$ for any $\gamma_{n+1} \in \bigcup_{\Gamma \in X'_{n+1}} \Gamma$ and $\gamma$ such that $\ell(\gamma) \leq n$. Now we need to check the compatibility between contours both in the $(n+1)$-level. When both of them are in $\rctr$ for some $\gamma$, the compatibility follows by hypothesis. For the remaining case, let $\gamma_{n+1} \in X^{(n+1)}_{\ctr_n}$ and $\gamma'_{n+1} \in X^{(n+1)}_{\ctr'_n}$. By what was discussed, $\gamma_{n+1} \sim \ctr'_n$ and $\ctr'_n \sim \gamma'_{n+1}$. Since $\ctr_n$ and $\ctr'_n$ are in the same level, they are mutually external, so we can use the fact above to see that $\gamma_{n+1} \sim \gamma'_{n+1}$. 

(II) Now we need to check the second condition for the compatibility between polymers. Let $\Gamma_{n+1}$ be some polymer in $X'_{n+1}$. Clearly, there is exactly one $\gamma_n \in \Gamma_{X_n}$ such that $\Gamma_{n+1} \in X^{(n+1)}_{\gamma_n}$. Let $\Gamma_n$ be the polymer containing $\gamma_n$. Clearly, $\Gamma_n \sim \Gamma_{n+1}$, because $\Gamma_{n+1}$ is contained in only one contour of $\Gamma_n$. Now, let $\Gamma$ be any other polymer in $X$. We know that $\Gamma \sim \Gamma_n$ by hypothesis. Since $\Gamma \in X$, we know that $\Gamma$ cannot be inside any contour of $\Gamma_n$. Thus, either $\Gamma$ and $\Gamma_n$ are disjoint, or there is only one contour $\gamma$ of $\Gamma$ such that $\Gamma_n$ is inside $\Gamma$. In the first case, $\Gamma_{n+1}$ and $\Gamma$ are also disjoint. In the second case, $\Gamma_{n+1}$ is also contained in only one contour of $\Gamma$. Regardless, the compatibility is verified. Now, take another $\Gamma'_{n+1} \in X'_{n+1}$. We need to see that $\Gamma_{n+1} \sim \Gamma'_{n+1}$. It both are in $X^{(n+1)}_{\gamma_n}$ for some $\gamma_n$, then $\Gamma_{n+1}$ and $\Gamma'_{n+1}$ are disjoint by the hypothesis that $X^{(n+1)}_{\gamma_n} \in \classe(\gamma_n)$. If they are in different contours, they must also be disjoint because the two contours are. \\[0.3cm]

$(\leftarrow)$ Conversely, take some $X'$ satisfying the conditions in the right-hand side. Put $X$ as the collection of polymers in $X'$ whose contours have level $n$ or less. If $X_n = \emptyset$, simply take $(X', \emptyset)$ as the desired pair. Otherwhise, for any $\gamma_n \in X_n$ let $X^{(n + 1)}_{\gamma_n}$ be the set of polymers in $X'$ whose contours are internal to $\gamma_n$. Notice that this is well-defined due to the compatibility condition: one cannot have a polymer composed of contours $\gamma_{n+1}, \gamma'_{n+1}$ where $\gamma_{n+1}$ is inside $\gamma_n$ and $\gamma'_{n+1}$ is inside $\gamma'_n$.

Clearly, the polymers in $X$ are pairwise compatible; thus, we only need to show that, for any $\gamma_{n} \in \bigcup_{\Gamma \in \Gamma_{X_n}} \Gamma$, we have $X^{(n + 1)}_{\gamma_n} \in \classe(\gamma_n)$. By construction, $X^{(n + 1)}_{\gamma_n} \subset \mathcal{E}^+_{\I_-(\gamma_n)}$. The collection $X^{(n + 1)}_{\gamma_n}$ is also disjoint. If we had $\gamma_{n+1} \in \Gamma_{n+1} \cap \Gamma'_{n+1}$, this two polymers would not be compatible. Finally, notice that each $\gamma_n \cup \gamma_{n+1}$ is compatible for any $\gamma_{n+1} \in X^{(n + 1)}_{\gamma_n}$. By Proposition \ref{Comp}, we conclude that the whole family $\gamma_n \cup X^{(n + 1)}_{\gamma_n}$ is compatible.

We finally end the induction by arriving at the following equality.

\begin{equation}
      \Ztil^{\bc}_{\Lambda} = \sum_{\substack{\polymerb \in \mathcal{E}^+_{\Lambda}\\ \ell(X) \leq n+1}} \left(\prod_{\{\Gamma, \Gamma'\}} \mathbbm{1}_{\Gamma \sim \Gamma'} \right) \prod_{\polymer \in \polymerb} \left[ K^{\bc}(\polymer)\prod_{\ctr \in \polymer}  W^{\bc}(\ctr)\mathcal{Z}_{X, n+1}(\gamma)\right].
\end{equation}
\end{proof}

\subsection{Activity Bounds}
In this section, we prove the necessary bounds for the activities $z^+_\beta(\polymer)$. First notice that, given two external contours $\gamma,\gamma'$ and families of contours $\Gamma, \Gamma'$ internal, respectively, to $\gamma$ and $\gamma'$ we have
\begin{equation}\label{bound_phi_2_2}
    -\Phi_2(\ctr \cup \ctrb, \ctr' \cup \ctrb') \leq 4\sum_{\substack{x \in \widetilde{V}(\ctr) \\ y \in \widetilde{V}(\ctr')}} J_{xy} \coloneqq 4F_{\ctr,\ctr'}.
\end{equation}

\begin{lemma}\label{lemma_K}
    There exists a constant $c_3\coloneqq c_3(\alpha,d,J,M)>0$ such that $\displaystyle \lim_{M \to \infty} c_3 = 0$, and given a contour $\gamma$, for any $\Gamma \in \mathcal{E}^+$ satisfying $\gamma \sim \Gamma$ and $\gamma \not\in \Gamma$ it holds
    \[
    4\sum_{\substack{\ctr' \in \Gamma}} F_{\ctr,\ctr'} \leq c_3F_{\widetilde{V}(\ctr)}.
    \]
\end{lemma}
\begin{proof}
      Given a contour $\gamma$ and a polymer $\Gamma\sim\gamma$, define the sets $\Upsilon_1 = \{\gamma'\in \Gamma: |V(\gamma')| \geq |V(\gamma)|\}$ and $\Upsilon_2=\{\gamma'\in \Gamma: |V(\gamma')| < |V(\gamma)|\}$. Hence, 
\begin{equation}\label{c_3_eq1}
    \sum_{\substack{\ctr' \in \Gamma}} F_{\ctr,\ctr'} = \sum_{\substack{x\in \widetilde{V}(\gamma) \\ y\in \widetilde{V}(\Upsilon_1)}} J_{xy} +  \sum_{\substack{x\in \widetilde{V}(\gamma) \\ y\in \widetilde{V}(\Upsilon_2)}} J_{xy}.
\end{equation}
For any  $\gamma' \in \Upsilon_1$, it holds that $\dis(\gamma,\gamma')> M |\widetilde{V}(\gamma)|^{\frac{a}{d+1}}$ by condition \textbf{(B)}, so we get 
\begin{equation}\label{Eq: Lemma_aux_2_1}
    \sum_{\substack{x\in \widetilde{V}(\gamma) \\ y\in \widetilde{V}(\Upsilon_1)}} J_{xy}\leq \sum_{\substack{x\in \widetilde{V}(\gamma) \\ |y-x| > R}} J_{xy} = |\widetilde{V}(\gamma)|\sum_{ |y|\geq R}J_{0y},
\end{equation}
with $R\coloneqq \lceil M |\widetilde{V}(\gamma)|^{\frac{a}{d+1}}\rceil$. Defining $s_d(n) \coloneqq |\{x\in \Z^d : |x|=n \}|$, it is known that $s_d(n)\leq 2^{2d - 1}e^{d-1}n^{d-1}$. Using an upper bound by an integral together with \eqref{Eq: Lemma_aux_2_1}, we can show that
\begin{align*}\label{Eq: Lemma_aux_2.Part1}
    \sum_{\substack{x\in \widetilde{V}(\gamma)\\ y \in \widetilde{V}(\Upsilon_1)}}J_{xy} \leq \frac{J2^{d-1 + \alpha }e^{d-1}}{(\alpha - d)M^{\alpha - d}}|\widetilde{V}(\gamma)|^{1+\frac{a}{d+1}(d-\alpha)}.
\end{align*}
To bound the remaining term in \eqref{c_3_eq1}, split $\Upsilon_2$ into layers $\Upsilon_{2,m} \coloneqq \{ \gamma' \in \Upsilon_2 : |V(\gamma')|=m\}$, for $1\leq m\leq |V(\gamma)|-1$. Given some $x\in \widetilde{V}(\gamma)$, we can bound 
\begin{equation}\label{aux_2_2}
\begin{split}
   \sum_{\substack{x\in \widetilde{V}(\gamma)\\ y \in \widetilde{V}(\Upsilon_{2,m})}}J_{xy}  \leq J m\sum_{\substack{x \in \widetilde{V}(\gamma) \\ \gamma' \in \Upsilon_{2,m}}} \frac{1}{\dis(x,\gamma')^\alpha}.
\end{split}
\end{equation}
  Define, for each $\gamma'\in \Upsilon_{2,m}$, the set $B_{\gamma'} = \{y\in\mathbb{Z}^d:\dis(y,\gamma')\leq M m^{\frac{a}{d+1}}/3\}$. Any pair 
 of distinct elements $\gamma,\gamma'\in \Upsilon_{2,m}$ satisfies $\dis(\gamma,\gamma')> M m^{\frac{a}{d+1}}$, implying that the sets $B_\gamma$ and $B_{\gamma'}$ are disjoint. Moreover, for each site $x \in \widetilde{V}(\gamma)$ it holds
\begin{equation}\label{aux_2_3}
\begin{split}
J\sum_{\gamma'\in \Upsilon_{2,m}}\frac{1}{\dis(x,\gamma')^{\alpha}}\leq \frac{3}{Mm^{\frac{a}{d+1}}}\sum_{\substack{y\in B_{\gamma'} \\\gamma'\in \Upsilon_{2,m}}}J_{xy}\leq \frac{3}{Mm^{\frac{a}{d+1}}}\sum_{y\in \widetilde{V}(\gamma)^c}J_{xy}. 
\end{split}
\end{equation}
Joining Inequalities  \eqref{aux_2_2} and \eqref{aux_2_3}, and summing over $m$ we get
\begin{equation*}
    \sum_{\substack{x\in \widetilde{V}(\gamma) \\ y \in \widetilde{V}(\Upsilon_2)}} J_{xy} \leq \frac{3\zeta(2)}{M}F_{\widetilde{V}(\gamma)}.
\end{equation*} 
By our choice of $a$, our statement follows by choosing $c_3$ to be
\[
c_3 = \frac{b^*}{M^{(\alpha-d)\wedge1}}\quad \text{where} \quad b^* = \max\Big\{\frac{2^{d+2 + \alpha }e^{d-1}}{(\alpha - d)}, 24\zeta(2)\Big\}.
\]
\end{proof}

Recall that $\mathcal{T}_{\Gamma}$ denotes the set of trees on the vertex set $\Gamma$, $\|\ctr\| \coloneqq |\ctr| + F_{\I_-(\ctr)} + F_{\Sp(\ctr)}$ and $c_2$ is the constant in Proposition \ref{Prop: Cost_erasing_contour}. In order to better analyze the activity $z^+_\beta(\Gamma)$, it will be convenient to bound it by a much simpler expression, given by:
\begin{equation}\label{tilde_activity}
    \widetilde{z}^+_{\beta}(\polymer) \coloneqq 
        \sum_{T \in \mathcal{T}_{\Gamma}}\prod_{\ctr\in\polymer} e^{-\beta \frac{c_2}{2}\|\ctr\|}\prod_{\{\ctr,\ctr'\} \in E(T)}F_{\ctr,\ctr'}.
\end{equation}
 Notice that when $\Gamma$ has only one element, then necessarily we have $E(T) = \emptyset$ above. In this case, we adopt the convention that the product is equal to $1$.

\begin{proposition}\label{important_prop}
    
For $M, \beta$ large enough, it holds, for every $\Gamma$, that

\begin{equation}\label{def_ztilde}
    z^+_\beta(\polymer) \leq \widetilde{z}^+_{\beta}(\polymer).
\end{equation}
\end{proposition}
\begin{proof}
Proposition \ref{Prop: Cost_erasing_contour} implies that $W(\ctr) \leq e^{-\beta c_2\|\gamma\|}$, and Inequality \eqref{bound_phi_2_2} yields 
\begin{equation*}
   K(\polymer) 
    \leq  \sum_{G \in \mathcal{G}_{\Gamma}}\prod_{\{\gamma,\gamma'\} \in E(G)} \left( e^{4\beta F_{\ctr,\ctr'}} - 1 \right). 
\end{equation*}
Since each $G$ is a connected graph, it has at least one spanning tree $T$. Then one can write
   \begin{equation*}
   \begin{split}
   \sum_{G \in \mathcal{G}_{\Gamma}}\prod_{\{\gamma,\gamma'\}\subset E(G)}\left( e^{4\beta F_{\ctr,\ctr'}} - 1 \right) &\leq \sum_{T \in \mathcal{T}_{\Gamma}}\prod_{\{\gamma,\gamma'\}\in E(T)}\left( e^{4\beta F_{\ctr,\ctr'}} - 1 \right)\prod_{\{\gamma,\gamma'\}\in E(K_\Gamma)\setminus E(T)} e^{4\beta F_{\ctr,\ctr'}},
   \end{split}
   \end{equation*}
   where $K_\Gamma$ is the complete graph with vertex set $\Gamma$. Using Lemma \ref{lemma_K} together with the fact $e^x - 1\leq x e^x$ yields
\begin{equation*}
    |z^+_\beta(\polymer)| \leq\prod_{\gamma \in \polymer} e^{-\beta(c_2|\ctr| + (c_2 - 2c_3)(F_{\I_-(\ctr)}+F_{\Sp(\ctr)}))}\sum_{T \in \mathcal{T}_{\Gamma}}\prod_{\{\gamma,\gamma'\}\in E(T)} 4\beta F_{\ctr,\ctr'}.
\end{equation*}
Taking $M$ large such that $4c_3\leq c_2$ and $\beta > 32/c_2^2$ gives us the desired result. 
\end{proof}

In what follows, we will introduce the following sets
\[
\mathcal{E}_{n,\gamma_0}^+ \coloneqq \{\Gamma\in \mathcal{E}^+: \Gamma=\{\gamma_0,\gamma_1,\dots,\gamma_n\}\}.
\]
 In particular, when $n=1$, we have $\Gamma = \{\gamma,\gamma_0\}$ for some $\gamma\neq \gamma_0$ external and compatible with $\gamma_0$.

\begin{lemma}\label{F_Vol}

There exists a constant $c_\beta(\alpha,d,J,M) \coloneqq c_\beta>0$ such that, for all sufficiently large $\beta$, one has
    \begin{equation*}\label{eq_1_F_Vol}
    \sum_{\ctr \in \mathcal{E}_{1,\gamma_0}^+} e^{-\beta \frac{c_2}{2} \|\ctr\|}F_{\ctr,\ctr_0} \leq c_\beta F_{\Sp(\gamma_0)},         
    \end{equation*}
for every fixed contour $\ctr_0$. Moreover, $\underset{\beta \rightarrow \infty}{\lim}c_\beta = 0$.
\end{lemma}

\begin{proof}
Given some $\gamma_0 \in \mathcal{E}_\Lambda^+$, it holds that $$F_{\gamma,\gamma_0}\leq \sum_{x\in \Sp(\gamma_0)}J\frac{|\widetilde{V}(\gamma)|}{\dis(x,\gamma)^\alpha}, $$
therefore, 
\[
\sum_{\substack{\ctr\in \mathcal{E}_{1,\gamma_0}^+}} e^{-\beta \frac{c_2}{2} \|\ctr\|}F_{\ctr,\ctr_0}\leq \sum_{x\in \Sp(\gamma_0)}\sum_{\gamma \in \mathcal{E}^+_{1,\gamma_0}}e^{-\beta \frac{c_2}{2}\|\gamma\|}\frac{J|\widetilde{V}(\gamma)|}{\dis(x,\gamma)^\alpha}.
\]
Using that $\|\gamma\|\geq |\gamma|$, it holds
\begin{equation}\label{F_Vol_eq_1}
\begin{split}
      \sum_{x\in \Sp(\gamma_0)}\sum_{\gamma \in \mathcal{E}^+_{1,\gamma_0}}e^{-\beta \frac{c_2}{2}\|\gamma\|}\frac{J|\widetilde{V}(\gamma)|}{\dis(x,\gamma)^\alpha} &\leq \sum_{n \geq 1}\sum_{x \in \Sp(\ctr_0)}\sum_{\substack{\ctr \in \mathcal{E}_{1,\gamma_0}^+  \\ |\ctr| = n}} e^{-\beta \frac{c_2}{2} \|\ctr\|} \frac{J|\widetilde{V}(\ctr)|}{\dis(x,\ctr)^\alpha}\\
     &= \sum_{n \geq 1}\sum_{x \in \Sp(\ctr_0)} \sum_{m \geq 1} e^{-\beta \frac{c_2}{2}n} \frac{J}{m^\alpha} \sum_{\substack{\ctr \in \mathcal{E}_{1,\gamma_0}^+ \\ |\ctr| = n  \\ d(x, \ctr) = m}}  |\widetilde{V}(\ctr)| .
\end{split}
\end{equation}
For every site $x \in \Sp(\gamma_0)$ and $m \geq 1$, we have
\begin{equation*}
   \{\ctr:|\gamma|=n, \dis(x, \ctr) = m \text{ and } \ctr \in \mathcal{E}_{1,\gamma_0}^+\}\subset \bigcup_{\substack{y \in  \Sp(\gamma_0)^c \\ |x-y|=m }}\mathcal{C}_y(n),
\end{equation*}
where $\mathcal{C}_y(n)$ is the same as in Proposition \ref{Bound_on_C_0_n}. Applying such proposition together with the fact that $|\widetilde{V}(\ctr)| \leq n^{\frac{d}{d-1}}$, which follows from the isoperimetric inequality, we have that 
\begin{equation}\label{F_Vol_eq_3}
\begin{split}
     \sum_{x \in \Sp(\ctr_0)} \sum_{m \geq 1} e^{-\beta \frac{c_2}{2}n} \frac{J}{m^\alpha} \sum_{\substack{\ctr\in \mathcal{E}_{1,\gamma_0}^+ \\ |\ctr| = n \\ d(x, \ctr) = m}}  |\widetilde{V}(\ctr)| &\leq e^{-\beta \frac{c_2}{2}n} n^{\frac{d}{d-1}} \sum_{x\in \Sp(\gamma_0)}\sum_{m\geq 1}\frac{J}{m^\alpha} \sum_{\substack{y \in \Sp(\ctr_0)^c \\ |x - y| = m}}|\mathcal{C}_y(n)| \\
     &=  e^{-\beta \frac{c_2}{2}n} n^{\frac{d}{d-1}} |\mathcal{C}_0(n))| \sum_{x\in \Sp(\gamma_0)}\sum_{m\geq 1}\frac{J}{m^\alpha} |\{y\in \widetilde{V}(\gamma_0)^c: |x-y|=m\}| \\
     &\leq e^{(c_1 -\beta \frac{c_2}{2} )n}n^{\frac{d}{d-1}}F_{\Sp(\gamma_0)}.
\end{split}
\end{equation}
Plugging Inequality \eqref{F_Vol_eq_3} into \eqref{F_Vol_eq_1}, we get
\[
\sum_{\ctr \in \mathcal{E}_{1,\gamma_0}^+} e^{-\beta \frac{c_2}{2} \|\ctr\|}F_{\ctr,\ctr_0} \leq \left(\sum_{n\geq 1}n^{\frac{d}{d-1}}e^{(c_1-\beta \frac{c_2}{2})n}\right)F_{\Sp(\gamma_0)}.
\]
We get the desired result by taking $\beta c_2\geq 4c_1+8$ and choosing $c_\beta = e^{-\beta \frac{c_2}{4}}(1-e^{-\beta \frac{c_2}{4}})^{-1}$.
\end{proof}

\begin{proposition}\label{prop_tree}
For $\beta$ as in the previous lemma, it holds that
\begin{equation}\label{eq1_tree}
    \sum_{\Gamma \in \mathcal{E}_{n,\gamma_0}^+}\widetilde{z}_\beta^+(\Gamma)\leq (4c_{\beta/2})^n e^{-\beta \frac{c_2}{4} \|\gamma_0\|},
\end{equation}
    for any contour $\gamma_0 \in \mathcal{E}^+$ and $n \geq 0$. 
\end{proposition}
\begin{proof}
 The case $n = 0$ is straightforward. For $n\geq 1$, fixed  a contour $\gamma_0$, the sum can be rewritten as
 \begin{equation}\label{eq2_tree}
 \sum_{\Gamma \in \mathcal{E}^+_{n,\gamma_0}} =\frac{1}{n!}\sum_{\substack{\gamma_k \\ 1\leq k \leq n}}\prod_{i<j}\mathbbm{1}_{\left\{\substack{\gamma_i\sim\gamma_j \\ \tilde{V}(\gamma_i)\cap \tilde{V}(\gamma_j)=\emptyset}\right\}} \quad \text{where} \quad \sum_{\substack{\gamma_k \\ 1\leq k\leq n}} = \sum_{\gamma_1}\ldots \sum_{\gamma_n}.
 \end{equation}
 Notice that, in the product above, $0 \leq i < j \leq n$. Then we can consider the sum on the right-hand side of \eqref{tilde_activity} as being over all trees rooted at $0$ with vertex set $\{0, 1,\dots,n\}$. For each such a tree, it is true that
\begin{equation}\label{eq3_tree}
\prod_{i<j}\mathbbm{1}_{\left\{\substack{\gamma_i\sim\gamma_j \\ \tilde{V}(\gamma_i)\cap \tilde{V}(\gamma_j)=\emptyset}\right\}}\prod_{\{i,j\}\in E(T)}F_{\ctr_i,\ctr_j} \leq \prod_{\{i,j\}\in E(T)}\mathbbm{1}_{\left\{\substack{\gamma_i\sim\gamma_j \\ \tilde{V}(\gamma_i)\cap \tilde{V}(\gamma_j)=\emptyset}\right\}}F_{\ctr_i,\ctr_j},
\end{equation}
hence
\begin{equation}\label{eq7_tree}
    \sum_{\Gamma \in \mathcal{E}^+_{n,\gamma_0}} \sum_{T \in \mathcal{T}_{\Gamma}} \prod_{\ctr \in \polymer} e^{-\beta \frac{c_2}{2}\|\ctr\|} \prod_{\{\gamma,\gamma'\} \in E(T)} F_{\ctr,\ctr'}\leq \frac{e^{-\beta\frac{c_2}{2}\|\gamma_0\|}}{n!}\sum_{T \in \mathcal{T}^0_{n+1}}\sum_{\substack{\gamma_k\\ 1\leq k \leq n}} \prod_{k=1}^n e^{-\beta \frac{c_2}{2}\|\gamma_k\|}\prod_{\{i,j\}\in E(T)}\widetilde{F}_{\ctr_i,\ctr_j},
\end{equation} 
where $\widetilde{F}_{\ctr_i,\ctr_j} = \mathbbm{1}_{\left\{\substack{\gamma_i\sim\gamma_j}\right\}}\mathbbm{1}_{\{\tilde{V}(\gamma_i)\cap \tilde{V}(\gamma_j)=\emptyset\}}F_{\ctr_i,\ctr_j}$. 

We will proceed to sum, for each tree $T$, the left-hand side of \eqref{eq1_tree} erasing each generation of the tree using Lemma \ref{F_Vol}. For a given rooted tree $T\in \mathcal{T}^0_{n+1}$, let us denote its vertices as pairs $(i,j)$, where $i$ is the generation of the vertex and $j$ enumerates the vertex in the same generation, from left to right (see Figure \ref{partition_tree}). Denote also by $\ell$ the number of generations of the tree and by $m_i$ the number of vertices in the $i$-th generation. With this labeling, we can write 
\begin{equation*}
    \prod_{k=1}^n e^{-\beta \frac{c_2}{2}\|\gamma_k\|}\prod_{\{i,j\}\in E(T)}\widetilde{F}_{\ctr_i,\ctr_j} = \prod_{i=1}^\ell \prod_{j=1}^{m_i}\left( e^{-\beta \frac{c_2}{2}\|\gamma_{i,j}\|}\widetilde{F}_{\gamma_{i-1,j'},\gamma_{i,j}}\right) ,
\end{equation*}
where, for each pair $(i, j)$, $j'$ is the unique value such that $(i-1, j')$ is connected to $(i, j)$, that is, $(i-1, j')$ is the parent of $(i, j)$. For each fixed tree, we start to sum the contours corresponding to vertices of the last generation, 
\begin{equation}\label{eq4_tree}
    \begin{split}
        &\sum_{\substack{\gamma_k\\ 1\leq k \leq n}}\prod_{k=1}^n e^{-\beta \frac{c_2}{2}\|\gamma_k\|}\prod_{\{i,j\}\in E(T)}\widetilde{F}_{\ctr_i,\ctr_j} = \\
        &=\hspace{-0.5cm}\sum_{\substack{\gamma_{i,j} \\ i=1,\dots,\ell-1 \\ j=1,\dots , m_i}}\prod_{i=1}^{\ell-1} \left( \prod_{j=1}^{m_i} e^{-\beta \frac{c_2}{2}\|\gamma_{i,j}\|} \widetilde{F}_{\gamma_{i-1,j'},\gamma_{i,j}}\right)\left[\sum_{\substack{\gamma_{\ell,j} \\ j=1,\dots,m_\ell}}\prod_{j=1}^{m_\ell} e^{-\beta \frac{c_2}{2}\|\gamma_{\ell,j}\|} \widetilde{F}_{\gamma_{\ell-1,j'},\gamma_{\ell,j}}\right].
    \end{split}
\end{equation}
Notice that the integers from $1$ to $m_{\ell}$ can be partitioned into $m_{\ell - 1}$ groups, the $q$-th group being composed of the integers $j$ such that $(\ell, j)$ is connected to the vertex $(\ell - 1, q)$ from the previous generation, that is, the partition is according to the parents. The ordering of vertices within the same generation can be changed in such a way that the vertices in group $q$ are numbered before the vertices in group $q'$ if $q < q'$. In this way, we can find integers $j_1, \dots, j_{m_{\ell - 1}}$ such that vertices $(\ell, j)$ with $j_q < j \leq j_{q+1}$ are connected to $(\ell - 1, q)$ (put $j_{m_{\ell - 1}} = m_{\ell} + 1$), see Figure \ref{partition_tree}.

  \begin{figure}[H]
\centering

\scalebox{0.6}{

\tikzset{every picture/.style={line width=0.75pt}} 

\begin{tikzpicture}[x=0.75pt,y=0.75pt,yscale=-1,xscale=1]

\draw   (291,345) .. controls (291,331.19) and (302.19,320) .. (316,320) .. controls (329.81,320) and (341,331.19) .. (341,345) .. controls (341,358.81) and (329.81,370) .. (316,370) .. controls (302.19,370) and (291,358.81) .. (291,345) -- cycle ;
\draw   (138,247) .. controls (138,233.19) and (149.19,222) .. (163,222) .. controls (176.81,222) and (188,233.19) .. (188,247) .. controls (188,260.81) and (176.81,272) .. (163,272) .. controls (149.19,272) and (138,260.81) .. (138,247) -- cycle ;
\draw   (292,219) .. controls (292,205.19) and (303.19,194) .. (317,194) .. controls (330.81,194) and (342,205.19) .. (342,219) .. controls (342,232.81) and (330.81,244) .. (317,244) .. controls (303.19,244) and (292,232.81) .. (292,219) -- cycle ;
\draw   (437,241) .. controls (437,227.19) and (448.19,216) .. (462,216) .. controls (475.81,216) and (487,227.19) .. (487,241) .. controls (487,254.81) and (475.81,266) .. (462,266) .. controls (448.19,266) and (437,254.81) .. (437,241) -- cycle ;
\draw   (55,137) .. controls (55,123.19) and (66.19,112) .. (80,112) .. controls (93.81,112) and (105,123.19) .. (105,137) .. controls (105,150.81) and (93.81,162) .. (80,162) .. controls (66.19,162) and (55,150.81) .. (55,137) -- cycle ;
\draw   (141,118) .. controls (141,104.19) and (152.19,93) .. (166,93) .. controls (179.81,93) and (191,104.19) .. (191,118) .. controls (191,131.81) and (179.81,143) .. (166,143) .. controls (152.19,143) and (141,131.81) .. (141,118) -- cycle ;
\draw   (228,114) .. controls (228,100.19) and (239.19,89) .. (253,89) .. controls (266.81,89) and (278,100.19) .. (278,114) .. controls (278,127.81) and (266.81,139) .. (253,139) .. controls (239.19,139) and (228,127.81) .. (228,114) -- cycle ;
\draw   (303,103) .. controls (303,89.19) and (314.19,78) .. (328,78) .. controls (341.81,78) and (353,89.19) .. (353,103) .. controls (353,116.81) and (341.81,128) .. (328,128) .. controls (314.19,128) and (303,116.81) .. (303,103) -- cycle ;
\draw   (380,113) .. controls (380,99.19) and (391.19,88) .. (405,88) .. controls (418.81,88) and (430,99.19) .. (430,113) .. controls (430,126.81) and (418.81,138) .. (405,138) .. controls (391.19,138) and (380,126.81) .. (380,113) -- cycle ;
\draw   (484,126) .. controls (484,112.19) and (495.19,101) .. (509,101) .. controls (522.81,101) and (534,112.19) .. (534,126) .. controls (534,139.81) and (522.81,151) .. (509,151) .. controls (495.19,151) and (484,139.81) .. (484,126) -- cycle ;
\draw    (179,266) -- (291,342) ;
\draw    (316,244) -- (316,320) ;
\draw    (341,345) -- (447,260.8) ;
\draw    (89.2,160.4) -- (145.4,229.6) ;
\draw    (166,143) -- (167,221.8) ;
\draw    (253,139) -- (300,200) ;
\draw    (328,128) -- (319.6,194.4) ;
\draw    (394.4,136) -- (337.8,203.4) ;
\draw    (468.6,216.6) -- (501.8,150.2) ;
\draw   (40.58,93.45) -- (189.23,67.91) -- (205,159.67) -- (56.35,185.21) -- cycle ;
\draw   (222.21,59.72) -- (443.65,69.7) -- (439.46,162.61) -- (218.02,152.63) -- cycle ;
\draw   (474,81.28) -- (558.01,94.23) -- (547,165.72) -- (462.99,152.77) -- cycle ;

\draw (309.33,332.73) node [anchor=north west][inner sep=0.75pt]  [font=\Large]  {$0$};
\draw (141.67,236.73) node [anchor=north west][inner sep=0.75pt]    {$( 1,\ 1)$};
\draw (296,208.4) node [anchor=north west][inner sep=0.75pt]    {$( 1,\ 2)$};
\draw (441.33,231.4) node [anchor=north west][inner sep=0.75pt]    {$( 1,\ 3)$};
\draw (59,126.4) node [anchor=north west][inner sep=0.75pt]    {$( 2,\ 1)$};
\draw (144.67,107.07) node [anchor=north west][inner sep=0.75pt]    {$( 2,\ 2)$};
\draw (232,104.4) node [anchor=north west][inner sep=0.75pt]    {$( 2,\ 3)$};
\draw (306.67,93.73) node [anchor=north west][inner sep=0.75pt]    {$( 2,\ 4)$};
\draw (383.33,102.73) node [anchor=north west][inner sep=0.75pt]    {$( 2,\ 5)$};
\draw (488.67,114.4) node [anchor=north west][inner sep=0.75pt]    {$( 2,\ 6)$};
\draw (75.22,54.21) node [anchor=north west][inner sep=0.75pt]  [font=\large,rotate=-352.63]  {$q\ =\ 1$};
\draw (310.26,35.97) node [anchor=north west][inner sep=0.75pt]  [font=\large,rotate=-1.52]  {$q\ =\ 2$};
\draw (500.77,57.17) node [anchor=north west][inner sep=0.75pt]  [font=\large,rotate=-8.76]  {$q\ =\ 3$};

\end{tikzpicture}

}

\caption{\label{partition_tree} The partition of the vertices of some generation according to their parents.  In this picture, the partition is performed in generation $2$. Notice that the number of elements of the partition is the number of elements in the previous generation, and they are being indexed by $q$. Here, we have $j_1 = 1$, $j_2 = 3$ and $j_3 = 6$.}
\end{figure}
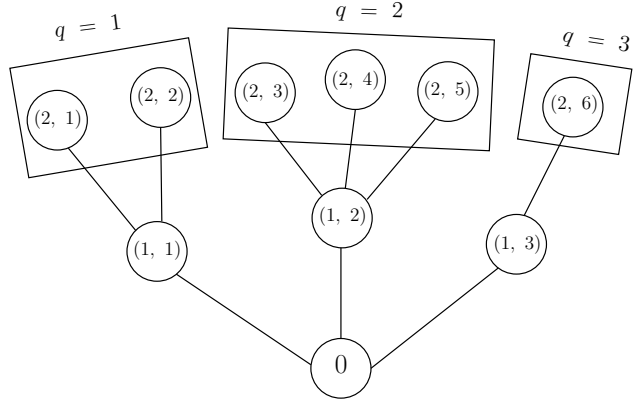
 
 Using Lemma \ref{F_Vol} we have
  \begin{equation}\label{eq5_tree}
      \begin{split}
           \sum_{\substack{\gamma_{\ell,j} \\ j=1,\dots,m_\ell}}\prod_{j=1}^{m_\ell} e^{-\beta \frac{c_2}{2}\|\gamma_{\ell,j}\|} \widetilde{F}_{\gamma_{\ell-1,j'},\gamma_{\ell,j}} &= \prod_{q=1}^{m_{\ell-1}}\prod_{j=j_q}^{j_{q+1}-1}\left(\sum_{\substack{\gamma_{\ell,j}}}e^{-\beta \frac{c_2}{2}\|\gamma_{\ell,j}\|}\widetilde{F}_{\gamma_{\ell-1,q},\gamma_{\ell,j}}\right)\\        
          &\leq c_\beta^{m_\ell}\prod_{q=1}^{m_{\ell - 1}} (j_{q+1} - j_q)! e^{\|\gamma_{\ell-1,q}\|},
      \end{split} 
  \end{equation}

  where we used in the last inequality above that $F_{\Sp(\gamma)} \leq \|\gamma\| $ and $x^n \leq n!e^x$. Using now that $\beta c_2 > 4$ and that $j_q-j_{q-1} = \deg(\ell-1,q)-1$, we can plug Inequality \eqref{eq5_tree} into \eqref{eq4_tree} and obtain
  \begin{equation}
      \begin{split}
          &\sum_{\substack{\gamma_k\\ 1\leq k \leq n}}\prod_{k=1}^n e^{-\beta \frac{c_2}{2}\|\gamma_k\|}\prod_{\{i,j\}\in E(T)}\widetilde{F}_{\ctr_i,\ctr_j} \\
          &\leq\sum_{\substack{\gamma_{i,j} \\ i=1,\dots,\ell-1 \\ j=1,\dots , m_i}}\prod_{i=1}^{\ell-1} \left( \prod_{j=1}^{m_i} e^{-\beta \frac{c_2}{2}\|\gamma_{i,j}\|} \widetilde{F}_{\gamma_{i-1,j'},\gamma_{i,j}}\right)c_\beta^{m_\ell}\prod_{q=1}^{m_{\ell - 1}} (\deg(\ell-1,q)-1)!e^{\beta\frac{c_2}{4}\|\gamma_{\ell-1,q}\|}.
      \end{split}
  \end{equation}

  Notice that, when cutting a generation $n$, the effect is, aside from some constant factors, to add $e^{\beta\frac{c_2}{4}\|\gamma\|}$ to the weight of each vertex of generation $n-1$.
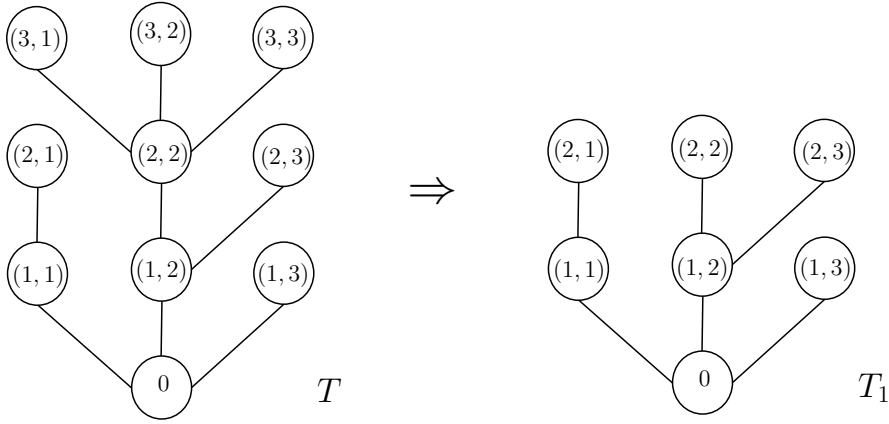
\begin{figure}[H]
\centering

\scalebox{0.75}{

\tikzset{every picture/.style={line width=0.75pt}} 

\begin{tikzpicture}[x=0.75pt,y=0.75pt,yscale=-1,xscale=1]

\draw   (148.29,289) .. controls (137.28,289.18) and (128.2,279.86) .. (128.01,268.2) .. controls (127.83,256.54) and (136.6,246.95) .. (147.62,246.77) .. controls (158.63,246.6) and (167.71,255.91) .. (167.89,267.57) .. controls (168.08,279.23) and (159.3,288.83) .. (148.29,289) -- cycle ;
\draw   (65.68,212.71) .. controls (54.66,212.88) and (45.59,203.57) .. (45.4,191.91) .. controls (45.22,180.25) and (53.99,170.65) .. (65,170.48) .. controls (76.02,170.3) and (85.09,179.61) .. (85.28,191.27) .. controls (85.46,202.94) and (76.69,212.53) .. (65.68,212.71) -- cycle ;
\draw   (148.16,209.94) .. controls (137.15,210.11) and (128.07,200.8) .. (127.89,189.14) .. controls (127.7,177.48) and (136.48,167.89) .. (147.49,167.71) .. controls (158.5,167.54) and (167.58,176.85) .. (167.76,188.51) .. controls (167.95,200.17) and (159.17,209.76) .. (148.16,209.94) -- cycle ;
\draw    (128.01,268.2) -- (65.68,212.71) ;
\draw   (229.32,212.03) .. controls (240.33,212.01) and (249.24,202.66) .. (249.22,191.16) .. controls (249.2,179.65) and (240.26,170.34) .. (229.25,170.37) .. controls (218.23,170.39) and (209.32,179.74) .. (209.34,191.24) .. controls (209.36,202.74) and (218.3,212.05) .. (229.32,212.03) -- cycle ;
\draw    (167.97,267.87) -- (229.32,212.03) ;
\draw    (147.62,246.77) -- (148.16,209.94) ;
\draw   (65.42,133.7) .. controls (54.41,133.88) and (45.33,124.56) .. (45.15,112.9) .. controls (44.96,101.24) and (53.74,91.65) .. (64.75,91.47) .. controls (75.76,91.3) and (84.84,100.61) .. (85.02,112.27) .. controls (85.21,123.93) and (76.43,133.53) .. (65.42,133.7) -- cycle ;
\draw   (147.91,130.93) .. controls (136.89,131.11) and (127.82,121.8) .. (127.63,110.14) .. controls (127.45,98.47) and (136.22,88.88) .. (147.23,88.7) .. controls (158.25,88.53) and (167.32,97.84) .. (167.51,109.5) .. controls (167.69,121.16) and (158.92,130.76) .. (147.91,130.93) -- cycle ;
\draw   (229.06,133.02) .. controls (240.08,133) and (248.99,123.66) .. (248.97,112.15) .. controls (248.95,100.65) and (240,91.34) .. (228.99,91.36) .. controls (217.98,91.38) and (209.07,100.73) .. (209.09,112.23) .. controls (209.11,123.74) and (218.05,133.05) .. (229.06,133.02) -- cycle ;
\draw    (167.71,188.86) -- (229.06,133.02) ;
\draw    (147.36,167.77) -- (147.91,130.93) ;
\draw   (65.17,54.69) .. controls (54.15,54.87) and (45.08,45.56) .. (44.89,33.9) .. controls (44.71,22.24) and (53.48,12.64) .. (64.5,12.46) .. controls (75.51,12.29) and (84.58,21.6) .. (84.77,33.26) .. controls (84.96,44.92) and (76.18,54.52) .. (65.17,54.69) -- cycle ;
\draw   (147.65,51.93) .. controls (136.64,52.1) and (127.56,42.79) .. (127.38,31.13) .. controls (127.19,19.47) and (135.97,9.87) .. (146.98,9.7) .. controls (157.99,9.52) and (167.07,18.83) .. (167.25,30.5) .. controls (167.44,42.16) and (158.66,51.75) .. (147.65,51.93) -- cycle ;
\draw    (127.5,110.19) -- (65.17,54.69) ;
\draw   (228.81,54.02) .. controls (239.82,53.99) and (248.73,44.65) .. (248.71,33.14) .. controls (248.69,21.64) and (239.75,12.33) .. (228.74,12.36) .. controls (217.72,12.38) and (208.81,21.72) .. (208.83,33.23) .. controls (208.85,44.73) and (217.8,54.04) .. (228.81,54.02) -- cycle ;
\draw    (167.46,109.86) -- (228.81,54.02) ;
\draw    (147.11,88.76) -- (147.65,51.93) ;
\draw    (65,170.48) -- (65.55,133.65) ;
\draw   (508.29,285.68) .. controls (497.28,285.86) and (488.2,276.55) .. (488.01,264.89) .. controls (487.83,253.22) and (496.6,243.63) .. (507.62,243.45) .. controls (518.63,243.28) and (527.71,252.59) .. (527.89,264.25) .. controls (528.08,275.91) and (519.3,285.51) .. (508.29,285.68) -- cycle ;
\draw   (425.68,209.39) .. controls (414.66,209.56) and (405.59,200.25) .. (405.4,188.59) .. controls (405.22,176.93) and (413.99,167.33) .. (425,167.16) .. controls (436.02,166.98) and (445.09,176.3) .. (445.28,187.96) .. controls (445.46,199.62) and (436.69,209.21) .. (425.68,209.39) -- cycle ;
\draw   (508.16,206.62) .. controls (497.15,206.8) and (488.07,197.49) .. (487.89,185.82) .. controls (487.7,174.16) and (496.48,164.57) .. (507.49,164.39) .. controls (518.5,164.22) and (527.58,173.53) .. (527.76,185.19) .. controls (527.95,196.85) and (519.17,206.45) .. (508.16,206.62) -- cycle ;
\draw    (488.01,264.89) -- (425.68,209.39) ;
\draw   (589.32,208.71) .. controls (600.33,208.69) and (609.24,199.34) .. (609.22,187.84) .. controls (609.2,176.34) and (600.26,167.03) .. (589.25,167.05) .. controls (578.23,167.07) and (569.32,176.42) .. (569.34,187.92) .. controls (569.36,199.43) and (578.3,208.74) .. (589.32,208.71) -- cycle ;
\draw    (527.97,264.55) -- (589.32,208.71) ;
\draw    (507.62,243.45) -- (508.16,206.62) ;
\draw   (425.42,130.38) .. controls (414.41,130.56) and (405.33,121.25) .. (405.15,109.59) .. controls (404.96,97.92) and (413.74,88.33) .. (424.75,88.15) .. controls (435.76,87.98) and (444.84,97.29) .. (445.02,108.95) .. controls (445.21,120.61) and (436.43,130.21) .. (425.42,130.38) -- cycle ;
\draw   (507.91,127.62) .. controls (496.89,127.79) and (487.82,118.48) .. (487.63,106.82) .. controls (487.45,95.16) and (496.22,85.56) .. (507.23,85.39) .. controls (518.25,85.21) and (527.32,94.52) .. (527.51,106.19) .. controls (527.69,117.85) and (518.92,127.44) .. (507.91,127.62) -- cycle ;
\draw   (589.06,129.71) .. controls (600.08,129.68) and (608.99,120.34) .. (608.97,108.83) .. controls (608.95,97.33) and (600,88.02) .. (588.99,88.04) .. controls (577.98,88.07) and (569.07,97.41) .. (569.09,108.92) .. controls (569.11,120.42) and (578.05,129.73) .. (589.06,129.71) -- cycle ;
\draw    (527.71,185.55) -- (589.06,129.71) ;
\draw    (507.36,164.45) -- (507.91,127.62) ;
\draw    (425,167.16) -- (425.55,130.33) ;

\draw (310,130.4) node [anchor=north west][inner sep=0.75pt]    { \huge $\Rightarrow$};

\draw (250,260.4) node [anchor=north west][inner sep=0.75pt]    {\Large $T$};
\draw (144,258.4) node [anchor=north west][inner sep=0.75pt]    {$0$};
\draw (47,183.4) node [anchor=north west][inner sep=0.75pt]    {$(1,1)$};
\draw (129,180) node [anchor=north west][inner sep=0.75pt]    {$(1,2)$};
\draw (210,183) node [anchor=north west][inner sep=0.75pt]    {$(1,3)$};
\draw (47,103.4) node [anchor=north west][inner sep=0.75pt]    {$(2,1)$};
\draw (130,102.4) node [anchor=north west][inner sep=0.75pt]    {$(2,2)$};
\draw (212,104.4) node [anchor=north west][inner sep=0.75pt]    {$(2,3)$};
\draw (45,25.4) node [anchor=north west][inner sep=0.75pt]    {$(3,1)$};
\draw (129,20) node [anchor=north west][inner sep=0.75pt]    {$(3,2)$};
\draw (209,24.4) node [anchor=north west][inner sep=0.75pt]    {$(3,3)$};
\draw (610,257.08) node [anchor=north west][inner sep=0.75pt]    {\Large $T_1$};
\draw (504,255.08) node [anchor=north west][inner sep=0.75pt]    {$0$};
\draw (407,180.08) node [anchor=north west][inner sep=0.75pt]    {$(1,1)$};
\draw (489,179.08) node [anchor=north west][inner sep=0.75pt]    {$(1,2)$};
\draw (571,179.08) node [anchor=north west][inner sep=0.75pt]    {$(1,3)$};
\draw (407,100.08) node [anchor=north west][inner sep=0.75pt]    {$(2,1)$};
\draw (490,99.08) node [anchor=north west][inner sep=0.75pt]    {$(2,2)$};
\draw (572,101.08) node [anchor=north west][inner sep=0.75pt]    {$(2,3)$};

\end{tikzpicture}

}

\caption{\label{cutting_tree} The erasure procedure transforming a tree $T$ into a $T_1$.}
\end{figure}	

  Before iterating the previous procedure to the remaining generations of the tree, we need to distinguish three cases. If $\ell =1$, then we have the bound
  \begin{equation}
      \begin{split}
          &\sum_{\substack{\gamma_k\\ 1\leq k \leq n}}\prod_{k=1}^n e^{-\beta \frac{c_2}{2}\|\gamma_k\|}\prod_{\{i,j\}\in E(T)}\widetilde{F}_{\ctr_i,\ctr_j} \leq c_\beta^n(n-1)!e^{\beta\frac{c_2}{4}\|\gamma_0\|}.
      \end{split}
  \end{equation}
  If $\ell>2$, then we can proceed similarly to \eqref{eq5_tree} and write
 \begin{equation*}
    \begin{split}
        &\sum_{\substack{\gamma_{i,j} \\ i=1,\dots,\ell-1 \\ j=1,\dots , m_i}}\prod_{i=1}^{\ell-1} \left( \prod_{j=1}^{m_i} e^{-\beta \frac{c_2}{2}\|\gamma_{i,j}\|}\widetilde{F}_{\gamma_{i-1,j'},\gamma_{i,j}}\right)\prod_{q=1}^{m_{\ell - 1}}e^{\beta \frac{c_2}{4}\|\gamma_{\ell-1,q}\|}  \\
        &\leq\hspace{-0.5cm}\sum_{\substack{\gamma_{i,j} \\ i=1,\dots,\ell-2 \\ j=1,\dots , m_i}}\prod_{i=1}^{\ell-2} \left( \prod_{j=1}^{m_i} e^{-\beta \frac{c_2}{2}\|\gamma_{i,j}\|} \widetilde{F}_{\gamma_{i-1,j'},\gamma_{i,j}}\right)\left[\sum_{\substack{\gamma_{\ell-1,j} \\ j=1,\dots,m_{\ell-1}}}\prod_{j=1}^{m_{\ell-1}} e^{-\beta \frac{c_2}{4}\|\gamma_{\ell-1,j}\|} \widetilde{F}_{\gamma_{\ell-2,j'},\gamma_{\ell-1,j}}\right]\\
        &\leq\hspace{-0.5cm}\sum_{\substack{\gamma_{i,j} \\ i=1,\dots,\ell-2 \\ j=1,\dots , m_i}}\prod_{i=1}^{\ell-2} \left( \prod_{j=1}^{m_i} e^{-\beta \frac{c_2}{2}\|\gamma_{i,j}\|} \widetilde{F}_{\gamma_{i-1,j'},\gamma_{i,j}}\right)\left[c_{\beta/2}^{m_{\ell-1}}\prod_{q=1}^{m_{\ell-2}} (\deg(\ell-2,q)-1)!e^{\beta\frac{c_2}{4}\|\gamma_{\ell-2,q}\|}\right].
    \end{split}
\end{equation*}
  For $\ell = 2$, we only have the term inside the square brackets above. In every case (possibly after an iteration), we have
  \begin{equation}\label{eq6_tree}
      \sum_{\substack{\gamma_k\\ 1\leq k \leq n}}\prod_{k=1}^n e^{-\beta \frac{c_2}{2}\|\gamma_k\|}\prod_{\{i,j\}\in E(T)}\widetilde{F}_{\ctr_i,\ctr_j} \leq c_{\beta/2}^n\prod_{k=0}^n c_T(k)! e^{\beta \frac{c_2}{4}\|\gamma_0\|},
  \end{equation}
  where $c_T(k)$ is the number of children of $k$ in the tree $T$, that is
  \[
    c_T(0)=\deg_T(0), \qquad c_T(k)=\deg_T(k)-1 \quad\text{for }k\neq 0.
\]
  Plugging Inequality \eqref{eq6_tree} back into \eqref{eq7_tree} and denoting by $\mathcal{T}_{n+1}(d_0,d_1,\dots,d_n)$ the set of labeled trees with specified degrees for its $n+1$ vertices, we get 
\begin{equation}\label{eq8_tree}
    \begin{split}
        \sum_{T \in \mathcal{T}^0_{n+1}}\sum_{\substack{\polymer \in \mathcal{E}^+_{n,\gamma_0}}} \prod_{\ctr \in \polymer} e^{-\beta \frac{c_2}{4}\|\ctr\|}& \prod_{\{\gamma,\gamma'\}\in E(T)} F_{\ctr,\ctr'}\leq c_{\beta/2}^ne^{-\beta \frac{c_2}{4}\|\ctr_0\|}\sum_{T\in \mathcal{T}^0_{n+1}}\frac{1}{n!}\left(\prod_{k=1}^n c_T(k)!\right)\\
         &=c_{\beta/2}^ne^{-\beta \frac{c_2}{4}\|\ctr_0\|}\sum_{\substack{d_0, ..., d_n \geq 1\\ d_0+\dots+d_n=2n}}\frac{1}{n!}\left(\prod_{k=1}^n c_T(k)!\right)|\mathcal{T}_{n+1}(d_0,d_1,\dots,d_n)|.\\
    \end{split}
\end{equation}

We now estimate the combinatorial factor. By Cayley's formula for labeled trees with prescribed degrees,
\[
|\mathcal T_{n+1}(d_0,d_1,\dots,d_n)|
=
\frac{(n-1)!}{\prod_{k=0}^n(d_k-1)!},
\]
and hence
\[
\begin{split}
\frac{1}{n!} \sum_{T\in\mathcal T^0_{n+1}} \prod_{v=0}^n c_T(v)! &= \sum_{\substack{ d_0,\dots,d_n\geq 1 \\ d_0+\cdots+d_n=2n }} \frac{1}{n!} d_0! \prod_{k=1}^n(d_k-1)!
|\mathcal T_{n+1}(d_0,\dots,d_n)| = \sum_{\substack{ d_0,\dots,d_n\geq 1\\ d_0+\cdots+d_n=2n }} \frac{d_0}{n}.
\end{split}
\]
Writing $a_k=d_k-1$, this becomes
\[
    \frac{1}{n}
    \sum_{\substack{
        a_0,\dots,a_n\geq 0\\
        a_0+\cdots+a_n=n-1
    }}
    (a_0+1) \leq \sum_{\substack{
        a_0,\dots,a_n\geq 0\\
        a_0+\cdots+a_n=n-1
    }} 1 = \binom{2n-1}{n}
\]
Using the inequality $\binom{n}{k} \leq 2^n$,

\begin{equation*}
\sum_{T \in \mathcal{T}^0_{n+1}}\sum_{\substack{\polymer \in \mathcal{E}^+_{n,\gamma_0}}} \prod_{\ctr \in \polymer} e^{-\beta \frac{c_2}{4}\|\ctr\|} \prod_{\{\gamma,\gamma'\}\in E(T)} F_{\ctr,\ctr'} \leq c_{\beta/2}^ne^{-\beta\frac{c_2}{4}\|\gamma_0\|}\binom{2n - 1}{n} \leq (4c_{\beta/2})^ne^{-\beta \frac{c_2}{4}\|\gamma_0\|},
\end{equation*}
\end{proof}

\begin{corollary}\label{main_corol_1}
    For $\beta$ large enough it holds, for every polymer $\polymer$, that

    \begin{equation}\label{sum_polymer_zero}
        \sum_{V(\polymer) \ni 0}\widetilde{z}^+_\beta(\polymer) \leq 2e^{-\beta \frac{c_2}{8}}.
    \end{equation}
\end{corollary}

\begin{proof}
   Note that we can bound the sum over polymers $\polymer$ such that $V(\polymer)$ contains $0$ by the sum over all contours $\gamma_0$ such that $0 \in V(\gamma_0)$ and then sum over all polymers containing all possible choices of $\gamma_0$,

\begin{equation}\label{corollary3.1}
\begin{split}
 \sum_{V(\polymer) \ni 0}\widetilde{z}^+_\beta(\polymer)&\leq\sum_{V(\ctr_0) \ni 0} \sum_{n\geq 0}\,\sum_{\substack{\polymer \in \mathcal{E}_{n, \gamma_0}}} \widetilde{z}^+_\beta(\polymer). 
\end{split}
\end{equation}
Proposition \ref{prop_tree} gives us, 

\begin{equation}\label{sum_prop3.3}
    \sum_{n\geq 0}\,\sum_{\substack{\polymer \in \mathcal{E}_{n, \gamma_0}}} \widetilde{z}^+_\beta(\polymer) \leq e^{-\beta \frac{c_2}{4}\|\gamma_0\|}\left[\sum_{n\geq 0}(4c_{\beta/2})^{n}\right] = \frac{e^{-\beta \frac{c_2}{4}\|\gamma_0\|}}{1-4c_{\beta/2}}.
\end{equation}

This fact, together with Proposition \ref{Bound_on_C_0_n} allows us to bound the right-hand side of the Inequality \eqref{corollary3.1} above as
\begin{equation*}
    \begin{split}
    \sum_{V(\polymer) \ni 0}\widetilde{z}^+_\beta(\polymer) &\leq \sum_{V(\gamma_0)\ni 0} \frac{e^{-\beta \frac{c_2}{4}\|\gamma_0\|}}{1-4c_{\beta/2}} \leq \frac{c_{\beta/2}}{1- 4c_{\beta/2}}.
    \end{split}
\end{equation*}
for $\beta > 8c_1/c_2$. Using the explicit expression for $c_{\beta/2} = e^{-\beta \frac{c_2}{8}}(1-e^{-\beta \frac{c_2}{8}})^{-1}$ we get the desired result for $\beta > 8\log(10)/c_2$.
  \end{proof}
\begin{proposition}\label{main_prop}
    For all large enough $\beta$ it holds, for every polymer $\polymer$, that

\begin{equation*}\label{sum_polymer_incompat}
        \sum_{\polymer' \not\sim \polymer} \widetilde{z}^+_\beta(\polymer') \leq e^{-\beta \frac{c_2}{16}}\|\Gamma\|,
    \end{equation*}
    where $\|\Gamma\| \coloneqq \sum_{\gamma \in \Gamma}\|\gamma\|$.
\end{proposition}
    \begin{proof}
When a polymer $\Gamma'$ is not compatible with another polymer $\Gamma$, two cases may happen. The first one is when there is $\gamma \in \Gamma$ and $\gamma' \in \Gamma'$ such that $\gamma \not\sim \gamma'$. The second case is when for any pair $\gamma \in \Gamma$ and $\gamma'\in \Gamma'$ we have $\gamma\sim \gamma'$ but condition (II) of compatibility is not satisfied. For this violation to happen, we must have $\widetilde{V}(\Gamma) \cap \widetilde{V}(\Gamma) \neq \emptyset$, $V(\Gamma)\cap \I_-(\gamma')^c\neq \emptyset$, for all $\gamma' \in \Gamma'$ and $V(\Gamma')\cap \I_-(\gamma)^c\neq\emptyset$, for all $\gamma \in \Gamma$. If $V(\Gamma) \cap \I_-(\gamma') = \emptyset$ for any $\gamma'\in \Gamma'$, this implies that $V(\Gamma)\subset \widetilde{V}(\Gamma')^c$, thus $\widetilde{V}(\Gamma) \cap \widetilde{V}(\Gamma') = \emptyset$. Since this is a contradiction, there must exist $\gamma' \in \Gamma'$ with $V(\Gamma) \cap \I_-(\gamma')\neq \emptyset$. 
Again, since they form an $(M,a)$-partition, we must have one of the following options happening:
\begin{enumerate}
    \item There is $\gamma\in \Gamma$ and $\gamma'\in \Gamma'$ such that $V(\gamma) \subset \I_-(\gamma')$ and $\widetilde{V}(\Gamma) \cap \widetilde{V}(\gamma')^c\neq \emptyset$.
    \item There is $\gamma' \in \Gamma'$ and $\gamma \in \Gamma$ such that $V(\gamma') \subset \I_-(\gamma)$ and $\widetilde{V}(\Gamma') \cap \widetilde{V}(\gamma)^c\neq \emptyset$.
\end{enumerate}

Let us call $A_1(\gamma)$ the set of polymers $\Gamma'$ satisfying the item $1$ above for a given $\gamma \in \Gamma$, and $A_2(\Gamma)$ the polymers $\Gamma'$ satisfying item $2$. We stress that both items are not mutually exclusive. Thus we get the upper bound
\begin{equation}\label{eq1_main_prop}
\sum_{\polymer' \not\sim \polymer}\widetilde{z}_\beta^+(\polymer') \leq  \sum_{\ctr \in \polymer}\sum_{\ctr'\not\sim \ctr}\sum_{\Gamma'\ni \gamma'}\widetilde{z}^+_\beta(\polymer')+\sum_{\gamma \in \Gamma}\sum_{\Gamma'\in A_1(\gamma)} \widetilde{z}_\beta^+(\polymer')+\sum_{\Gamma'\in A_2(\Gamma)} \widetilde{z}_\beta^+(\polymer').
\end{equation}
For the first term, we can use Equation \eqref{sum_prop3.3}, yielding
\[
 \sum_{\ctr'\not\sim \ctr}\sum_{\Gamma'\ni \gamma'}\widetilde{z}^+_\beta(\polymer')\leq \sum_{\ctr'\not\sim \ctr}\sum_{n\geq 0}\sum_{\Gamma' \in \mathcal{E}_{n, \gamma'}}\widetilde{z}^+_\beta(\polymer') \leq \frac{1}{1-4c_{\beta/2}} \sum_{\ctr'\not\sim \gamma} e^{-\beta \frac{c_2}{4}\|\gamma'\|}.
\]
Moreover, letting $R_{\gamma'} =  M \min\{|V(\gamma')|, |V(\gamma)|\}^{\frac{a}{d+1}}$, we get
\begin{equation}\label{bound111}
\begin{split}
 \sum_{\ctr'\not\sim \gamma} e^{-\beta \frac{c_2}{4}\|\gamma'\|} \leq \sum_{x \in \Sp(\gamma)}\sum_{\gamma': d(\gamma',x)\leq R_{\gamma'}}e^{-\beta \frac{c_2}{4}\|\gamma'\|} \leq \sum_{x \in \Sp(\gamma)}\sum_{\Sp(\gamma'') \ni x}e^{-\beta \frac{c_2}{4}\|\gamma''\|}|\mathcal{C}_{\gamma''}|,
\end{split}
\end{equation}
where $\mathcal{C}_{\gamma''}\coloneqq \{\gamma' \in \mathcal{E}^+: d(\gamma',x) \leq R_{\gamma''}, \exists y \in \Z^d \text{ s.t. } \Sp(\gamma')+y = \Sp(\gamma'')\}$. But notice that we have the inclusion
\[
\mathcal{C}_{\gamma''}\subset \bigcup_{z \in B_{R_{\gamma''}}(x)}\{\gamma' \in \mathcal{C}_{\gamma''}: z \in \Sp(\gamma')\}.
\]
Hence, we get 
\begin{equation}
\begin{split}
\sum_{x \in \Sp(\gamma)}\sum_{\Sp(\gamma'') \ni x}e^{-\beta \frac{c_2}{4}\|\gamma''\|}|\mathcal{C}_{\gamma''}| &\leq (3M)^d\sum_{x \in \Sp(\gamma)}\sum_{\Sp(\gamma'') \ni x}e^{-\beta \frac{c_2}{4}\|\gamma''\|}|\gamma''||V(\gamma'')|^{\frac{ad}{d+1}}\\
&= (3M)^d|\gamma|\sum_{\Sp(\gamma'') \ni 0}e^{-\beta \frac{c_2}{4}\|\gamma''\|}|\gamma''||V(\gamma'')|^{\frac{ad}{d+1}}.
\end{split}
\end{equation}
The usual Peierls argument yields
\[
    \sum_{\Sp(\gamma'') \ni 0}e^{-\beta \frac{c_2}{4}\|\gamma''\|}|\gamma''||V(\gamma'')|^{\frac{ad}{d+1}} \leq c_{\beta/2},
\]
where the last inequality holds for $\beta > \frac{8}{c_2}\left(c_1+\frac{ad^2}{d^2-1}+1\right)$. The last two sums on the right-hand side of \eqref{eq1_main_prop} can be estimated similarly. First, we get 
\begin{equation}\label{eq_incompatibility1}
\sum_{\Gamma'\in A_1(\gamma)} \widetilde{z}_\beta^+(\polymer') \leq \sum_{V(\gamma) \subset \I_-(\gamma') }\sum_{\Gamma'\ni \gamma'}\widetilde{z}_\beta^+(\polymer') \leq \frac{1}{1-4c_{\beta/2}}\sum_{\Sp(\gamma)\subset \I_-(\gamma')}e^{-\beta \frac{c_2}{4} \|\gamma'\|}\leq \frac{c_{\beta/2}|\gamma|}{1-4c_{\beta/2}}.
\end{equation}
Observe that, for every $\Gamma'\in A_2(\Gamma)$, there will always be a pair $\{\gamma_1,\gamma_2\}\subset \Gamma'$ such that $V(\gamma_1)\subset \I_-(\gamma)$ and $V(\gamma_2) \subset \I_-(\gamma)^c$. Thus the activity can be decomposed as
\begin{equation}
\widetilde{z}_\beta^+(\Gamma')\leq\sum_{\substack{\{\gamma_1,\gamma_2\}\subset \Gamma' \\ V(\gamma_1) \subset \I_-(\gamma) \\ V(\gamma_2)\subset \I_-(\gamma)^c}}F_{\gamma_1,\gamma_2}\sum_{\substack{\Gamma_1\cupdot \Gamma_2 = \Gamma' \\ \gamma_i \in \Gamma_i, i=1,2}}\widetilde{z}_\beta^+(\Gamma_1)\widetilde{z}_\beta^+(\Gamma_2).
\end{equation}
Hence,
\begin{equation}\label{eq_almosthere}
\begin{split}
\sum_{\Gamma'\in A_2(\Gamma)} \widetilde{z}_\beta^+(\polymer') &\leq \sum_{\gamma \in \Gamma}\sum_{\substack{\Gamma'\in A_2(\Gamma) \\ \exists \gamma'\in \Gamma' \text{ s.t. } V(\gamma')\subset \I_-(\gamma)}} \widetilde{z}_\beta^+(\Gamma') \\
&\leq \sum_{\gamma\in \Gamma}\sum_{\substack{\{\gamma_1,\gamma_2\} \\ V(\gamma_1) \subset \I_-(\gamma) \\ V(\gamma_2)\subset \I_-(\gamma)^c}}F_{\gamma_1,\gamma_2}\sum_{ \gamma_i \in \Gamma_i, i=1,2}\widetilde{z}_\beta^+(\Gamma_1)\widetilde{z}_\beta^+(\Gamma_2).
\end{split}
\end{equation}
Using Equation \eqref{sum_prop3.3}, we have
\begin{equation}\label{eq2incompatibility}
\begin{split}
\sum_{\Gamma'\in A_2(\Gamma)} \widetilde{z}_\beta^+(\polymer') \leq \frac{1}{(1-4c_{\beta/2})^2}\sum_{\gamma\in \Gamma}\sum_{\substack{\{\gamma_1,\gamma_2\} \\ V(\gamma_1) \subset \I_-(\gamma) \\ V(\gamma_2)\subset \I_-(\gamma)^c}}F_{\gamma_1,\gamma_2}e^{-\beta \frac{c_2}{4}( \|\gamma_1\|+ \|\gamma_2\|)}.
\end{split}
\end{equation}
Thus, to finish, we just need to consider the last sum. Notice that
\begin{equation}
\begin{split}
\sum_{\substack{\{\gamma_1,\gamma_2\} \\ V(\gamma_1) \subset \I_-(\gamma) \\ V(\gamma_2)\subset \I_-(\gamma)^c}}F_{\gamma_1,\gamma_2}e^{-\beta \frac{c_2}{4}( \|\gamma_1\|+ \|\gamma_2\|)} \leq \sum_{\substack{\{\gamma_1,\gamma_2\} \\ V(\gamma_1) \subset \I_-(\gamma) \\ V(\gamma_2)\subset \I_-(\gamma)^c}}\frac{J}{\dis(\gamma_1,\gamma_2)^\alpha}|V(\gamma_1)||V(\gamma_2)|e^{-\beta \frac{c_2}{4}( \|\gamma_1\|+ \|\gamma_2\|)} \\
\leq \sum_{\substack{x \in \I_-(\gamma) \\ y \in \I_-(\gamma)^c}}J_{xy}\left(\sum_{V(\gamma_1) \ni x}|V(\gamma_1)|e^{-\beta \frac{c_2}{4} \|\gamma_1\|}\right)\left(\sum_{ V(\gamma_2)\ni y}|V(\gamma_2)|e^{-\beta \frac{c_2}{4} \|\gamma_2\|}\right),
\end{split}
\end{equation}
yielding us that
\begin{equation}
\sum_{\Gamma'\in A_2(\Gamma)} \widetilde{z}_\beta^+(\polymer') \leq \left(\frac{c_{\beta/2}}{1-4c_{\beta/2}}\right)^2 \sum_{\gamma \in \Gamma}F_{\I_-(\gamma)}.
\end{equation}
 The stated inequality thus follows, when $\beta > \frac{16}{c_2}\log\left(4[(3M)^d+1]\right)$, by the explicit expression of $c_{\beta/2}$ and a trivial bound $F_{\I_-(\gamma)}+|\gamma|\leq \|\gamma\|$.
\end{proof}
\begin{corollary}\label{main_prop_weighted}
For all large enough $\beta$, it holds, for every polymer $\Gamma$, that
\[
    \sum_{\Gamma'\not\sim \Gamma}
    \widetilde z_\beta^+(\Gamma')
    e^{\beta\frac{c_2}{4}\|\Gamma'\|}
    \leq
    e^{-\beta\frac{c_2}{32}}\|\Gamma\|.
\]
\end{corollary}

\begin{proof}
By the definition of the activities, for every polymer $\Gamma'$,
\[
    \widetilde z_\beta^+(\Gamma')
    e^{\beta\frac{c_2}{4}\|\Gamma'\|}
    =
    \widetilde z_{\beta/2}^+(\Gamma').
\]
Therefore, applying Proposition \ref{main_prop} at inverse temperature
$\beta/2$, we get
\[
\begin{split}
    \sum_{\Gamma'\not\sim \Gamma}
    \widetilde z_\beta^+(\Gamma')
    e^{\beta\frac{c_2}{4}\|\Gamma'\|}
    &=
    \sum_{\Gamma'\not\sim \Gamma}
    \widetilde z_{\beta/2}^+(\Gamma')\leq e^{-\beta\frac{c_2}{32}}\|\Gamma\|.
\end{split}
\]
\end{proof}
\subsection{Convergence of the Cluster Expansion}

We start by stating a lemma that will allow us to conclude the convergence of the cluster expansion. The original statement and its proof can be found in \cite[Lemma 3.1]{Pfister1991LargeDA}. 
\begin{lemma}\label{lemma_pfister}
    Let $\mathscr{X}$ be a discrete countable set and for each $n \geq 1$ let $\psi_n:\mathscr{X}^n\rightarrow \mathbb{C}$ be a complex function symmetric with respect to any permutation of the elements of $\mathscr{X}^n$. Also, assume that
    \begin{equation*}
    \sum_{n \geq 1}\frac{1}{n!}\sum_{x_1 \in \mathscr{X}} \cdots \sum_{x_n \in \mathscr{X}}|\psi_n(x_1, \ldots, x_n)|<\infty.
    \end{equation*}
    Then it follows that
\begin{equation*}
    \exp\left( \sum_{n\geq 1}\frac{1}{n!} \sum_{(x_1,\dots,x_n)\in \mathscr{X}^n} \psi_n(x_1,\dots,x_n)\right) = 1+ \sum_{n\geq 1}\frac{1}{n!} \sum_{(x_1,\dots,x_n)\in \mathscr{X}^n} \Psi_n(x_1,\dots,x_n),
\end{equation*}
where
\begin{equation*}
    \Psi_n(x_1,\dots,x_n)
    =
    \sum_{\pi\in\mathcal P_n}
    \prod_{B\in\pi}
    \psi_{|B|}\bigl((x_i)_{i\in B}\bigr),
\end{equation*}
\end{lemma}
and $\mathcal{P}_n$ is the collection of all partitions of $\{1, ..., n\}$.
Let $G_n$ be the set of connected graphs with vertex set $\{1,\dots,n\}$. Given $n \geq 1$ and $X=(\Gamma_1,\dots,\Gamma_n)\in(\mathcal E^+)^n$, we define
\begin{equation}\label{psi}
    \psi_n(\Gamma_1,\dots,\Gamma_n) := \phi^T(\Gamma_1,\dots,\Gamma_n) \prod_{i=1}^n z_\beta^+(\Gamma_i),
\end{equation}

where $\phi^T$ are the traditional \emph{Ursell functions}, given by
\begin{equation}\label{ursell}
    \phi^T(\Gamma_1,\dots,\Gamma_n)
    =
    \sum_{g\in G_n}
    (-1)^{|E(g)|}
    \prod_{\{i,j\}\in E(g)}
    \mathbbm 1_{\Gamma_i\not\sim \Gamma_j}.
\end{equation}
When $n=1$, the empty product is taken to be equal to $1$. The following lemma is a consequence of the Mayer trick.

\begin{lemma}\label{mayer_trick_labelled}
    For every $n\geq 1$ and every
    $(\Gamma_1,\dots,\Gamma_n)\in(\mathcal E_\Lambda^+)^n$, the functions
    $\Psi_n$ associated with $\psi_n$ by Lemma \ref{lemma_pfister} satisfy
    \[
        \Psi_n(\Gamma_1,\dots,\Gamma_n)
        =
        \prod_{i=1}^n z_\beta^+(\Gamma_i)
        \prod_{1\leq i<j\leq n}
        \mathbbm 1_{\Gamma_i\sim\Gamma_j}.
    \]
\end{lemma}

The next result was first proved by Penrose in \cite{Penrose1963}; see also
\cite{Pfister1991LargeDA}. It is paramount for the proof of convergence of
the cluster expansion.

\begin{theorem}[\textbf{The tree-graph bound}]
\label{tree_graph_bound}
    Let $(\Gamma_1,\dots,\Gamma_n)\in(\mathcal E^+)^n$. Then
    \[
        \left|
            \phi^T(\Gamma_1,\dots,\Gamma_n)
        \right|
        \leq
        \sum_{T\in\mathcal T_n}
        \prod_{\{i,j\}\in E(T)}
        \mathbbm 1_{\Gamma_i\not\sim\Gamma_j},
    \]
    where $\mathcal T_n$ is the set of trees with vertex set
    $\{1,\dots,n\}$.
\end{theorem}

Given a family $(\Gamma_1, \dots, \Gamma_n) \in (\mathcal{E}^+)^n$, we denote by $V(\Gamma_1, ..., \Gamma_n)$ the union $\bigcup_{i = 1}^n V(\Gamma_i)$.

\begin{lemma}\label{main_lemma}
    For $\beta$ large enough, it holds that
    \begin{equation}\label{eq_0_main_lemma}
        \frac{1}{n!}
        \sum_{\substack{
            (\Gamma_1,\dots,\Gamma_n) \in (\mathcal E^+)^n\\
            x\in V(\Gamma_1,\dots,\Gamma_n)
        }}
        |\psi_n(\Gamma_1,\dots,\Gamma_n)|
        \leq
        2e^{-\beta\frac{c_2}{64}n},
    \end{equation}
    for every $x\in\mathbb Z^d$ and every $n\geq 1$.
\end{lemma}
\begin{proof}
Let us first consider the case $n=1$. Since
    $\phi^T(\Gamma)=1$, using \eqref{psi}, \eqref{def_ztilde}, and
    Corollary \ref{main_corol_1}, we get
    \begin{equation}\label{eq_1_main_lemma}
        \sum_{\substack{\Gamma\in\mathcal E^+\\ x\in V(\Gamma)}}
        |\psi_1(\Gamma)|
        \leq
        \sum_{\substack{\Gamma\in\mathcal E^+\\ x\in V(\Gamma)}}
        \widetilde z_\beta^+(\Gamma)
        \leq
        2e^{-\beta\frac{c_2}{8}}.
    \end{equation}
    In particular, this gives \eqref{eq_0_main_lemma} for $n=1$.

    We now assume $n > 1$ and consider labeled families with $n+1$ polymers. By symmetry of $\psi_{n+1}$,
    \begin{equation*}
        \frac{1}{(n+1)!}
        \sum_{\substack{
            (\Gamma_0,\dots,\Gamma_n)\in(\mathcal E^+)^{n+1}\\
            x\in V(\Gamma_0,\dots,\Gamma_n)
        }}
        |\psi_{n+1}(\Gamma_0,\dots,\Gamma_n)| \leq
        \frac{1}{n!}
        \sum_{\substack{
            (\Gamma_0,\dots,\Gamma_n)\in(\mathcal E^+)^{n+1}\\
            x\in V(\Gamma_0)
        }}
        |\psi_{n+1}(\Gamma_0,\dots,\Gamma_n)|.
    \end{equation*}
    Using \eqref{psi}, \eqref{def_ztilde}, and Theorem
    \ref{tree_graph_bound}, we obtain
    \[
    \begin{split}
        &\frac{1}{(n+1)!}
        \sum_{\substack{
            (\Gamma_0,\dots,\Gamma_n)\in(\mathcal E^+)^{n+1}\\
            x\in V(\Gamma_0,\dots,\Gamma_n)
        }}
        |\psi_{n+1}(\Gamma_0,\dots,\Gamma_n)|
        \\
        &\leq
        \sum_{\substack{\Gamma_0\in\mathcal E^+\\ x\in V(\Gamma_0)}}
        \widetilde z_\beta^+(\Gamma_0)
        \left(
            \frac{1}{n!}
            \sum_{T\in\mathcal T^0_{n+1}}
            \sum_{(\Gamma_1,\dots,\Gamma_n)\in(\mathcal E^+)^n}
            \prod_{k=1}^n
            \widetilde z_\beta^+(\Gamma_k)
            \prod_{\{i,j\}\in E(T)}
            \mathbbm 1_{\Gamma_i\not\sim\Gamma_j}
        \right),
    \end{split}
    \]
    where $\mathcal T^0_{n+1}$ denotes the set of trees with labeled vertex set $\{0,\dots,n\}$, rooted at $0$.

We now introduce tree-dependent coordinates on the vertices of a fixed rooted tree $T\in\mathcal T^0_{n+1}$. Recall that the vertices of $T$ are the original labels $\{0,1,\dots,n\}$, with root $0$. For each $i\geq 0$, let
\begin{equation*}
L_i(T) := \left\{v\in\{0,\dots,n\}: d_T(0,v)=i \right\}
\end{equation*}
be the $i$-th generation of $T$ and let $m_i := |L_i(T)|$. The coordinate of a vertex $v$ will be a pair $(i, j)$ such that $i$ is each generation and $j$ corresponds to an ordering in $L_i(T)=\{v_{i,1},\dots,v_{i,m_i}\}$ consistent with the parents. More precisely, we assume that the children of the vertex $v_{i-1,q}$ are precisely those $v_{i,j}$ with $j_q\leq j<j_{q+1}$, where $q=1,\dots,m_{i-1}$. Notice that $m_0=1$ and $v_{0,1}=0$. If $v_{i,j}$ belongs to generation $i\geq 1$, we denote by $v_{i-1,j'}$ its unique parent in generation $i-1$. Equivalently,
\[
\{v_{i-1,j'},v_{i,j}\}\in E(T).
\]

In what follows, we write
\[
\Gamma_{i,j}:=\Gamma_{v_{i,j}}.
\]
This notation does not introduce new labels for the polymers; it is only a tree-dependent coordinate system for the original labeled vertices $v\in\{0,\dots,n\}$.

With these notations and conventions, we can rewrite the tree weight as
\begin{equation*}
\sum_{\Gamma_1,\dots,\Gamma_n\in\mathcal E^+}
\prod_{k=1}^n
\widetilde z_\beta^+(\Gamma_k)
\prod_{\{r,s\}\in E(T)}
\mathbbm 1_{\Gamma_r\not\sim\Gamma_s} =
\sum_{\substack{\Gamma_{i,j} \\ i=1,\dots,\ell \\ j=1,\dots , m_i}}\hspace{-0.3cm}
\prod_{i=1}^{\ell} \prod_{j=1}^{m_i} \left(
\widetilde z_\beta^+(\Gamma_{i,j})
\mathbbm 1_{\Gamma_{i-1,j'}\not\sim\Gamma_{i,j}}
\right).
\end{equation*}
where $\ell$ is the last generation of $T$. Notice that the root polymer $\Gamma_0=\Gamma_{0,1}$ is not included in the product of activities on the left-hand side, since its activity has already been factored out.

Separating the last generation, this is
\begin{equation*}
\sum_{\substack{\Gamma_{i,j} \\ i=1,\dots,\ell-1 \\ j=1,\dots , m_i}}\hspace{-0.3cm} \prod_{i=1}^{\ell-1} \prod_{j=1}^{m_i} \left( \widetilde z_\beta^+(\Gamma_{i,j}) \mathbbm 1_{\Gamma_{i-1,j'}\not\sim\Gamma_{i,j}} \right) \times \sum_{\substack{ \Gamma_{\ell,j}\\ j=1,\dots,m_\ell }}
\prod_{j=1}^{m_\ell} \left( \widetilde z_\beta^+(\Gamma_{\ell,j}) \mathbbm 1_{\Gamma_{\ell-1,j'}\not\sim\Gamma_{\ell,j}} \right).
\end{equation*}

By the definition of $\widetilde z_\beta^+$, for every polymer $\Gamma$ we have
\[
    \widetilde z_\beta^+(\Gamma) \leq  \widetilde z_\beta^+(\Gamma) e^{\beta\frac{c_2}{4}\|\Gamma\|}
    =
    \widetilde z_{\beta/2}^+(\Gamma).
\]
Thus Proposition \ref{main_prop}, tells us that
\begin{equation}\label{eq_weighted_main_prop_use}
    \sum_{\Gamma'\not\sim \Gamma}
    \widetilde z_{\beta/2}^+(\Gamma')
    \leq
    e^{-\beta\frac{c_2}{32}}\|\Gamma\|,
\end{equation}
for $\beta$ large enough. Let $c_T(v)$ be the number of children of $v$ in $T$. Then $c_T(0)=\deg_T(0)$ and, for $v\neq 0$, $c_T(v)=\deg_T(v)-1$. Using these facts together with inequality $x^r\leq r!e^x$, we get
\[
\begin{split}
\sum_{\substack{ \Gamma_{\ell,j} \\ j=1,\dots,m_\ell }} \prod_{j=1}^{m_\ell} \left( \widetilde z_\beta^+ (\Gamma_{\ell,j}) \mathbbm 1_{\Gamma_{\ell-1,j'}\not\sim\Gamma_{\ell,j}} \right) &\leq \prod_{q=1}^{m_{\ell-1}} \left( \sum_{\Gamma'\not\sim \Gamma_{\ell-1,q}} \widetilde z_{\beta/2}^+(\Gamma') \right)^{c_T(v_{\ell-1,q})} \\
&\leq \prod_{q=1}^{m_{\ell-1}} \left( e^{-\beta\frac{c_2}{32}} \|\Gamma_{\ell-1,q}\| \right)^{c_T(v_{\ell-1,q})} \\
&\leq e^{-\beta\frac{c_2}{32}m_\ell} \prod_{q=1}^{m_{\ell-1}}
c_T(v_{\ell-1,q})!\, e^{\|\Gamma_{\ell-1,q}\|}.
\end{split}
\]
For $\beta$ large enough, $e^{\|\Gamma\|}\leq e^{\beta\frac{c_2}{4}\|\Gamma\|}$. Therefore,
\[
\begin{split}
\sum_{\substack{ \Gamma_{\ell,j} \\ j=1,\dots,m_\ell }} \prod_{j=1}^{m_\ell} \left( \widetilde z_\beta^+(\Gamma_{\ell,j}) \mathbbm 1_{\Gamma_{\ell-1,j'}\not\sim\Gamma_{\ell,j}} \right) \leq e^{-\beta\frac{c_2}{32}m_\ell} \prod_{q=1}^{m_{\ell-1}} c_T(v_{\ell-1,q})!\, e^{\beta\frac{c_2}{4}\|\Gamma_{\ell-1,q}\|}.
\end{split}
\]
After this first erasure, the only effect on the previous generation is the appearance of the additional factor
\[
e^{\beta\frac{c_2}{4}\|\Gamma_{\ell-1,q}\|}
\]
in the weight of each polymer in generation $\ell-1$, together with the factorial contribution associated with the number of its children. Together with the original activity of this polymer, this gives
\[
    \widetilde z_\beta^+(\Gamma_{\ell-1,q}) e^{\beta\frac{c_2}{4}\|\Gamma_{\ell-1,q}\|}
    =
    \widetilde z_{\beta/2}^+(\Gamma_{\ell-1,q}).
\]

Each non-root vertex receives such an extra factor at most once, namely when its children are erased. This factor is then combined with the original activity $\widetilde z_\beta^+$ of that same vertex, producing $\widetilde z_{\beta/2}^+$. When this vertex is later summed as a child of its own parent, the estimate is applied to this new activity, and there is no iteration of the form $\beta\mapsto \beta/2\mapsto \beta/4$.

Iterating the same argument from generation $\ell$ down to generation $1$ gives 
\begin{equation}\label{eq5_main_lemma_fixed_tree}
\begin{split}
&\sum_{\Gamma_1,\dots,\Gamma_n\in\mathcal E^+} \prod_{k=1}^n \widetilde z_\beta^+(\Gamma_k) \prod_{\{r,s\}\in E(T)} \mathbbm 1_{\Gamma_r\not\sim\Gamma_s} \leq e^{-\beta\frac{c_2}{32}n} e^{\beta\frac{c_2}{4}\|\Gamma_0\|} \prod_{v=0}^n c_T(v)!.
\end{split}
\end{equation}
The factor involving $\Gamma_0$ remains outside the sum over the non-root polymers, since the root activity has already been taken outside the parentheses.

Therefore,
\[
\begin{split}
&\frac{1}{n!} \sum_{T\in\mathcal T^0_{n+1}} \sum_{\Gamma_1,\dots,\Gamma_n\in\mathcal E^+} \prod_{k=1}^n \widetilde z_\beta^+(\Gamma_k) \prod_{\{r,s\}\in E(T)} \mathbbm 1_{\Gamma_r\not\sim\Gamma_s} \leq e^{-\beta\frac{c_2}{32}n} e^{\beta\frac{c_2}{4}\|\Gamma_0\|} \frac{1}{n!} \sum_{T\in\mathcal T^0_{n+1}} \prod_{v=0}^n c_T(v)! .
\end{split}
\]

We already know, by the proof of Proposition \ref{prop_tree}, that the combinatorial factor can be bounded from above by $4^n$. Increasing $\beta$, if necessary, this factor can be absorbed into the exponential decay. Thus we obtain
\begin{equation}\label{eq5_main_lemma}
\frac{1}{n!} \sum_{T\in\mathcal T^0_{n+1}} \sum_{\Gamma_1,\dots,\Gamma_n\in\mathcal E^+} \prod_{k=1}^n \widetilde z_\beta^+(\Gamma_k) \prod_{\{r,s\}\in E(T)} \mathbbm 1_{\Gamma_r\not\sim\Gamma_s} \leq e^{-\beta\frac{c_2}{64}n} e^{\beta\frac{c_2}{4}\|\Gamma_0\|}.
\end{equation}

Plugging \eqref{eq5_main_lemma} into the previous estimate yields
\[
\begin{split}
&\frac{1}{(n+1)!} \sum_{\substack{ (\Gamma_0,\dots,\Gamma_n)\in(\mathcal E^+)^{n+1} \\ x\in V(\Gamma_0,\dots,\Gamma_n) }} |\psi_{n+1}( \Gamma_0,\dots,\Gamma_n)| \leq e^{-\beta\frac{c_2}{64}n} \sum_{\substack{ \Gamma_0\in\mathcal E^+ \\ x\in V(\Gamma_0)}} \widetilde z_\beta^+(\Gamma_0) e^{\beta\frac{c_2}{4}\|\Gamma_0\|}.
\end{split}
\]
Using $\widetilde z_\beta^+(\Gamma_0)e^{\beta\frac{c_2}{4}\|\Gamma_0\|} = \widetilde z_{\beta/2}^+(\Gamma_0)$, in Corollary \ref{main_corol_1}, applied at inverse temperature $\beta/2$, the last sum is bounded by $2e^{-\beta\frac{c_2}{16}}$ for $\beta$ large enough. Therefore,
\[
\frac{1}{(n+1)!} \sum_{\substack{ (\Gamma_0,\dots,\Gamma_n)\in(\mathcal E^+)^{n+1}\\ x\in V(\Gamma_0,\dots,\Gamma_n) }} |\psi_{n+1}(\Gamma_0,\dots,\Gamma_n)| \leq 2e^{-\beta\frac{c_2}{64}(n+1)}.
\]
This proves \eqref{eq_0_main_lemma}.
\end{proof}

\begin{proposition}\label{cluster_expansion_log_partition}
For $\beta$ large enough, the logarithm of the partition function of Proposition \ref{gas_polymer_partition} can be written as
\[
\log \widetilde Z_{\Lambda,\beta}^+
=
\sum_{n\geq 1} \frac{1}{n!} \sum_{(\Gamma_1,\dots,\Gamma_n)\in(\mathcal E_\Lambda^+)^n} \phi^T(\Gamma_1,\dots,\Gamma_n) \prod_{i=1}^n z_\beta^+(\Gamma_i).
\]
\end{proposition}

\begin{proof}
Let
\[
\psi_n(\Gamma_1,\dots,\Gamma_n)
=
\phi^T(\Gamma_1,\dots,\Gamma_n) \prod_{i=1}^n z_\beta^+(\Gamma_i).
\]
By Lemma \ref{main_lemma}, for $\beta$ large enough,
\[
\begin{split}
\sum_{n\geq 1} \frac{1}{n!} \sum_{(\Gamma_1,\dots,\Gamma_n)\in(\mathcal E_\Lambda^+)^n} |\psi_n(\Gamma_1,\dots,\Gamma_n)| &\leq \sum_{x\in\Lambda} \sum_{n\geq 1} \frac{1}{n!} \sum_{\substack{ (\Gamma_1,\dots,\Gamma_n)\in(\mathcal E_\Lambda^+)^n \\ x\in V(\Gamma_1,\dots,\Gamma_n) }} |\psi_n(\Gamma_1,\dots,\Gamma_n)| \\[0.3cm]
&\leq 2|\Lambda| \sum_{n\geq 1} e^{-\beta\frac{c_2}{64}n} = 2|\Lambda| c_{\beta/16} <\infty.
\end{split}
\]
Therefore, the cluster series is absolutely convergent and Lemma \ref{lemma_pfister} gives us the desired identity.
\end{proof}
\begin{corollary}
The free energy of the normalized system can be written as
\[
f_\beta \coloneqq \lim_{N\to\infty} \frac{1}{|\Lambda_N|} \log \widetilde Z_{\Lambda_N,\beta}^+ = \sum_{n\geq 1} \frac{1}{n!} \sum_{\substack{ (\Gamma_1,\dots,\Gamma_n)\in(\mathcal E^+)^n \\ 0\in V(\Gamma_1,\dots,\Gamma_n) }} \frac{ \psi_n(\Gamma_1, ..., \Gamma_n) }{ |V(\Gamma_1,\dots,\Gamma_n)| }.
\]
\end{corollary}
\begin{proof}
See \cite{Kotecky2006Cluster}, page 533.
\end{proof}
   
\section{Decay of Correlations}\label{correlations}

In this section, we will estimate the decay of the two-point truncated correlation function. Our strategy consists in providing improved versions of some of the estimates given in Section 2. The next proposition is a straightforward generalization of Proposition \ref{gas_polymer_partition}, but with new weights.

The \emph{truncated two-point correlation function} at the finite volume $\Lambda$ between points $x_1$ and $x_2$ in $\mathbb{Z}^d$ is defined as 
\begin{equation}\label{correlation}
\langle\sigma_{x_1} ;\sigma_{x_2}\rangle_{\Lambda,\beta,\mathbf{h}}^+ \coloneqq  \langle\sigma_{x_1} \sigma_{x_2}\rangle_{\Lambda,\beta, \mathbf{h}}^+ -\langle\sigma_{x_1} \rangle_{\Lambda,\beta, \mathbf{h}}^+\langle\sigma_{x_2}\rangle_{\Lambda,\beta, \mathbf{h}}^+.
\end{equation}
More generally, we can also define the correlation of two bounded measurable functions $f,g:\Omega \rightarrow \mathbb{R}$ as 
\[
\langle f ; g\rangle_{\beta,\mathbf{h}}^+ \coloneqq \langle f g\rangle_{\beta, \mathbf{h}}^+ - \langle f\rangle_{\beta, \mathbf{h}}^+\langle g\rangle_{\beta, \mathbf{h}}^+.
\]

For a given collection of $n$ distinct points $x_1$, \dots, $x_n$, consider the Hamiltonian $H_{\Lambda,\mathbf{h}}^+$ with a complex magnetic field $\mathbf{h}=\{h_{x_k}\}_{1\leq k \leq n}$ as in Equation \eqref{hamiltonian_spin} and consider the partition functions
\[
Z_{\Lambda,\beta,\mathbf{h}}^+ = \sum_{\sigma \in \Omega_\Lambda^+}e^{-\beta H_{\Lambda,\mathbf{h}}^+(\sigma)} \quad \text{ and } \quad \Ztil^{+}_{\Lambda, \beta,\mathbf{h}} \coloneqq e^{\beta H_{\Lambda,\mathbf{h}}^+(\sigma_+)}Z_{\Lambda,\beta,\mathbf{h}}^+.
\]
Deriving the logarithm of the partition function with respect to the magnetic field variables yields us the \emph{n-point truncated correlation},
\begin{equation}\label{n-point}
    \langle \sigma_{x_1};\dots;\sigma_{x_n}\rangle_{\Lambda,\beta}^+ \coloneqq \beta^{-n}\frac{\partial}{\partial h_{x_1}}\ldots \frac{\partial}{\partial h_{x_n}}\log Z_{\Lambda,\beta,\mathbf{h}}^+\Bigr|_{\mathbf{h}=0}.
\end{equation}

\begin{proposition}\label{gas_polymer_partition_field}
    Let $\Lambda \Subset \Z^d$. Then it holds
\begin{equation}
    \Ztil^{+}_{\Lambda,\beta,\mathbf{h}} = 1 + \sum_{ \emptyset\neq \polymerb \subset \mathcal{E}^+_{\Lambda}}\prod_{\polymer \in \polymerb}z_{\beta,\mathbf{h}}^+(\polymer) \prod_{\left\{\polymer, \polymer'\right\}\subset X} \mathbbm{1}_{\polymer \sim \polymer'}.
\end{equation}
The quantity $\Ztil^{+}_{\beta,\mathbf{h},\Lambda}$ can be seen as the partition function of a gas of polymers with activity
\begin{equation}
    z_{\beta,\mathbf{h}}^+(\polymer) = K_\mathbf{h}(\polymer)\prod_{\ctr \in \polymer} W_\mathbf{h}(\ctr),
\end{equation}
where the functions $W_{\mathbf{h}}(\gamma)$ and $K_{\mathbf{h}}(\Gamma)$ are defined by modifying the Equations \eqref{weight} and \eqref{functionA} by the new one-body interaction
\[
\Phi_{1,\mathbf{h}}(\gamma\cup \Gamma_\gamma) = \Phi_1(\gamma\cup \Gamma_\gamma)+\sum_{x_k \in \widetilde{V}(\gamma)} 2h_{x_k}\mathbbm{1}_{\{\sigma_{x_k}=-1\}}.
\]
\end{proposition}
 Notice that we can write
\[
K_{\mathbf{h}}(\Gamma) \coloneqq \sum_{G\in \mathcal{G}_{\Gamma}}\langle\varphi_G\rangle_{\Gamma,\mathbf{h}},
\]
where
\[
\langle f \rangle_{\Gamma,\mathbf{h}} \coloneqq \frac{1}{\prod_{\gamma \in \Gamma}\widetilde{Z}_{\mathbf{h}}(\gamma)}\sum_{\substack{\ctrb_\gamma \in \is(\gamma) \\ \gamma \in \Gamma}}  f(\sigma)\prod_{\gamma \in \Gamma} e^{-\beta \Phi_{1,\mathbf{h}}(\gamma\cup \ctrb_\gamma)},
\]
and $\displaystyle \widetilde{Z}_{\mathbf{h}}(\gamma)=\sum_{\ctrb_\gamma \in \is(\gamma)} e^{-\beta \Phi_{1,\mathbf{h}}(\ctr \cup \ctrb_\gamma)}$, $\displaystyle \varphi_G = \prod_{\{\gamma,\gamma'\} \in E(G)} e^{-\beta \Phi_2(\gamma \cup \ctrb_\gamma, \gamma' \cup \ctrb_{\gamma'}) } - 1$. We adopt the convention
$\langle\cdot \rangle_{\Gamma,0} \coloneqq \langle \cdot \rangle_\Gamma$ when the magnetic field $\mathbf{h}$ is zero everywhere. Similarly, for each $\gamma \in \mathcal{E}^+$, we define 
\[
W_\mathbf{h}(\gamma) := \frac{\widetilde{Z}_{\mathbf{h}}(\gamma)}{\widetilde{Z}_{\mathbf{h}}(\check{\gamma})},
\]
and we also establish the convention to omit the magnetic field $\mathbf{h}$ in the notation when it is zero. 

The essential step for the convergence of the cluster expansion in Section 3 was Proposition \ref{important_prop}. Thus, if we can prove an analog of that proposition for the new activities $z_{\beta,\mathbf{h}}^+(\Gamma)$, the convergence will follow by the same arguments. 
\begin{proposition}\label{important_prop_fields}
For $M, \beta$ large enough, it holds that
\begin{equation}
    |z^+_{\beta,\mathbf{h}}(\polymer)| \leq \widetilde{z}^+_{\beta}(\polymer),
\end{equation}
for every $\Gamma$ and any $\mathbf{h}$ in the polydisc $D_n(0,r) \coloneqq \{\mathbf{h} \in \mathbb{C}^n: |h_{x_k}|< r, k=1,\dots, n\}$ of radius $r = (12\beta n)^{-1}$.
\end{proposition}
\begin{proof}
    Notice that $|z_{\beta,\mathbf{h}}^+(\Gamma)|\leq |K_{\mathbf{h}}(\Gamma)|\prod_{\gamma \in \Gamma}|W_{\mathbf{h}}(\gamma)|$, hence it is sufficient to give upper bounds to the absolute values of $K_\mathbf{h}(\Gamma)$ and $W_\mathbf{h}(\gamma)$. The Cauchy inequality gives us
    \begin{equation}
    \begin{split}
        |K_{\mathbf{h}}(\Gamma)|\leq  \sum_{G \in \mathcal{G}_\Gamma}\left|\frac{d}{d z}\log \langle e^{ z \varphi_G -2\beta \langle\mathbf{h},\mathbbm{1}_{\{\sigma_x=-1\}}\rangle}\rangle_\Gamma\Bigr|_{z=0}\right| \leq \sum_{G \in \mathcal{G}_\Gamma}\frac{1}{r_G}\sup_{|z| \leq r_G}\left|\log \langle e^{z \varphi_G -2\beta \langle\mathbf{h},\mathbbm{1}_{\{\sigma_x=-1\}}\rangle}\rangle_\Gamma\right|,
    \end{split}
    \end{equation}
    where $r_G = (6\|\varphi_G\|_\infty)^{-1}$ and $\displaystyle \langle\mathbf{h},\mathbbm{1}_{\{\sigma_x=-1\}}\rangle = \sum_{k=1}^n h_{x_k}\mathbbm{1}_{\{\sigma_{x_k}=-1\}}$. By our choice of radius for the polydisc, we can use that $|e^z - 1|\leq |z|e^{|z|}$ to get
    \begin{equation}
    \begin{split}
   |\langle e^{z\varphi_G -2\beta \langle\mathbf{h},\mathbbm{1}_{\{\sigma_x=-1\}}\rangle}-1\rangle_\Gamma|&\leq \left(\frac{1}{6}+2\beta \sum_{k=1}^n|h_{x_k}|\right)e^{\left(\frac{1}{6}+2\beta \sum_{k=1}^n|h_{x_k}|\right)} \leq \frac{e^{1/3}}{3}.
    \end{split}
    \end{equation}
    In this case, using that $|\log(1+z)|\leq 2|z|$ if $|z|\leq 1/2$, we get
    \[
    \left|\log \langle e^{z \varphi_G -2\beta \langle\mathbf{h},\mathbbm{1}_{\{\sigma_x=-1\}}\rangle}\rangle_\Gamma\right| \leq 2|\langle e^{z \varphi_G -2\beta \langle\mathbf{h},\mathbbm{1}_{\{\sigma_x=-1\}}\rangle}-1\rangle_\Gamma|\leq \frac{2e^{1/3}}{3}.
    \]
    Thus, using Inequality \eqref{bound_phi_2_2} and proceeding similarly to Proposition \ref{important_prop} we get
    \begin{equation}\label{bound1}
    |K_\mathbf{h}(\Gamma)|\leq 4e^{1/3}\left(\prod_{\gamma \in \Gamma}e^{\beta\frac{c_3}{2}F_{\widetilde{V}(\gamma)}}\right)\sum_{T \in \mathcal{T}_\Gamma}\prod_{\{\gamma,\gamma'\}\in E(T)}4\beta F_{\gamma,\gamma'}.
    \end{equation}

    For the weight $W_\mathbf{h}(\gamma)$, Proposition \ref{Prop: Cost_erasing_contour} together with the Cauchy inequality yields
    \begin{equation*}
        |W_\mathbf{h}(\gamma)|\leq e^{-\beta c_2 \|\gamma\|}\left|\frac{d}{dz}\log(\langle e^{z e^{-\beta \Delta H_\gamma}-\beta \langle\mathbf{h}, \mathbf{1}_{\sigma_x=-1}\rangle}\rangle_{\check{\gamma}})\Bigr|_{z=0}\right|\leq \frac{e^{-\beta c_2 \|\gamma\|}}{R}\sup_{|z|\leq R}|\log(\langle e^{z e^{-\beta \Delta H_\gamma}-2\beta \langle\mathbf{h}, \mathbf{1}_{\sigma_x=-1}\rangle}\rangle_{\check{\gamma}})|,
    \end{equation*}
where $R = (6 \|e^{-\beta \Delta H_\gamma}\|_\infty)^{-1}$. An analogous argument as we did for $K_\mathbf{h}(\Gamma)$ applies in this case, giving us
\begin{equation}\label{bound2}
 |W_\mathbf{h}(\gamma)|\leq 4e^{1/3} e^{-2\beta c_2 \|\gamma\|}.
\end{equation}
Using the upper bounds \eqref{bound1} and \eqref{bound2} and taking $M$ large such that $4c_3\leq c_2$ and $\beta > 32/c_2^2$ gives us the desired result.
\end{proof}

\begin{corollary}\label{corol_fields}
For $\beta$ large enough and every
$\mathbf h\in D_n(0,r)$, with $r=(12\beta n)^{-1}$, one has
\begin{equation}\label{cluster_fields}
    \log \widetilde Z_{\Lambda,\beta,\mathbf h}^+
    =
    \sum_{m\geq 1} \frac{1}{m!} \sum_{(\Gamma_1,\dots,\Gamma_m)\in(\mathcal E_\Lambda^+)^m} \phi^T(\Gamma_1,\dots,\Gamma_m) \prod_{j=1}^m z_{\beta,\mathbf h}^+(\Gamma_j).
\end{equation}
Moreover, the series on the right-hand side is absolutely convergent, uniformly for $\mathbf h\in D_n(0,r)$.
\end{corollary}

\begin{proof}
By Proposition \ref{important_prop_fields}, for every $\mathbf h\in D_n(0,r)$,
\[
    |z_{\beta,\mathbf h}^+(\Gamma)|
    \leq
    \widetilde z_\beta^+(\Gamma).
\]
Therefore, by the tree-graph bound and Lemma \ref{main_lemma},
\[
\begin{split}
    &\sum_{m\geq 1} \frac{1}{m!} \sum_{(\Gamma_1,\dots,\Gamma_m)\in(\mathcal E_\Lambda^+)^m} \left| \phi^T(\Gamma_1,\dots,\Gamma_m) \right| \prod_{j=1}^m \left| z_{\beta,\mathbf h}^+(\Gamma_j) \right| \\
    &\leq \sum_{m\geq 1} \frac{1}{m!} \sum_{(\Gamma_1,\dots,\Gamma_m)\in(\mathcal E_\Lambda^+)^m} \left| \phi^T(\Gamma_1,\dots,\Gamma_m) \right|  \prod_{j=1}^m \widetilde z_\beta^+(\Gamma_j)<\infty.
\end{split}
\]
The bound is uniform in $\mathbf h\in D_n(0,r)$. This yields \eqref{cluster_fields}.
\end{proof}

The $n$-point truncated correlation function \eqref{n-point} is multilinear, thus
\[
\langle \sigma_{x_1};\dots;\sigma_{x_n}\rangle_{\Lambda,\beta}^+ = (-2)^n\langle \sigma_{x_1}=-1;\dots;\sigma_{x_n}=-1\rangle_{\Lambda,\beta}^+,
\]
where in the right-hand side above is the $n$-point truncated correlation function as in Equation \eqref{n-point}, but with the normalized partition function $ \Ztil^+_{\Lambda,\beta,\mathbf{h}}$ instead.

\begin{proposition}\label{prop_derivatives}
For $\beta$ large enough, let $x_1,\dots,x_n\in\mathbb Z^d$ be distinct. Then
\[
\begin{split}
    \left|\left\langle \sigma_{x_1}=-1;\dots;  \sigma_{x_n}=-1\right\rangle_{\Lambda,\beta}^+ \right| \leq (12n)^n \sum_{m\geq 1} \frac{1}{m!} \sum_{\substack{ (\Gamma_1,\dots,\Gamma_m)\in(\mathcal E_\Lambda^+)^m \\ \{x_1,\dots,x_n\}\subset \widetilde{V}(\Gamma_1,\dots,\Gamma_m) }} \left| \phi^T(\Gamma_1,\dots,\Gamma_m) \right| \prod_{j=1}^m \widetilde z_\beta^+(\Gamma_j).
\end{split}
\]
\end{proposition}

\begin{proof}
Since the cluster expansion \eqref{cluster_fields} converges absolutely and uniformly in the polydisc $D_n(0,r)$, it may be differentiated term by term. Thus,
\[
\begin{split}
    \left\langle \sigma_{x_1}=-1;\dots; \sigma_{x_n}=-1\right\rangle_{\Lambda,\beta}^+  &= \beta^{-n} \sum_{m\geq 1} \frac{1}{m!} \sum_{\substack{ (\Gamma_1,\dots,\Gamma_m)\in(\mathcal E_\Lambda^+)^m \\  \{x_1,\dots,x_n\}\subset \widetilde{V}(\Gamma_1,\dots,\Gamma_m) }} \phi^T(\Gamma_1,\dots,\Gamma_m) \\
    &\qquad\qquad\times \left. \frac{\partial^n} {\partial h_{x_1}\cdots\partial h_{x_n}} \prod_{j=1}^m z_{\beta,\mathbf h}^+(\Gamma_j) \right|_{\mathbf h=0}.
\end{split}
\]
Indeed, if $\{x_1,\dots,x_n\}\not\subset \widetilde{V}(\Gamma_1,\dots,\Gamma_m)$, then the corresponding derivative vanishes.

By the multivariate Cauchy inequality,
\[
\begin{split}
    \left| \left. \frac{\partial^n} {\partial h_{x_1}\cdots\partial h_{x_n}} \prod_{j=1}^m z_{\beta,\mathbf h}^+(\Gamma_j) \right|_{\mathbf h=0}\right| \leq
    \frac{1}{r^n} \sup_{\mathbf h\in D_n(0,r)} \prod_{j=1}^m \left| z_{\beta,\mathbf h}^+(\Gamma_j) \right| \leq
    (12\beta n)^n \prod_{j=1}^m \widetilde z_\beta^+(\Gamma_j),
\end{split}
\]
where the last inequality follows from Proposition
\ref{important_prop_fields}. Taking absolute values in the differentiated cluster expansion and using $\beta^{-n}(12\beta n)^n=(12n)^n$ proves the claim.
\end{proof}

\subsection{The two-point function}

For the next lemma, consider for any $A,B$ disjoint finite subsets of $\Z^d$ the quantity $\displaystyle F_{A,B} \coloneqq \sum_{\substack{x \in A \\ y \in B}}J_{x,y}$.

\begin{lemma}\label{erasing_edges}
    For every $A,B,C \Subset \Z^d$, pairwise disjoint, it holds that
    \[
    F_{A,B}F_{B,C}\leq 2^{2\alpha-1}J|B|^2\diam(B)^\alpha F_{A,C}\left(\frac{1}{\dis(A,B)^\alpha}+\frac{1}{\dis(B,C)^\alpha}\right)
    \]
\end{lemma}
\begin{proof}
A simple manipulation gives us
    \[
     F_{A,B}F_{B,C} = \sum_{\substack{x\in A, x'\in C \\ y,y'\in B}}J_{x,y}J_{x',y'} = \sum_{\substack{x\in A \\ x' \in C}}J_{x,x'}\left(\sum_{y,y'\in B}\frac{J_{x,y}J_{x',y'}}{J_{x,x'}}\right).
    \]
    Using the triangle inequality, we get 
    \begin{equation}\label{eq_lemma_erasing}
    \begin{split}
    \sum_{y,y'\in B}\frac{J_{x,y}J_{x',y'}}{J_{x,x'}} &\leq \sum_{y,y'\in B} J\left(\frac{1}{|x-y|} + \frac{1}{|x'-y'|} + \frac{|y-y'|}{|x'-y'||x-y|}\right)^\alpha \\
    &\leq J|B|^2\left(\frac{1}{\dis(A,B)}+\frac{1}{\dis(B,C)}+\frac{\diam(B)}{\dis(A,B)\dis(B,C)}\right)^\alpha
    \\ 
    &\leq J|B|^2\diam(B)^\alpha \left(\frac{1}{\dis(A,B)}+\frac{1}{\dis(B,C)}+\frac{1}{\dis(A,B)\dis(B,C)}\right)^\alpha \\
    &\leq J|B|^2\diam(B)^\alpha\left(\frac{1}{\dis(A,B)^{\frac{1}{2}}}+\frac{1}{\dis(B,C)^{\frac{1}{2}}}\right)^{2\alpha},
    \end{split} 
    \end{equation}
    and using the Reverse H\"older Inequality, $(a^{1/p}+b^{1/p})^p \leq 2^{p-1}(a+b)$, for $a,b>0$ and $p>1$, we conclude the proof.
\end{proof}

\begin{proposition}\label{important_corol}
    For sufficiently large $\beta$, there exists $C_\beta >0$ such that, for each pair of compatible contours $\gamma_0,\gamma_1 \in \mathcal{E}^+$, it holds 
    \begin{equation}\label{eq_pr_important_corol}
    \sum_{\gamma \in \mathcal{E}_{1,\gamma_0,\gamma_1}^+}|\gamma|^2e^{-\beta \frac{c_2}{2} \|\gamma\|} F_{\gamma_0,\gamma}F_{\gamma,\gamma_1} \leq C_\beta |\gamma_0||\gamma_1|F_{\gamma_0,\gamma_1},
    \end{equation}
    where $\mathcal{E}_{n,\gamma_0,\gamma_1}= \{\Gamma\in\mathcal{E}^+: \Gamma=\{\gamma_0,\gamma_1,\dots, \gamma_{n+2}\} \}$.
\end{proposition}
\begin{proof}
Using Lemma \ref{erasing_edges} and that $\diam(\gamma) \leq C|\gamma|^{1+\frac{ad}{d^2-1}}$, where $C = Md2^{a(d+6)}(a+1)^{4a}$ (see Proposition 4.2.10 in \cite{maia2024phase} and the choice of $r$ in page 67) we get that 
\begin{equation}\label{eq_1_important_corol}
\begin{split}
F_{\gamma_0,\gamma} F_{\gamma,\gamma_1} \leq 2^{2\alpha - 1}C^\alpha J |\gamma|^b\left(\frac{1}{\dis(\gamma_0,\gamma)^\alpha}+\frac{1}{\dis(\gamma,\gamma_1)^\alpha}\right) F_{\gamma_0,\gamma_1},
\end{split}
\end{equation}
where $b= \frac{2d}{d-1} + \alpha\left(1+\frac{ad}{d^2 - 1}\right)$. The only terms in \eqref{eq_1_important_corol} depending on $\gamma$ are those containing the distances. Therefore, the left-hand side of \eqref{eq_pr_important_corol} is less than  
\begin{equation*}
\begin{split}
2^{2\alpha - 1}C^\alpha J  F_{\gamma_0,\gamma_1}\left(\sum_{\gamma \in \mathcal{E}^+_{1,\gamma_0}} e^{-\beta\frac{c_2}{2}|\gamma|}\frac{|\gamma|^{b+2}}{\dis(\gamma_0,\gamma)^\alpha}+\sum_{\gamma\in \mathcal{E}^+_{1,\gamma_1}} e^{-\beta\frac{c_2}{2}|\gamma|}\frac{|\gamma|^{b+2}}{\dis(\gamma,\gamma_1)^\alpha}\right).
\end{split}
\end{equation*}
Proceeding similarly as in Lemma \ref{F_Vol}, we get, for $\beta> 4(c_1 + b+1)/c_2$,
\begin{equation*}
\begin{split}
J\sum_{\gamma \in \mathcal{E}^+_{1,\gamma_0}} e^{-\beta\frac{c_2}{2}|\gamma|}\frac{|\gamma|^{b+2}}{\dis(\gamma_0,\gamma)^\alpha} \leq \left(\sum_{n\geq 1}n^{b+2}e^{(c_1-\beta \frac{c_2}{2})n}\right)F_{\Sp(\gamma_0)}\leq c_\beta F_{\Sp(\gamma_0)}\leq F_{\{0\}}c_\beta|\gamma_0|,
\end{split}
\end{equation*} 
where $F_{\{0\}} = \sum_{x \neq 0} J_{0x}$. Doing the same for $\gamma_1$, we can take $C_\beta = 2^{2\alpha - 1}C^\alpha F_{\{0\}}c_{\beta}$ to yield the desired result.
\end{proof}

\begin{proposition}\label{decay_proposition}
    For $\beta$ large enough, there is $C^{(1)}_\beta>0$ such that for any $x_1,x_2 \in \Z^d$, we have
    \[
    \sum_{\{x_1,x_2\}\subset \widetilde{V}(\Gamma)} \widetilde{z}_{\beta}^+(\Gamma) \leq C^{(1)}_\beta J_{x_1,x_2}.
    \]
\end{proposition}

\begin{proof}
First, we split the series into two terms
\begin{equation*}
\begin{split}
\sum_{\{x_1,x_2\}\subset \widetilde{V}(\Gamma)}\widetilde{z}_\beta^+(\Gamma) = \sum_{\substack{\{x_1,x_2\}\subset \widetilde{V}(\gamma) \\ \Gamma \ni \gamma}}\widetilde{z}_\beta^+(\Gamma)+\sum_{\substack{ x_1 \in \widetilde{V}(\gamma), x_2 \in \widetilde{V}(\gamma') \\ \{\gamma,\gamma'\}\subset \Gamma}}\widetilde{z}_\beta^+(\Gamma).
\end{split}
\end{equation*}
We will analyze the two series above separately. The first one can be bounded similarly as in Proposition \ref{prop_tree} and Corollary \ref{main_corol_1}
\begin{equation*}
  \sum_{\substack{\{x_1,x_2\}\subset \widetilde{V}(\gamma) \\ \Gamma \ni \gamma}}\widetilde{z}_\beta^+(\Gamma)\leq \frac{1}{1-6c_{\beta/2}}\sum_{\substack{\gamma \in \mathcal{E}^+ \\ \{x_1,x_2\}\subset \widetilde{V}(\gamma)}}\hspace{-0.5cm}e^{-\beta \frac{c_2}{4}\|\gamma\|}\leq \frac{|x_1-x_2|^\alpha e^{-\beta c_5|x_1-x_2|^{a'}}}{J(1-7e^{-\beta\frac{c_2}{8}})} J_{x_1x_2},
\end{equation*}
where the last inequality is due to the fact that $\diam(\gamma) \leq C|\gamma|^{1+\frac{ad}{d^2-1}}$, with $C = C(M, a, d)$ (see the beginning of the proof of Proposition \ref{important_corol}), $c_5 \coloneqq \frac{c_2}{8 C^{a'}}$ and $a'\coloneqq \frac{d^2-1}{d^2+ad-1}$. Thus, considering $\beta > 2\alpha/(a' c_5)$ yields 
\begin{equation*}
    \sum_{\substack{\{x_1,x_2\}\subset \widetilde{V}(\gamma)\\ \Gamma\ni \gamma}}\widetilde{z}_\beta^+(\Gamma) \leq \frac{e^{-\beta \frac{c_5}{2}}}{J(1-7e^{-\beta \frac{c_2}{8}})}J_{x_1,x_2}
\end{equation*}

For the second series, we will label the trees choosing the path between the contour containing $x_1$ and the one containing $x_2$ to be enumerated in an increasing order starting from $1$. The strategy consists of erasing first the trees that have roots in this path and, after that, estimating the sum over paths (see Figure \ref{Fig: Erasing_contours_outside_of_path}). 

\begin{figure}[ht]
    \centering
    \input{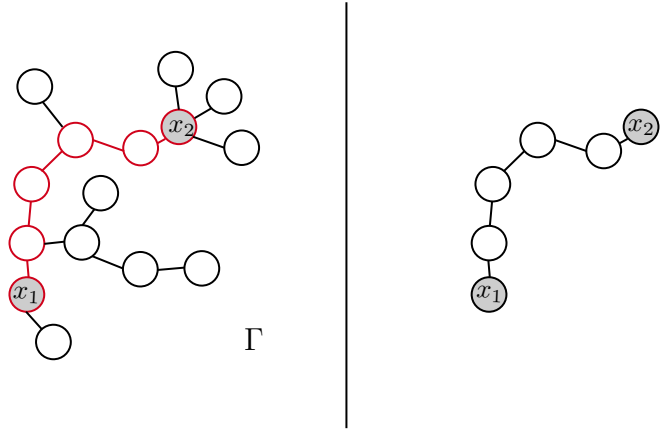}
    \caption{\small{On the left we have a polymer $\Gamma$, where each vertex represents a contour. The vertices in gray are $\gamma_{i}$ and $\gamma_{j}$, containing $x_1$ and $x_2$ respectively. The figure on the right depicts $\Gamma$ after we erase the vertices outside the path connecting $\gamma_{i}$ and $\gamma_{j}$.}}
    \label{Fig: Erasing_contours_outside_of_path}
\end{figure}

Let $\mathcal{T}_{x_1,x_2}$ be the set of all trees $T$ with $v(T)\in\mathcal{E}^+$ containing two distinct vertices  $\gamma, \gamma^\prime \in v(T)$ with $\{x_1,x_2\}\cap \widetilde{V}(\gamma)=\{x_1\}$ and $\{x_1,x_2\}\cap \widetilde{V}(\gamma')=\{x_2\}$. Given $T\in\mathcal{T}_{x_1,x_2}$, let $\gamma_{x_i}$ denote the unique vertex of $T$ with $x_i\in\widetilde{V}(\gamma_{x_i})$, for $i=1,2$. Then we can write
\begin{align*}
    \sum_{\substack{\Gamma; \exists \gamma, \gamma'\in \Gamma\\\gamma \neq \gamma'\\x_1 \in \widetilde{V}(\gamma), x_2 \in \widetilde{V}(\gamma')}} \widetilde{z}^+_\beta (\Gamma) 
     &= \sum_{T\in \mathcal{T}_{x_1,x_2}}\prod_{\ctr\in v(T)} e^{-\beta \frac{c_2}{2}\|\ctr\|}\prod_{\{\ctr,\ctr'\} \in E(T)}F_{\ctr,\ctr'}\\
     &=  \sum_{m\geq 2}\sum_{\ell=2}^m \sum_{\substack{T\in \mathcal{T}_{x_1,x_2} 
 \\  |v(T)|=m \\ d(\gamma_{x_1},\gamma_{x_2}) = \ell - 1}}\prod_{\ctr\in v(T)} e^{-\beta \frac{c_2}{2}\|\ctr\|}\prod_{\{\ctr,\ctr'\} \in E(T)}F_{\ctr,\ctr'},
\end{align*}
where in the last equation we split the sum according to the size of the tree and the distance in the tree between $\gamma_{x_1}$ and $\gamma_{x_2}$. For each tree $T$ appearing in the sum, we take $\gamma_1\coloneqq \gamma_{x_1}$, $\gamma_\ell\coloneqq\gamma_{x_2}$ and $\gamma_j$ the $j$-th contour in the path connecting $\gamma_1$ and $\gamma_\ell$. Denoting by $\mathcal{P}(\gamma_1, \ \dots, \ \gamma_\ell)$ the path passing through $\{\gamma_1, \ \dots, \ \gamma_\ell\}$, that is, a graph with edge set $\{\{\gamma_i,\gamma_{i+1} \}: i=1,\ \dots, \ell-1\}$, we have 
\begin{align*}
  &\sum_{m\geq 2}\sum_{\ell=2}^m \sum_{\substack{T\in \mathcal{T}_{x_1,x_2} 
 \\  |v(T)|=m \\ d(\gamma_{x_1},\gamma_{x_2}) = \ell - 1}}\prod_{\ctr\in v(T)} e^{-\beta \frac{c_2}{2}\|\ctr\|}\prod_{\{\ctr,\ctr'\} \in E(T)}F_{\ctr,\ctr'}\\
 & \hspace{1cm}= \sum_{m\geq 2}\sum_{\ell=2}^m \sum_{\substack{\gamma_p \\ 1\leq p\leq \ell}}\mathbbm{1}_{\left\{\substack{x_1 \in \widetilde{V}(\gamma_1), x_2 \in \widetilde{V}(\gamma_\ell)}\right\}}\sum_{\substack{T\in \mathcal{T}_{x_1,x_2} 
 \\  |v(T)|=m \\ \mathcal{P}(\gamma_1, \ \dots, \ \gamma_\ell)\subset T }}\prod_{\ctr\in v(T)} e^{-\beta \frac{c_2}{2}\|\ctr\|}\prod_{\{\ctr,\ctr'\} \in E(T)}F_{\ctr,\ctr'}\\
 & \hspace{1cm}= \sum_{m\geq 2}\sum_{\ell=2}^m \sum_{\substack{\gamma_p \\ 1\leq p\leq \ell}}\mathbbm{1}_{\left\{\substack{x_1 \in \widetilde{V}(\gamma_1), x_2 \in \widetilde{V}(\gamma_\ell)}\right\}} \sum_{\substack{\Gamma: \Gamma\cup\{\gamma_1, \dots, \gamma_\ell\}\in \mathcal{E}^+ \\ |\Gamma|=m-\ell}}\sum_{\substack{T\in \mathcal{T}_{\Gamma\cup\{\gamma_1, \dots, \gamma_\ell\}} 
 \\  \mathcal{P}(\gamma_1, \ \dots, \ \gamma_\ell)\subset T }}\prod_{\ctr\in v(T)} e^{-\beta \frac{c_2}{2}\|\ctr\|}\prod_{\{\ctr,\ctr'\} \in E(T)}F_{\ctr,\ctr'}\\
  & \hspace{1cm}= \sum_{m\geq 2}\sum_{\ell=2}^m \sum_{\substack{\gamma_p \\ 1\leq p\leq \ell}}\mathbbm{1}_{\left\{\substack{x_1 \in \widetilde{V}(\gamma_1), x_2 \in \widetilde{V}(\gamma_\ell)}\right\}} \sum_{\substack{\gamma_p \\ \ell+1\leq p\leq m}} \frac{\mathbbm{1}_{\left\{\{\gamma_1,\dots,\gamma_m\}\in \mathcal{E}^+\right\}}}{(m-\ell)!} \times \\
  & \hspace{9cm}\times\sum_{\substack{T\in \mathcal{T}_{\{\gamma_1, \dots, \gamma_m\}} 
 \\  \mathcal{P}(\gamma_1, \ \dots, \ \gamma_\ell)\subset T }}\prod_{1\leq p\leq m} e^{-\beta \frac{c_2}{2}\|\ctr_p\|}\prod_{\{\ctr,\ctr'\} \in E(T)}F_{\ctr,\ctr'},
\end{align*}
Each tree appearing in the summation can be uniquely associated with a tree in $\mathcal{T}_m^{1,\ell}$, the set of all labeled trees with $m$ vertices such that the vertices $1, 2, \dots, \ell$ are connected through a path, that is, the vertex $p$ is connected to $p+1$ 
 for $1\leq p \leq \ell-1$. This finally yields
 \begin{multline*}
       \sum_{\substack{\Gamma; \exists \gamma, \gamma'\in \Gamma\\\gamma \neq \gamma'\\x_1 \in \widetilde{V}(\gamma), x_2 \in \widetilde{V}(\gamma')}} \widetilde{z}^+_\beta (\Gamma) =
     \sum_{m \geq 2} \sum_{\ell=2}^m\frac{1}{(m-\ell)!} \sum_{\substack{\gamma_p\\1 \leq p \leq m}}\mathbbm{1}_{\left\{\substack{\substack{\{\gamma_1, \dots, \gamma_m\} \in \mathcal{E}^+};\\x_1 \in \widetilde{V}(\gamma_1), x_2 \in \widetilde{V}(\gamma_\ell)}\right\}} \left[ \sum_{T \in \mathcal{T}_m^{1,\ell}}\prod_{p = 1}^m e^{-\beta \frac{c_2}{2}\|\gamma_p\|}\prod_{\{r, s\} \in E(T)}F_{\gamma_r, \gamma_s}\right].
 \end{multline*}
 
 We can reorganize the sum in the following way: 

\begin{align*}
      \sum_{\substack{\Gamma; \exists \gamma, \gamma'\in \Gamma\\\gamma \neq \gamma'\\x_1 \in \widetilde{V}(\gamma), x_2 \in \widetilde{V}(\gamma')}} \widetilde{z}^+_\beta (\Gamma) 
      &\leq \sum_{m \geq 2} \sum_{\ell=2}^m\frac{1}{(m-\ell)!}\sum_{\substack{\gamma_p\\ 1\leq p \leq \ell}} \mathbbm{1}_{\left\{\substack{\substack{\{\gamma_1, \dots, \gamma_\ell\} \in \mathcal{E}^+};\\x_1 \in \widetilde{V}(\gamma_1), x_2 \in \widetilde{V}(\gamma_\ell)}\right\}} \left(\prod_{p=1}^\ell e^{-\beta \frac{c_2}{2}\|\gamma_p\|}\prod_{p=1}^{\ell-1}F_{\gamma_p, \gamma_{p+1}}\right) \times \\ &\hspace{1.0cm}\times \left[\sum_{T \in \mathcal{T}_m^{1,\ell}}\sum_{\substack{\gamma_p\\ \ell +1 \leq p \leq m}}\mathbbm{1}_{\{\{\gamma_{\ell+1}, \dots, \gamma_m\} \in \mathcal{E}^+\}}  \left(\prod_{p =\ell+1}^{m} e^{-\beta \frac{c_2}{2}\|\gamma_p\|}\prod_{\substack{\{r, s\} \in E(T) \\ \max\{r,s\} > \ell}}F_{\gamma_r, \gamma_s}\right) \right].
\end{align*}

Now we will rewrite the first summation in the last line as follows. Define $\mathcal{P}_\ell(M)$ as the collection of all the ordered partitions of $M$ with $\ell$ elements.

\begin{equation*}
    \sum_{T \in \mathcal{T}_m^{1,\ell}} \mapsto \sum_{\substack{\{m_p\}_{1 \leq p \leq \ell} \\ \sum m_p = m-\ell}} \sum_{\substack{(A_1, \dots, A_\ell) \in \mathcal{P}_{\ell}(\{\ell+1, \dots, m\})\\ |A_p| = m_p }} \sum_{\substack{T_p \in \mathcal{T}_{A_p \cup \{p\}}\\1 \leq p \leq \ell}} .
\end{equation*}

The terms above in square brackets can then be written as

\begin{align*}
   &  \sum_{\substack{\{m_p\}_{1 \leq p \leq \ell} \\ \sum m_p = m-\ell}} \sum_{\substack{(A_1, \dots, A_\ell) \in \mathcal{P}(\{\ell+1, \dots, m\})\\ |A_p| = m_p }} \sum_{\substack{T_p \in \mathcal{T}_{A_p \cup \{p\}}\\1 \leq p \leq \ell}} \sum_{\substack{\gamma_p\\ \ell +1 \leq p \leq m}}\mathbbm{1}_{\{\{\gamma_{\ell+1}, \dots, \gamma_m\} \in \mathcal{E}^+\}}  \left(\prod_{p =\ell+1}^{m} e^{-\beta \frac{c_2}{2}\|\gamma_p\|}\prod_{\substack{\{r, s\} \in E(T) \\ \max\{
   r,s\} > \ell }}F_{\gamma_r, \gamma_s}\right)  \\
&  \hspace{8.0cm} \leq \sum_{\substack{\{m_p\}_{1 \leq p \leq \ell} \\ \sum m_p = m-\ell}} \sum_{\substack{(A_1, \dots, A_\ell) \in \mathcal{P}(\{\ell+1, \dots, m\})\\ |A_p| = m_p }} \prod_{p =1}^\ell S_{A_p},
\end{align*}

where

\begin{equation*}
    S_{A_p} =\sum_{\substack{T_p \in \mathcal{T}_{A_p \cup \{p\}}}} \sum_{\substack{\gamma_r\\ r \in A_p}} \mathbbm{1}_{\left\{\substack{\{\gamma_r\}_{A_p\cup\{p\}} \in \mathcal{E}^+}\right\}}\prod_{ r\in A_p} e^{-\beta \frac{c_2}{2}\|\gamma_r\|}\prod_{\{r, s\} \in E(T_p)}F_{\gamma_r, \gamma_s}. 
\end{equation*}

Proceeding as the proof of Proposition \ref{prop_tree}, we can bound $S_{A_p} \leq (6c_{\beta/2})^{|A_p|}|A_p|! e^{\beta \frac{c_2}{4}\|\gamma_p\|}$, so

\begin{align*}
   \sum_{\substack{(m_p)_{1\leq p\leq \ell}\\ \sum m_p = m- \ell}}\;\sum_{\substack{(A_1, \dots, A_\ell) \in \mathcal{P}(\{\ell+1, \dots, m\})\\ |A_p| = m_p }} \prod_{p =1}^\ell S_{A_p} &\leq \prod_{p=1}^\ell e^{\beta \frac{c_2}{4}\|\gamma_p\|} \sum_{\substack{(m_p)_{1\leq p\leq \ell}\\ \sum m_p = m- \ell}}\;\sum_{\substack{(A_1, \dots, A_\ell) \in \mathcal{P}(\{\ell+1, \dots, m\})\\ |A_p| = m_p }} \left(\prod_{p =1}^\ell (6c_{\beta/2})^{|A_p|}|A_p|! \right)\\
     &= (6c_{\beta/2})^{m-\ell}\prod_{p=1}^\ell e^{\beta \frac{c_2}{4}\|\gamma_p\|} \frac{(m-1)!}{(\ell - 1)!}.
\end{align*}

Returning this result to the whole expression, we have

\begin{align}\label{trambolho2}
     \sum_{\substack{\Gamma; \exists \gamma, \gamma'\in \Gamma\\\gamma \neq \gamma'\\x_1 \in \widetilde{V}(\gamma), x_2 \in \widetilde{V}(\gamma')}} \widetilde{z}^+_\beta (\Gamma) 
      &\leq \sum_{m \geq 2} \sum_{\ell=2}^m\binom{m-1}{\ell-1}(6c_{\beta/2})^{m-\ell} \sum_{\substack{\gamma_p\\ 1\leq p \leq \ell}} \mathbbm{1}_{\left\{\substack{\substack{\{\gamma_1, \dots, \gamma_\ell\} \in \mathcal{E}^+};\\x_1 \in \widetilde{V}(\gamma_1), x_2 \in \widetilde{V}(\gamma_\ell)}\right\}} \left(\prod_{p=1}^\ell e^{-\beta \frac{c_2}{4}\|\gamma_p\|}\prod_{p=1}^{\ell-1}F_{\gamma_p, \gamma_{p+1}}\right).
\end{align}

 We will give an upper bound to each element of the sum over $\ell$ in \eqref{trambolho2} in terms of the case $\ell = 2$. For this, we will recursively use Proposition \ref{important_corol}, always erasing the vertex in the middle of the path. In the first step, this corresponds to the term $\lceil \ell/2\rceil$ (See Figure \ref{erasing_lines}). If we erased the vertices from one end of the path to the other, the powers of $|\gamma_\ell|$ would accumulate and depend on $\ell$, which would not be convenient. Our strategy of erasing only the vertex in the middle of the path ensures that the powers appearing in the computation will be at most $2$.

 \begin{figure}[ht]
    \centering
    \input{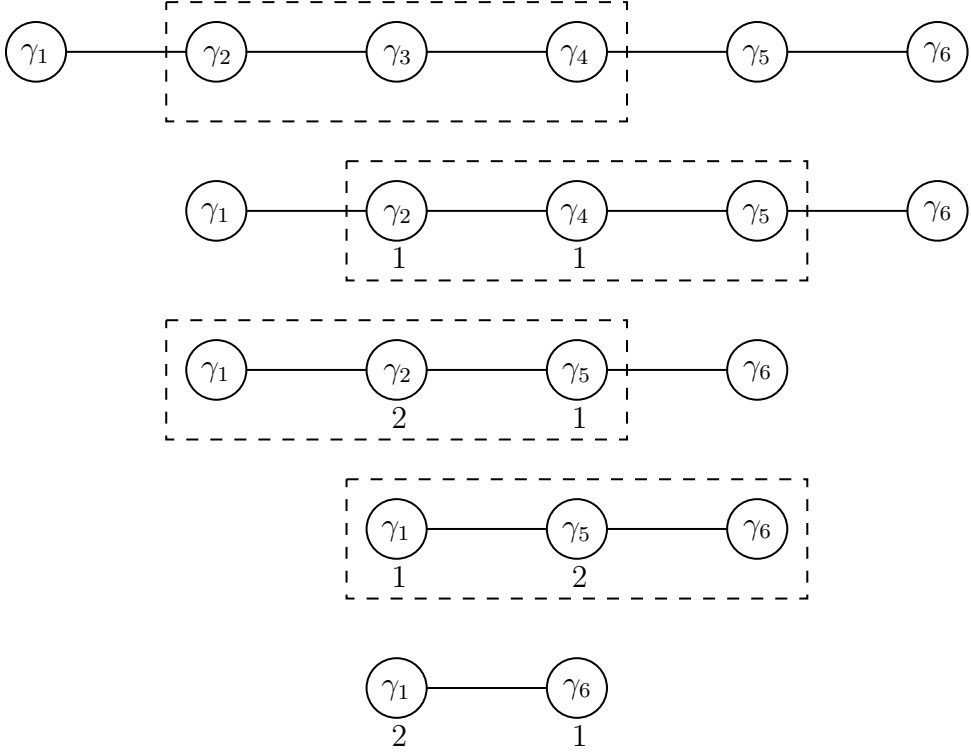}
    \caption{Example of the procedure to erase vertices between $\gamma_1$ and $\gamma_\ell$. Here, $\ell = 6$. The dashed boxes indicate the contours involved in each step. The numbers inside the boxes and below the vertices count the power of each contour in the calculation.}
    \label{erasing_lines}
\end{figure}

\begin{align*}
    &\sum_{\substack{\gamma_p\\ 1\leq p \leq \ell}} \mathbbm{1}_{\left\{\substack{\substack{\{\gamma_1, \dots, \gamma_\ell\} \in \mathcal{E}^+};\\x_1 \in \widetilde{V}(\gamma_1), x_2 \in \widetilde{V}(\gamma_\ell)}\right\}} \left(\prod_{p=1}^\ell e^{-\beta \frac{c_2}{2}\|\gamma_p\|}\prod_{p=1}^{\ell-1}F_{\gamma_p, \gamma_{p+1}}\right) \leq \\
    &\hspace{3.0cm}\sum_{\substack{\gamma_p\\ 1 \leq p \leq \ell\\ p \neq \lceil\ell/2\rceil }}\mathbbm{1}_{\left\{\substack{\substack{\{\gamma_p\}_{1\leq p \leq \ell, p \neq \lceil \ell/2 \rceil} \in \mathcal{E}^+};\\x_1 \in \widetilde{V}(\gamma_1), x_2 \in \widetilde{V}(\gamma_\ell)}\right\}}\left(\prod_{\substack{1 \leq p \leq \ell \\ p \neq \lceil\ell/2\rceil}} e^{-\beta \frac{c_2}{2}\|\gamma_p\|}\prod_{\substack{1 \leq p \leq \ell-1 \\ p \neq \lceil\ell/2\rceil}}F_{\gamma_p, \gamma_{p+1}}\right) \times  \\
    &\hspace{4.0cm}\times\sum_{\gamma_{\lceil\ell/2\rceil}} \mathbbm{1}_{\left\{ \{\gamma_{\lceil\ell/2\rceil - 1}, \gamma_{\lceil\ell/2\rceil}, \gamma_{\lceil\ell/2\rceil+1} \} \in \mathcal{E}^+\right\}} e^{-\beta \frac{c_2}{2}\|\gamma_{\lceil\ell/2\rceil}\|} F_{\gamma_{\lceil\ell/2\rceil - 1}, \gamma_{\lceil\ell/2\rceil}}F_{\gamma_{\lceil\ell/2\rceil}, \gamma_{\lceil\ell/2\rceil+1}}\\
   & \hspace{4.0cm}\leq \sum_{\substack{\gamma_p\\ 1 \leq p \leq \ell\\ p \neq \lceil\ell/2\rceil }}\mathbbm{1}_{\left\{\substack{\substack{\{\gamma_p\}_{1\leq p \leq \ell, p \neq \lceil \ell/2 \rceil} \in \mathcal{E}^+};\\x_1 \in \widetilde{V}(\gamma_1), x_2 \in \widetilde{V}(\gamma_\ell)}\right\}}\left(\prod_{\substack{1 \leq p \leq \ell \\ p \neq \lceil\ell/2\rceil}} e^{-\beta \frac{c_2}{2}\|\gamma_p\|}\prod_{\substack{1 \leq p \leq \ell-1 \\ p \neq \lceil\ell/2\rceil}}F_{\gamma_p, \gamma_{p+1}}\right) \times \\
   & \hspace{5cm} \times C_{\beta} |\gamma_{\lceil\ell/2\rceil -1}||\gamma_{\lceil\ell/2\rceil+1}| F_{\gamma_{\lceil\ell/2\rceil-1}, \gamma_{\lceil\ell/2\rceil+1}}.
   \end{align*}
Notice that, as we use Proposition \ref{important_corol} to erase the vertices with $1 < p < \ell$, we get an extra constant $C_\beta$. After the recursion, we will end up with

\begin{equation*}
    \begin{split}
&\sum_{\substack{\gamma_p\\ 1 \leq p \leq \ell }} \mathbbm{1}_{\left\{\substack{\substack{\{\gamma_1, \dots, \gamma_\ell\} \in \mathcal{E}^+};\\x_1 \in \widetilde{V}(\gamma_1), x_2 \in \widetilde{V}(\gamma_\ell)}\right\}}  \left(\prod_{p =1}^\ell e^{-\beta \frac{c_2}{4}\|\gamma_p\|}\prod_{p = 1}^{\ell - 1}F_{\gamma_p, \gamma_{p+1}}\right) \leq \\
&\hspace{4.0cm}C_\beta^{\ell - 2}\sum_{\gamma_1, \gamma_\ell}  \mathbbm{1}_{\left\{\substack{x_1 \in \widetilde{V}(\gamma_1), x_2 \in \widetilde{V}(\gamma_\ell)}\right\}} \mathbbm{1}_{\left\{ \{\gamma_1, \gamma_\ell\} \in \mathcal{E}^+ \right\}} |\gamma_1|^2|\gamma_\ell|^2e^{-\beta \frac{c_2}{4}(\|\gamma_1\|+\|\gamma_\ell\|)}F_{\gamma_{1}, \gamma_\ell}.
\end{split}
\end{equation*}

Using that, when $A,B \Subset \Z^d$ are disjoint, we have $F_{A,B} \leq 2^\alpha J_{x,y}|A\|B|\diam(A)^\alpha\diam(B)^\alpha$ for any $x\in A$ and $y \in B$, we get

\begin{align*}
    &\sum_{\substack{\gamma_p\\ 1 \leq p \leq \ell }} \mathbbm{1}_{\left\{\substack{\substack{\{\gamma_1, \dots, \gamma_\ell\} \in \mathcal{E}^+};\\x_1 \in \widetilde{V}(\gamma_1), x_2 \in \widetilde{V}(\gamma_\ell)}\right\}}  \left(\prod_{p =1}^\ell e^{-\beta \frac{c_2}{4}\|\gamma_p\|}\prod_{p = 1}^{\ell - 1}F_{\gamma_p, \gamma_{p+1}}\right) \leq \\[0.3cm]
    &2^{\alpha}J_{x_1, x_2} C_\beta^{\ell - 2}\sum_{\gamma_1, \gamma_\ell}  \mathbbm{1}_{\left\{\substack{\substack{\{\gamma_1, \gamma_\ell\} \in \mathcal{E}^+};\\x_1 \in \widetilde{V}(\gamma_1), x_2 \in \widetilde{V}(\gamma_\ell)}\right\}} |\gamma_1|^2 |\gamma_\ell|^2 e^{-\beta \frac{c_2}{4}(\|\gamma_1\|+\|\gamma_\ell\|)}|\widetilde{V}(\gamma_1)\|\widetilde{V}(\gamma_\ell)|\diam(\gamma_1)^{\alpha}\diam(\gamma_\ell)^\alpha.
\end{align*}

Now we are going to use the isoperimetric inequality so that $|\widetilde{V}(\gamma)| \leq |\partial_{\mathrm{in}} \widetilde{V}(\gamma)|^{\frac{d}{d-1}} \leq |\gamma|^{\frac{d}{d-1}}$. Finally, we will use again that $\diam(\gamma) \leq C|\gamma|^{1+\frac{ad}{d^2-1}}$.

\begin{align*}
    &\sum_{\substack{\gamma_p\\ 1 \leq p \leq \ell }} \mathbbm{1}_{\left\{\substack{\substack{\{\gamma_1, \dots, \gamma_\ell\} \in \mathcal{E}^+};\\x_1 \in \widetilde{V}(\gamma_1), x_2 \in \widetilde{V}(\gamma_\ell)}\right\}}  \left(\prod_{p =1}^\ell e^{-\beta \frac{c_2}{4}\|\gamma_p\|}\prod_{p = 1}^{\ell - 1}F_{\gamma_p, \gamma_{p+1}}\right) \leq \\[0.3cm]
    &\leq 2^{\alpha}J_{x_1, x_2} C_\beta^{\ell - 2}C^2\sum_{\substack{\gamma_1\\ x_1 \in \widetilde{V}(\gamma_1)}} \sum_{\substack{\gamma_\ell\\ x_2 \in \widetilde{V}(\gamma_\ell)}}   e^{-\beta \frac{c_2}{4}(\|\gamma_1\|+\|\gamma_\ell\|)}|\gamma_1|^{2 + b'} |\gamma_\ell|^{2+b'}\\
    & = 2^{\alpha}J_{x_1, x_2} C_\beta^{\ell - 2}C^2\left(\sum_{\substack{\gamma;\;  x_1 \in \widetilde{V}(\gamma)}}    e^{-\beta \frac{c_2}{4}\|\gamma\|}|\gamma|^{2 + b'}\right) \left( \sum_{\substack{\gamma'; \; x_2 \in \widetilde{V}(\gamma')}}e^{-\beta \frac{c_2}{4}\|\gamma'\|} |\gamma'|^{2+b'} \right),
\end{align*}

where $b' = \frac{d}{d-1} + \alpha(1 + \frac{ad}{d^2-1})$. Finally, 

\begin{equation*}
    \sum_{\substack{\gamma_p\\ 1 \leq p \leq \ell }} \mathbbm{1}_{\left\{\substack{\substack{\{\gamma_1, \dots, \gamma_\ell\} \in \mathcal{E}^+};\\x_1 \in \widetilde{V}(\gamma_1), x_2 \in \widetilde{V}(\gamma_\ell)}\right\}}  \left(\prod_{p =1}^\ell e^{-\beta \frac{c_2}{4}\|\gamma_p\|}\prod_{p = 1}^{\ell - 1}F_{\gamma_p, \gamma_{p+1}}\right) \leq 2^{\alpha}J_{x_1, x_2} C_\beta^{\ell - 2}C^2c_{\beta/2}^2\leq 2^\alpha J_{x_1,x_2} (\kappa c_{\beta/2})^\ell,
\end{equation*}
where $\kappa \coloneqq \max\{2^{2\alpha - 1}C^\alpha F_{\{0\}}, C\}$ (notice that we can take $M$ large enough so that $C \geq 1$). Putting everything together again yields,

\begin{align*}
   \sum_{\substack{\Gamma; \exists \gamma, \gamma'\in \Gamma\\\gamma \neq \gamma'\\x_1 \in \widetilde{V}(\gamma), x_2 \in \widetilde{V}(\gamma')}} \widetilde{z}^+_\beta (\Gamma) 
      &\leq 2^\alpha J_{x_1,x_2}\sum_{m \geq 2} \sum_{\ell=2}^m\binom{m-1}{\ell-1}(6c_{\beta/2})^{m-\ell} (\kappa c_{\beta/2})^\ell \\
      &= 2^\alpha J_{x_1, x_2}\kappa c_{\beta/2}\sum_{m\geq 2}\left([6+\kappa] c_{\beta/2}\right)^{m-1} = 2^\alpha \frac{([6+\kappa]\kappa) c_{\beta/2}^2}{1-[6+\kappa] c_{\beta/2}}J_{x_1,x_2}.
\end{align*}
Taking $C_\beta^{(1)} \coloneqq 2\max\left\{2^\alpha  \frac{(([6+\kappa]\kappa) c_{\beta/2}^2}{1-[6+\kappa] c_{\beta/2}},\frac{e^{-\beta \frac{c_5}{2}}}{J(1-7e^{-\beta \frac{c_2}{8}})}\right\}$ ends the proof. 

\end{proof}

\begin{theorem}\label{main_decay}
For $\beta$ large enough, there exists a constant $c_4(\alpha,d,\beta)
\coloneqq c_4>0$ such that, for every $x_1,x_2\in\mathbb Z^d$,
\[
    \left| \langle \sigma_{x_1}=-1;\sigma_{x_2}=-1\rangle_{\beta}^+ \right| \leq c_4J_{x_1,x_2}.
\]
\end{theorem}

\begin{proof}
By Proposition \ref{prop_derivatives} and the tree-graph bound, it is
sufficient to bound the quantity in the right-hand side below:
\begin{equation}\label{decay_correlations_eq_1}
\begin{split}
&\sum_{m\geq 1}\frac{1}{m!}
\sum_{\substack{
    (\Gamma_1,\dots,\Gamma_m)\in(\mathcal E^+)^m\\
    \{x_1,x_2\}\subset \widetilde{V}(\Gamma_1,\dots,\Gamma_m)
}}
\left|
    \phi^T(\Gamma_1,\dots,\Gamma_m)
\right|
\prod_{i=1}^m\widetilde z_\beta^+(\Gamma_i)
\\
&\leq
\sum_{m\geq 1}\frac{1}{m!}
\sum_{\substack{
    (\Gamma_1,\dots,\Gamma_m)\in(\mathcal E^+)^m\\
    \{x_1,x_2\}\subset\widetilde V(\Gamma_1,\dots,\Gamma_m)
}}
\sum_{T\in\mathcal T_m}
\prod_{i=1}^m\widetilde z_\beta^+(\Gamma_i)
\prod_{\{i,j\}\in E(T)}
\mathbbm 1_{\Gamma_i\not\sim\Gamma_j}.
\end{split}
\end{equation}
We will apply a strategy analogous to that of Proposition
\ref{decay_proposition}. We first split the sum in the right-hand side of
\eqref{decay_correlations_eq_1} into two terms according to whether there is
a polymer containing both $x_1$ and $x_2$. Thus, write the right-hand side of
\eqref{decay_correlations_eq_1} as $S_1+S_2$, where
\[
\begin{split}
S_1
&:=
\sum_{m\geq 1}\frac{1}{m!}
\sum_{\substack{
    (\Gamma_1,\dots,\Gamma_m)\in(\mathcal E^+)^m\\
    \exists i\in\{1,\dots,m\}:
    \{x_1,x_2\}\subset\widetilde V(\Gamma_i)
}}
\sum_{T\in\mathcal T_m}
\prod_{i=1}^m\widetilde z_\beta^+(\Gamma_i)
\prod_{\{i,j\}\in E(T)}
\mathbbm 1_{\Gamma_i\not\sim\Gamma_j},
\\
S_2
&:=
\sum_{m\geq 2}\frac{1}{m!}
\sum_{\substack{
    (\Gamma_1,\dots,\Gamma_m)\in(\mathcal E^+)^m\\
    \{x_1,x_2\}\subset\widetilde V(\Gamma_1,\dots,\Gamma_m)\\
    \nexists i\in\{1,\dots,m\}:
    \{x_1,x_2\}\subset\widetilde V(\Gamma_i)
}}
\sum_{T\in\mathcal T_m}
\prod_{i=1}^m\widetilde z_\beta^+(\Gamma_i)
\prod_{\{i,j\}\in E(T)}
\mathbbm 1_{\Gamma_i\not\sim\Gamma_j}.
\end{split}
\]

By symmetry, in the definition of $S_1$ we may distinguish the label $1$ and
use the union bound over the possible labels containing both points. Thus,
rooting the trees at the vertex $1$ and proceeding as in Lemma
\ref{main_lemma}, we obtain
\[
\begin{split}
S_1
&\leq
\sum_{m\geq 1}\frac{1}{(m-1)!}
\sum_{\substack{
    (\Gamma_1,\dots,\Gamma_m)\in(\mathcal E^+)^m\\
    \{x_1,x_2\}\subset\widetilde V(\Gamma_1)
}}
\sum_{T\in\mathcal T_m}
\prod_{i=1}^m\widetilde z_\beta^+(\Gamma_i)
\prod_{\{i,j\}\in E(T)}
\mathbbm 1_{\Gamma_i\not\sim\Gamma_j}
\\
&\leq
\left(
    \sum_{r\geq 0}
    e^{-\beta\frac{c_2}{64}r}
\right)
\sum_{\{x_1,x_2\}\subset\widetilde V(\Gamma)}
\widetilde z_{\beta/2}^+(\Gamma)
\\
&\leq
\frac{C_{\beta/2}^{(1)}}{
    1-e^{-\beta\frac{c_2}{64}}
}
J_{x_1,x_2}.
\end{split}
\]

We now consider $S_2$. Since no polymer contains both $x_1$ and $x_2$, each family contributing to $S_2$ contains two distinct polymers, one containing $x_1$ and one containing $x_2$. We use the labels $1$ and $\ell$ for these two polymers and enumerate the vertices in the unique path between them in increasing order. The remaining labeled vertices are then grouped according to the vertex of the path to which they are attached. The same tree-erasing
argument used in Lemma \ref{main_lemma} gives

\begin{equation}\label{eq_1_decay_thm}
\begin{split}
 S_2\leq \sum_{m \geq 2} \sum_{\ell=2}^m\binom{m-1}{\ell-1} e^{-\beta\frac{ c_2}{64}(m-\ell)}\sum_{\substack{\Gamma_p\\ 1 \leq p \leq \ell}}\prod_{p = 1}^{\ell - 1}\mathbbm{1}_{\{\Gamma_p \not\sim \Gamma_{p + 1}\}}\mathbbm{1}_{\left\{\substack{x_1 \in \widetilde{V}(\Gamma_1), x_2 \in \widetilde{V}(\Gamma_\ell)  \\ x_1,x_2 \not\in \widetilde{V}(\Gamma_r), 2\leq r\leq \ell-1}\right\}}  \prod_{p=1}^{\ell} \widetilde{z}_{\beta/2}^+(\Gamma_p).
\end{split}
\end{equation}

Now the procedure of erasing the vertices on the remaining path must differ, again by the fact that in Proposition \ref{decay_proposition} all the contours in the path must be compatible, while in this case we only can guarantee that the polymers are incompatible to its neighbors (See Figure \ref{fig:cortando_arvore}). 

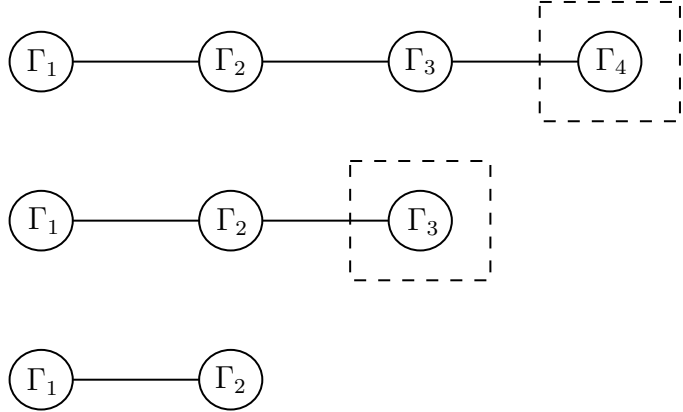
\begin{figure}[ht]
    \centering
    \tikzset{every picture/.style={line width=0.75pt}} 

\begin{tikzpicture}[x=0.75pt,y=0.75pt,yscale=-1,xscale=1]

\draw   (20,180) .. controls (20,171.72) and (27.06,165) .. (35.77,165) .. controls (44.48,165) and (51.54,171.72) .. (51.54,180) .. controls (51.54,188.28) and (44.48,195) .. (35.77,195) .. controls (27.06,195) and (20,188.28) .. (20,180) -- cycle ;
\draw    (51.54,180) -- (114.62,180) ;
\draw   (114.62,180) .. controls (114.62,171.72) and (121.68,165) .. (130.38,165) .. controls (139.09,165) and (146.15,171.72) .. (146.15,180) .. controls (146.15,188.28) and (139.09,195) .. (130.38,195) .. controls (121.68,195) and (114.62,188.28) .. (114.62,180) -- cycle ;
\draw    (146.15,180) -- (209.23,180) ;
\draw   (209.23,180) .. controls (209.23,171.72) and (216.29,165) .. (225,165) .. controls (233.71,165) and (240.77,171.72) .. (240.77,180) .. controls (240.77,188.28) and (233.71,195) .. (225,195) .. controls (216.29,195) and (209.23,188.28) .. (209.23,180) -- cycle ;
\draw    (240.77,180) -- (303.85,180) ;
\draw   (303.85,180) .. controls (303.85,171.72) and (310.91,165) .. (319.62,165) .. controls (328.32,165) and (335.38,171.72) .. (335.38,180) .. controls (335.38,188.28) and (328.32,195) .. (319.62,195) .. controls (310.91,195) and (303.85,188.28) .. (303.85,180) -- cycle ;
\draw  [dash pattern={on 4.5pt off 4.5pt}] (284.62,150) -- (354.62,150) -- (354.62,210) -- (284.62,210) -- cycle ;
\draw   (20,260) .. controls (20,251.72) and (27.06,245) .. (35.77,245) .. controls (44.48,245) and (51.54,251.72) .. (51.54,260) .. controls (51.54,268.28) and (44.48,275) .. (35.77,275) .. controls (27.06,275) and (20,268.28) .. (20,260) -- cycle ;
\draw    (51.54,260) -- (114.62,260) ;
\draw   (114.62,260) .. controls (114.62,251.72) and (121.68,245) .. (130.38,245) .. controls (139.09,245) and (146.15,251.72) .. (146.15,260) .. controls (146.15,268.28) and (139.09,275) .. (130.38,275) .. controls (121.68,275) and (114.62,268.28) .. (114.62,260) -- cycle ;
\draw    (146.15,260) -- (209.23,260) ;
\draw   (209.23,260) .. controls (209.23,251.72) and (216.29,245) .. (225,245) .. controls (233.71,245) and (240.77,251.72) .. (240.77,260) .. controls (240.77,268.28) and (233.71,275) .. (225,275) .. controls (216.29,275) and (209.23,268.28) .. (209.23,260) -- cycle ;
\draw  [dash pattern={on 4.5pt off 4.5pt}] (190,230) -- (260,230) -- (260,290) -- (190,290) -- cycle ;
\draw   (20,340) .. controls (20,331.72) and (27.06,325) .. (35.77,325) .. controls (44.48,325) and (51.54,331.72) .. (51.54,340) .. controls (51.54,348.28) and (44.48,355) .. (35.77,355) .. controls (27.06,355) and (20,348.28) .. (20,340) -- cycle ;
\draw    (51.54,340) -- (114.62,340) ;
\draw   (114.62,340) .. controls (114.62,331.72) and (121.68,325) .. (130.38,325) .. controls (139.09,325) and (146.15,331.72) .. (146.15,340) .. controls (146.15,348.28) and (139.09,355) .. (130.38,355) .. controls (121.68,355) and (114.62,348.28) .. (114.62,340) -- cycle ;

\draw (27,172.07) node [anchor=north west][inner sep=0.75pt]    {$\Gamma _{1}$};
\draw (120.88,171.53) node [anchor=north west][inner sep=0.75pt]    {$\Gamma _{2}$};
\draw (215.67,171.4) node [anchor=north west][inner sep=0.75pt]    {$\Gamma _{3}$};
\draw (311,171.4) node [anchor=north west][inner sep=0.75pt]    {$\Gamma _{4}$};
\draw (27.67,251.4) node [anchor=north west][inner sep=0.75pt]    {$\Gamma _{1}$};
\draw (121.55,252.2) node [anchor=north west][inner sep=0.75pt]    {$\Gamma _{2}$};
\draw (217,252.07) node [anchor=north west][inner sep=0.75pt]    {$\Gamma _{3}$};
\draw (27,332.73) node [anchor=north west][inner sep=0.75pt]    {$\Gamma _{1}$};
\draw (120.88,331.53) node [anchor=north west][inner sep=0.75pt]    {$\Gamma _{2}$};

\end{tikzpicture}
    \caption{The procedure of erasing the vertices in the path between $\Gamma_1$ and $\Gamma_\ell$. In this figure we have $\ell = 4$.}
    \label{fig:cortando_arvore}
\end{figure}

We claim that, if $\Gamma_{\ell-1} \not\sim \Gamma_\ell$, then there must exist some $x_3 \in \widetilde{V}(\Gamma_\ell)$ with $x_3 \neq x_2$ such that $\dis(x_3,\widetilde{V}(\Gamma_{\ell-1})) \leq 1 + M |V(\Gamma_{\ell - 1})|^{\frac{a}{d+1}}$. Indeed, since they are incompatible, there must exist some $x_0 \in \widetilde{V}(\Gamma_\ell)$ with $\dis(x_0,\widetilde{V}(\Gamma_{\ell-1}))=\dis(\widetilde{V}(\Gamma_{\ell}),\widetilde{V}(\Gamma_{\ell-1}))\leq M |V(\Gamma_{\ell-1})|^{\frac{a}{d+1}}$. If $x_0 \neq x_2$, just choose $x_3=x_0$ and we are done. Suppose that $x_0=x_2$. The claim will follow from the fact that for any $x \in \widetilde{V}(\Gamma)$, there must exist some $ y \in \widetilde{V}(\Gamma)$ with $|x-y|=1$. Indeed, if $x \in \Sp(\Gamma)$, there must exist some neighboring point $y$ with a different spin since $x$ is incorrect. But then, $y$ is also incorrect and must be inside $\Sp(\Gamma)$. Now assume that $x \in \I_-(\Gamma)$ and that $B_1(x) \cap \I_-(\Gamma) = \{x\}$. Then, every neighboring point of $x$ is incorrect; therefore, they are inside $\Sp(\Gamma)$. Taking $x = x_2$ and $y = x_3$, $\dis(x_3, \widetilde{V}(\Gamma_{\ell - 1})) \leq |x_3 - x_2| + \dis(x_2, \widetilde{V}(\Gamma_{\ell - 1})) \leq 1 + M|V(\Gamma_{\ell - 1})|^{\frac{a}{d+1}}$. Using this fact, we have

\begin{equation}\label{420}
\begin{split}
&\sum_{\substack{\Gamma_p\\ 1 \leq p \leq \ell}}\prod_{p = 1}^{\ell - 1}\mathbbm{1}_{\{\Gamma_p \not\sim \Gamma_{p + 1}\}}\mathbbm{1}_{\left\{\substack{x_1 \in \widetilde{V}(\Gamma_1), x_2 \in \widetilde{V}(\Gamma_\ell)  \\ x_1,x_2 \not\in \widetilde{V}(\Gamma_r), 2\leq r\leq \ell-1}\right\}} \prod_{p=1}^\ell\widetilde{z}_{\beta/2}^+(\Gamma_p)\\
&=  \sum_{\substack{\Gamma_p\\ 1 \leq p \leq \ell-1}}\prod_{p = 1}^{\ell - 2}\mathbbm{1}_{\{\Gamma_p \not\sim \Gamma_{p + 1}\}}\mathbbm{1}_{\left\{\substack{x_1 \in \widetilde{V}(\Gamma_1) \\ x_1,x_2 \not\in \widetilde{V}(\Gamma_r), 2\leq r\leq \ell-1}\right\}} \prod_{p=1}^{\ell-1}\widetilde{z}_{\beta/2}^+(\Gamma_p)\sum_{\substack{\Gamma_{\ell}\not\sim \Gamma_{\ell-1} \\ x_2\in \widetilde{V}(\Gamma_\ell)}}\widetilde{z}^+_{\beta/2}(\Gamma_\ell) \\ 
&\leq \sum_{\substack{\Gamma_p\\ 1 \leq p \leq \ell-1}}\prod_{p = 1}^{\ell - 2}\mathbbm{1}_{\{\Gamma_p \not\sim \Gamma_{p + 1}\}}\mathbbm{1}_{\left\{\substack{x_1 \in \widetilde{V}(\Gamma_1)  \\ x_1,x_2 \not\in \widetilde{V}(\Gamma_r), 2\leq r\leq \ell-1}\right\}} \prod_{p=1}^{\ell-1}\widetilde{z}_{\beta/2}^+(\Gamma_p)\hspace{-1.0cm}\sum_{\substack{x_3;\ x_3 \neq x_2 \\ \dis(x_3,\widetilde{V}(\Gamma_{\ell-1}))\leq 1 + M|V(\Gamma_{\ell-1})|^{\frac{a}{d+1}}}}\sum_{\widetilde{V}(\Gamma_\ell)\supset \{x_2,x_3\}}\widetilde{z}^+_{\beta/2}(\Gamma_\ell) \\
&\leq C_{\beta/2}^{(1)}\sum_{\substack{\Gamma_p\\ 1 \leq p \leq \ell-1}}\prod_{p = 1}^{\ell - 2}\mathbbm{1}_{\{\Gamma_p \not\sim \Gamma_{p + 1}\}}\mathbbm{1}_{\left\{\substack{x_1 \in \widetilde{V}(\Gamma_1)  \\ x_1,x_2 \not\in \widetilde{V}(\Gamma_r), 2\leq r\leq \ell-1}\right\}} \prod_{p=1}^{\ell-1}\widetilde{z}_{\beta/2}^+(\Gamma_p)\hspace{-1.0cm}\sum_{\substack{x_3;\ x_3 \neq x_2 \\ \dis(x_3,\widetilde{V}(\Gamma_{\ell-1}))\leq 1 + M|V(\Gamma_{\ell-1})|^{\frac{a}{d+1}}}}\hspace{-1.5cm}J_{x_2,x_3},
\end{split}
\end{equation}
where the last inequality is due to Proposition \ref{decay_proposition}. Notice that
\[
\frac{\dis(\widetilde{V}(\Gamma_{\ell-1}),x_2)}{|x_2-x_3|}\leq \frac{\dis(\widetilde{V}(\Gamma_{\ell-1}),x_3)}{|x_2-x_3|}+1\leq 2M |V(\Gamma_{\ell-1})|^{\frac{a}{d+1}}.
\]
 By multiplying and dividing by $\dis(\widetilde{V}(\Gamma_{\ell-1}),x_2)^\alpha$ we get

\begin{equation*}
\begin{split}
\sum_{\substack{x_3;\ x_3 \neq x_2 \\ \dis(x_3,\widetilde{V}(\Gamma_{\ell-1}))< 1 + M|V(\Gamma_{\ell-1})|^{\frac{a}{d+1}}}}\hspace{-1.5cm}J_{x_2,x_3}\leq \frac{J(2M)^\alpha(|V(\Gamma_{\ell-1})|)^{\frac{a\alpha}{d+1}}(3M|V(\Gamma_{\ell-1})|)^{\frac{ad}{d+1}}|\widetilde{V}(\Gamma_{\ell-1})|}{\dis(x_2,\widetilde{V}(\Gamma_{\ell-1}))^\alpha}\leq \frac{J (3M)^{\alpha+\frac{ad}{d+1}}|\Gamma_{\ell-1}|^{b_2}}{\dis(x_2,\widetilde{V}(\Gamma_{\ell-1}))^\alpha},
\end{split}
\end{equation*}
where $b_2 \coloneqq \frac{a\alpha d + ad^2 + d(d+1)}{d^2-1}$. We get a lower bound for \eqref{420}, given by

\begin{equation*}
\begin{split}
&J C^{(1)}_{\beta/2}(3M)^{\alpha+\frac{ad}{d+1}}\sum_{\substack{\Gamma_p\\ 1 \leq p \leq \ell-2}}\prod_{p = 1}^{\ell - 3}\mathbbm{1}_{\{\Gamma_p \not\sim \Gamma_{p + 1}\}}\mathbbm{1}_{\left\{\substack{x_1 \in \widetilde{V}(\Gamma_1) \\ x_1,x_2 \not\in \widetilde{V}(\Gamma_r), 2\leq r\leq \ell-2}\right\}} \prod_{p=1}^{\ell-2}\widetilde{z}_{\beta/2}^+(\Gamma_p) \times \\
& \hspace{7.5cm}\times \left[ \sum_{\substack{\Gamma_{\ell - 1};\\ \Gamma_{\ell-1}\not\sim \Gamma_{\ell-2}}}\mathbbm{1}_{\{x_1,x_2 \not\in \widetilde{V}(\Gamma_{\ell-1})\}}\widetilde{z}^+_{\beta/2}(\Gamma_{\ell-1}) \frac{|\Gamma_{\ell-1}|^{b_2}}{\dis(x_2,\widetilde{V}(\Gamma_{\ell-1}))^\alpha}\right].
\end{split}
\end{equation*}

We can bound the expression between square brackets by

\begin{equation*}
\begin{split}
     & \sum_{x_3 \in \Z^d\setminus\{x_2\}}\frac{1}{|x_2-x_3|^\alpha}\sum_{\substack{\Gamma_{\ell-2}\not\sim \Gamma_{\ell-1} \\ x_3 \in \widetilde{V}(\Gamma_{\ell-1})\\ d(x_2, \widetilde{V}(\Gamma_{\ell - 1})) = |x_2 - x_3|}}\mathbbm{1}_{\{x_2 \not\in \widetilde{V}(\Gamma_{\ell-1})\}}\widetilde{z}_{\beta/2}^+(\Gamma_{\ell-1})|\Gamma_{\ell-1}|^{b_2}\\
    &\leq \sum_{x_3 \in \Z^d\setminus\{x_2\}}\frac{1}{|x_2-x_3|^\alpha}\sum_{\substack{\Gamma_{\ell-2}\not\sim \Gamma_{\ell-1} \\ x_3 \in \widetilde{V}(\Gamma_{\ell-1})}}\mathbbm{1}_{\{x_2 \not\in \widetilde{V}(\Gamma_{\ell-1})\}}\widetilde{z}_{\beta/4}^+(\Gamma_{\ell-1}) \\
    &\leq \sum_{x_3 \in \Z^d\setminus\{x_2\}}\frac{1}{|x_2-x_3|^\alpha}\hspace{-1.0cm}\sum_{\substack{x_4;\ x_4 \neq x_3 \\ \dis(x_4,\widetilde{V}(\Gamma_{\ell-2}))< 1 + M|V(\Gamma_{\ell-2})|^{\frac{a}{d+1}}}}\sum_{\widetilde{V}(\Gamma_{\ell-1})\supset \{x_3,x_4\}}\mathbbm{1}_{\{x_2 \not\in \widetilde{V}(\Gamma_{\ell-1})\}}\widetilde{z}_{\beta/4}^+(\Gamma_{\ell-1}) \\
    &\leq C^{(1)}_{\beta/4}\sum_{x_3 \in \Z^d\setminus\{x_2\}}\frac{1}{|x_2-x_3|^\alpha}\sum_{\substack{x_4;\ x_4 \neq x_2,x_3 \\ \dis(x_4,\widetilde{V}(\Gamma_{\ell-2}))\leq 1 + M|V(\Gamma_{\ell-2})|^{\frac{a}{d+1}}}}\hspace{-1.5cm}J_{x_3,x_4}.
\end{split}
\end{equation*}

where in the first inequality we need to take $\beta > 8b_2/c_2$. The triangle inequality yields 
 \[
 \frac{J_{x_3,x_4}}{|x_2-x_3|^\alpha}\leq J_{x_2,x_4}\left(\frac{1}{|x_2-x_3|}+\frac{1}{|x_3-x_4|}\right)^\alpha,
 \]
 and, together with the reverse H\"{o}lder inequality, we get
\begin{equation}\label{eq_4_decay_thm}
\begin{split}
 \sum_{\substack{x_3, x_4: \\  x_3\neq x_2, x_4\notin \{x_2,x_3\} \\ \dis(x_4,\widetilde{V}(\Gamma_{\ell-2}))\leq 1 + M|V(\Gamma_{\ell-2})|^{\frac{a}{d+1}}}}\hspace{-1.5cm}\frac{J_{x_3,x_4}}{|x_2-x_3|^\alpha}\ \hspace{0.5cm}&\leq 2^\alpha J^{-1} F_{\{0\}}\hspace{-1.5cm}\sum_{\substack{x_4: x_4\neq x_2 \\ \dis(x_4,\widetilde{V}(\Gamma_{\ell-2}))\leq1+M|V(\Gamma_{\ell-2})|^{\frac{a}{d+1}}}}\hspace{-1.5cm}J_{x_2,x_4}.
\end{split}
\end{equation}

Notice that, putting $\kappa^{(1)} \coloneqq 2^\alpha F_{\{0\}} (3M)^{\alpha+ \frac{ad}{d+1}}$ we have proven

\begin{equation*}
    \sum_{\substack{\Gamma_{p}\\ \Gamma_{p} \not\sim \Gamma_{p-1}}} \mathbbm{1}_{\{x_1, x_2 \notin \widetilde{V}(\Gamma_{p})\}}\widetilde{z}^+_{\beta/2}(\Gamma_{p})\hspace{-1cm}\sum_{\substack{x_3;\ x_3 \neq x_2 \\ d(x_3,\widetilde{V}(\Gamma_{p}))\leq 1 + M|V(\Gamma_{p})|^{\frac{a}{d+1}}}}\hspace{-1.5cm}J_{x_2,x_3} \hspace{0.2cm} \leq \hspace{0.3cm}\kappa^{(1)}C^{(1)}_{\beta/4}\hspace{-1cm}\sum_{\substack{x_3;\ x_3 \neq x_2 \\ d(x_3,\widetilde{V}(\Gamma_{p-1}))\leq 1 + M|V(\Gamma_{p-1})|^{\frac{a}{d+1}}}}\hspace{-1.5cm}J_{x_2,x_3} 
\end{equation*}

for the particular case where $p = \ell - 1$. Iterating this inequality,

\begin{equation}\label{eq_5_decay_thm}
\begin{split}
&\sum_{\substack{\Gamma_p \\ 1 \leq p \leq \ell}} \prod_{p = 1}^{\ell - 1} \mathbbm{1}_{\{\Gamma_p \not\sim \Gamma_{p-1}\}}\mathbbm{1}_{\left\{\substack{x_1 \in \widetilde{V}(\Gamma_1), x_2 \in \widetilde{V}(\Gamma_\ell)  \\ x_1,x_2 \not\in \widetilde{V}(\Gamma_r), 2\leq r\leq \ell-1}\right\}} \prod_{p=1}^{\ell-2}\widetilde{z}_{\beta/2}^+(\Gamma_p)\\[0.2cm]
&\hspace{4cm}\leq (\kappa^{(1)})^{\ell-2} (C^{(1)}_{\beta/4})^{\ell-1}\hspace{-0.5cm} \sum_{\substack{ \Gamma_1 \\ x_1 \in \widetilde{V}(\Gamma_1), x_2 \not \in \widetilde{V}(\Gamma_1)}}\hspace{-0.5cm}\widetilde{z}_{\beta/2}^+(\Gamma_1)\sum_{\substack{x_3: x_3\neq x_2 \\ \dis(x_3,\widetilde{V}(\Gamma_1))<1+M|V(\Gamma_1)|^{\frac{a}{d+1}}}}\hspace{-1.5cm}J_{x_2,x_3}.
\end{split}
\end{equation}

We can adapt the steps above to erase the last vertex and we also get

\begin{equation*}
    \sum_{\substack{ \Gamma_1 \\ x_1 \in \widetilde{V}(\Gamma_1), x_2 \not \in \widetilde{V}(\Gamma_1)}}\hspace{-0.5cm}\widetilde{z}_{\beta/2}^+(\Gamma_1)\sum_{\substack{x_3: x_3\neq x_2 \\ \dis(x_3,\widetilde{V}(\Gamma_1))<1+M|V(\Gamma_1)|^{\frac{a}{d+1}}}}\hspace{-1.5cm}J_{x_2,x_3} \leq  (\kappa^{(1)}) (C^{(1)}_{\beta/4})J_{x_1,x_2},
\end{equation*}

so finally,

\begin{equation}\label{eq_6_decay_thm}
\begin{split}
\sum_{\substack{\Gamma_p \\ 1 \leq p \leq \ell}} \prod_{p = 1}^{\ell - 1} \mathbbm{1}_{\{\Gamma_p \not\sim \Gamma_{p-1}\}}\mathbbm{1}_{\left\{\substack{x_1 \in \widetilde{V}(\Gamma_1), x_2 \in \widetilde{V}(\Gamma_\ell)  \\ x_1,x_2 \not\in \widetilde{V}(\Gamma_r), 2\leq r\leq \ell-1}\right\}} \prod_{p=1}^{\ell-2}\widetilde{z}_{\beta/2}^+(\Gamma_p)\leq  (\kappa^{(1)})^{\ell-1} (C^{(1)}_{\beta/4})^{\ell} J_{x_1, x_2}.
\end{split}
\end{equation}

Plugging it back to Equation \eqref{eq_1_decay_thm}, we have

\begin{equation*}
\begin{split}
S_2 &\leq C^{(1)}_{\beta/4}J_{x_1 x_2}\sum_{m \geq 2} \sum_{\ell=2}^m\binom{m-1}{\ell-1}  e^{-\beta\frac{ c_2}{64}(m-\ell)}(\kappa^{(1)}C_{\beta/4}^{(1)})^{\ell-1} \\ 
&\leq C^{(1)}_{\beta/4}J_{x_1 x_2}\sum_{m \geq 2} (e^{-\beta\frac{ c_2}{64}}+\kappa^{(1)}C_{\beta/4}^{(1)})^{m-1}
\\&\leq \frac{C_{\beta/4}^{(1)}(e^{-\beta\frac{ c_2}{64}}+\kappa^{(1)}C_{\beta/4}^{(1)})}{1-e^{-\beta\frac{ c_2}{64}}-\kappa^{(1)}C_{\beta/4}^{(1)}}J_{x_1,x_2}.
\end{split}
\end{equation*}
The proof is concluded using Proposition \ref{prop_derivatives} and choosing
\[
c_4 \coloneqq 2(24)^2C_{\beta/4}^{(1)}\max\left\{\frac{e^{-\beta\frac{ c_2}{64}}+\kappa^{(1)}C_{\beta/4}^{(1)}}{1-e^{-\beta\frac{ c_2}{64}}-\kappa^{(1)}C_{\beta/4}^{(1)}}, c_{\beta/4}\right\}.
\]

\end{proof}

For every local function $f: \Omega \rightarrow \mathbb{R}$ there exists $\Lambda \Subset \mathbb{Z}^d$ and real numbers $\{f_A\}_{A\subset \Lambda}$ such that
\begin{equation}\label{local_function}
f = \sum_{A\subset \Lambda} f_A n_A,
\end{equation}
where $n_A\coloneqq \prod_{x\in A}n_x$, and $n_x \coloneqq \frac{1+\sigma_x}{2}$ are the occupation variables (see Lemma 3.19, Chapter 3 of \cite{FV-Book}). There exists also a smallest $\Lambda$, with respect to the inclusion, such that \ref{local_function} holds. This set is called the \textit{support of} $f$ and denoted by $\supp f$.

\begin{proof}[Proof of Corollary 1.1.:]
  In \cite{Lebowitz1972}, Lebowitz proves that, for any $A,B$ finite subsets of $\Z^d$, it holds
  \[
  \langle n_A;n_B\rangle^+_\beta \leq \sum_{\substack{x\in A \\ y \in B}}\langle n_x;n_y\rangle_\beta^+.
  \]
  Thus, the above inequality together with the expansion \eqref{local_function} we get
\[
\langle f;g\rangle^+_\beta = \sum_{\substack{A\subset \supp(f) \\ B \subset \supp(g)}}f_{A}g_{B}\langle n_A;n_B\rangle^+ \leq \|f\|\|g\|\sum_{\substack{A\subset \supp(f) \\ B \subset \supp(g)}}\sum_{\substack{x\in A \\ y \in B}} \langle n_x;n_y\rangle^+,
\]
where $\displaystyle \|f\|\coloneqq \max_{A \subset \supp(f)}|f_A|$ and similarly for $g$. Since the two-point correlation function is bilinear, our main result yields that 
\[
 \langle n_x,n_y\rangle^+ \leq \frac{J c_4}{\dis(\supp(f), \supp(g))^\alpha}.
\]  
Since $\displaystyle \sum_{A \subset \supp(f)}|A| = |\supp(f)|2^{|\supp(f)|-1}$, taking $C_{f,g} \coloneqq c_4|\supp(f)||\supp(g)|\|f\|\|g\|2^{|\supp(f)|+|\supp(g)|-2}$ yields the desired result.
\end{proof}


\section{Concluding Remarks}

In this paper, we prove the convergence of the cluster expansion at low temperatures for the long-range Ising models in $d\geq 2$ and regular interactions ($\alpha >d$) by using a recent notion of multidimensional multi-scaled contours introduced by three of us in \cite{Johanes}. These contours are multidimensional objects; the definition is a refinement of a previous one proposed in \cite{Affonso2024} where the diameter is used as defined by Fr\"ohlich-Spencer for one-dimensional long-range Ising models \cite{Frohlich.Spencer.82}, instead of our definition where we use the volume. The convergence of the cluster expansion has many consequences in statistical mechanics; we gave an application proving the polynomial decay of the truncated two-point correlation functions at low temperatures, showing that the exponent coincides with the exponent of the interaction. Our definition and methods also work for multidimensional Potts and other models, but we decided to present these results in a separate paper \cite{affonso2024phasetransitionferromagneticqstate}, where we proved the phase transition. Also, it is expected that the convergence of the cluster expansion and the control of truncated correlations can be obtained for the long-range Potts model as well.

\section*{Acknowledgements}

This study was financed, in part, by the São Paulo Research Foundation (FAPESP), Brazil, processes numbers 2016/25053-8, 2017/18152-2, 2018/26698-8, 2020/14563-0, 2022/00746-1, 2022/00748-4, and 2023/00854-1. RB is supported by CNPq grants 311658/2025-3, 312294/2018-2, and 408851/2018-0. The authors are very grateful to Titti Merola and Pierre Picco for helpful comments about the one-dimensional cluster expansion for long-range systems. RB thanks Aernout van Enter for the many insightful discussions on statistical mechanics over the years, which profoundly influenced his way of thinking. RB also thanks Yacine Aoun and Roberto Fernández for the information about the literature and results concerning the decay of correlations for long-range interactions. The authors thank Thiago Raszeja for his help with the pictures and for his friendship over the years.
\bibliographystyle{habbrv}
\bibliography{refs}

@article{Niss2004,
  doi = {10.1007/s00407-004-0088-3},
  url = {https://doi.org/10.1007/s00407-004-0088-3},
  year = {2004},
  publisher = {Springer Science and Business Media {LLC}},
  volume = {59},
  number = {3},
  pages = {267--318},
  author = {Martin Niss},
  title = {History of the {L}enz-Ising {M}odel 1920-1950: {F}rom {F}erromagnetic to {C}ooperative {P}henomena},
  journal = {Archive for History of Exact Sciences}
}

@misc{johanes_china,
  doi = {10.48550/ARXIV.2412.19281},
  url = {https://arxiv.org/abs/2412.19281},
  author = {Ding,  Jian and Huang,  Fenglin and Maia,  João},
  title = {Phase transitions in low-dimensional long-range random field Ising models},
  year={2024},
  eprint={2412.19281},
  archivePrefix={arXiv}
}

@incollection{Kotecky2006Cluster,
    author    = {Koteck{\'y}, Roman},
    title     = {Cluster Expansions},
    booktitle = {Encyclopedia of Mathematical Physics},
    volume    = {1},
    pages     = {531--536},
    editor    = {Fran{\c{c}}oise, Jean-Pierre and Naber, Gregory L. and Tsou, Tsou Sheung},
    publisher = {Elsevier},
    address   = {Oxford},
    year      = {2006},
    isbn      = {978-0-1251-2666-3}
}

@article{Affonso2024,
  title = {Long-range Ising models: Contours,  phase transitions and decaying fields},
  volume = {27},
  ISSN = {1435-9863},
  url = {http://dx.doi.org/10.4171/JEMS/1529},
  DOI = {10.4171/jems/1529},
  journal = {Journal of the European Mathematical Society},
  publisher = {European Mathematical Society - EMS - Publishing House GmbH},
  author = {Affonso,  Lucas and Bissacot,  Rodrigo and Endo,  Eric O. and Handa,  Satoshi},
  year = {2025},
  number = {4},
  pages = {1679–1714}
}

@unpublished{Bissacot_Corsini_cluster,
  author    = {R. Bissacot and H. Corsini},
  title     = {A Cluster Expansion and the Decay of Correlations of the 1D Long-Range Ising Model at Low Temperatures},
  year={2026},
  eprint={2602.12447},
  archivePrefix={arXiv},
  primaryClass={math-ph}
}

@article{Dyson1969,
  title = {Existence of a phase-transition in a one-dimensional Ising ferromagnet},
  volume = {12},
  ISSN = {1432-0916},
  url = {http://dx.doi.org/10.1007/BF01645907},
  DOI = {10.1007/bf01645907},
  number = {2},
  journal = {Communications in Mathematical Physics},
  publisher = {Springer Science and Business Media LLC},
  author = {Dyson,  Freeman J.},
  year = {1969},
  pages = {91–107}
}

@article{Ginibre1966,
  title = {Condensation of lattice gases},
  volume = {3},
  ISSN = {1432-0916},
  url = {http://dx.doi.org/10.1007/BF01645411},
  DOI = {10.1007/bf01645411},
  number = {3},
  journal = {Communications in Mathematical Physics},
  publisher = {Springer Science and Business Media LLC},
  author = {Ginibre,  J. and Grossmann,  A. and Ruelle,  D.},
  year = {1966},
  pages = {187–193}
}

@article{Malyshev1980,
  title = {Cluster Expansions in Lattice Models of Statistical Physics and the Quantum Theory of Fields},
  volume = {35},
  ISSN = {1468-4829},
  url = {http://dx.doi.org/10.1070/RM1980v035n02ABEH001622},
  DOI = {10.1070/rm1980v035n02abeh001622},
  number = {2},
  journal = {Russian Mathematical Surveys},
  publisher = {Steklov Mathematical Institute},
  author = {Malyshev,  V A},
  year = {1980},
  pages = {1–62}
}

@article{Biskup_Chayes_Kivelson_07,
  author  = {M. Biskup and L. Chayes and S. A. Kivelson},
  journal = {Communications in Mathematical Physics},
  pages   = {217-231},
  title   = {On the absence of ferromagnetism in typical 2D ferromagnets},
  volume  = {274},
  year    = {2007}
}

@article{Pirogov1976,
  title = {Phase diagrams of classical lattice systems continuation},
  volume = {26},
  DOI = {10.1007/bf01038255},
  number = {1},
  journal = {Theoretical and Mathematical Physics},
  publisher = {Springer Science and Business Media LLC},
  author = {Pirogov,  S. A. and Sinai,  Ya. G.},
  year = {1976},
  pages = {39–49}
}

@article{Ising1925,
  title = {Beitrag zur Theorie des Ferromagnetismus},
  volume = {31},
  ISSN = {0044-3328},
  url = {http://dx.doi.org/10.1007/BF02980577},
  DOI = {10.1007/bf02980577},
  number = {1},
  journal = {Zeitschrift f\"{u}r Physik},
  publisher = {Springer Science and Business Media LLC},
  author = {Ising,  Ernst},
  year = {1925},
  pages = {253–258}
}

@article{Klske2025,
  title = {The Ising model: highlights and perspectives},
  volume = {28},
  ISSN = {1572-9656},
  url = {http://dx.doi.org/10.1007/s11040-025-09515-1},
  DOI = {10.1007/s11040-025-09515-1},
  number = {3},
  journal = {Mathematical Physics,  Analysis and Geometry},
  publisher = {Springer Science and Business Media LLC},
  author = {K\"{u}lske,  Christof},
  year = {2025}
}

@article{PhysRevE.89.062120,
  title = {Relations between short-range and long-range Ising models},
  author = {Angelini, Maria Chiara and Parisi, Giorgio and Ricci-Tersenghi, Federico},
  journal = {Phys. Rev. E},
  volume = {89},
  issue = {6},
  pages = {062120},
  numpages = {11},
  year = {2014},
  publisher = {American Physical Society},
  doi = {10.1103/PhysRevE.89.062120},
  url = {https://link.aps.org/doi/10.1103/PhysRevE.89.062120}
}

@misc{affonso2024phasetransitionferromagneticqstate,
      title={Phase Transition in Long-Range $q-$state Models via Contours. Clock and Potts Models with Fields},
      author={Lucas Affonso and Rodrigo Bissacot and Gilberto Faria and Kelvyn Welsch},
      year={2025},
      eprint={2410.01234},
      archivePrefix={arXiv},
      primaryClass={math-ph},
      url={https://arxiv.org/abs/2410.01234}, 
}

@misc{Johanes,
  doi = {10.48550/arXiv.2307.14150},
  url = {https://arxiv.org/abs/2307.14150},
  author = {Affonso,  Lucas and Bissacot,  Rodrigo and Maia, Jo{\~a}o},
  title = {Phase Transitions in Multidimensional Long-Range Random Field Ising Models},
  publisher = {arXiv},
  year = {2023},
  eprint={2403.04921},
  archivePrefix={arXiv},
  copyright = {arXiv.org perpetual,  non-exclusive license}
}

@misc{maia2024phase,
  title         = {Phase Transitions in Ising models: the Semi-infinite with decaying field and the Random Field Long-range. Ph.D. Thesis - {U}niversity of {S}\~ao {P}aulo},
  author        = {João Maia},
  year          = {2024},
  eprint        = {2403.04921},
  archiveprefix = {arXiv},
  primaryClass={math-ph}
}

@misc{Pereira,
  title         = {Multidimensional Contours \`a la Fr{\"o}hlich-Spencer and Boundary Conditions for Quantum Spin Systems. Ph.D. Thesis - University of S\~ao Paulo},
  author        = {Lucas Affonso},
  year          = {2023},
  eprint        = {2310.07946},
  archiveprefix = {arXiv},
  primaryClass={math-ph}
}

@article{Bissacot2010,
  title = {On the {C}onvergence of {C}luster {E}xpansions for {P}olymer {G}ases},
  volume = {139},
  ISSN = {1572-9613},
  url = {http://dx.doi.org/10.1007/s10955-010-9956-1},
  DOI = {10.1007/s10955-010-9956-1},
  number = {4},
  journal = {Journal of Statistical Physics},
  publisher = {Springer Science and Business Media LLC},
  author = {Bissacot,  Rodrigo and Fernández,  Roberto and Procacci,  Aldo},
  year = {2010},
  pages = {598-617}
}

@article{Temmel2014,
  title = {Sufficient {C}onditions for {U}niform {B}ounds in {A}bstract {P}olymer {S}ystems and {E}xplorative {P}artition {S}chemes},
  volume = {157},
  ISSN = {1572-9613},
  url = {http://dx.doi.org/10.1007/s10955-014-1108-6},
  DOI = {10.1007/s10955-014-1108-6},
  number = {6},
  journal = {Journal of Statistical Physics},
  publisher = {Springer Science and Business Media LLC},
  author = {Temmel,  Christoph},
  year = {2014},
  pages = {1225-1254}
}

@article{Duminil_Copin_2019,
   title={Exponential Decay of Truncated Correlations for the Ising Model in any Dimension for all but the Critical Temperature},
   volume={374},
   number={2},
   journal={Communications in Mathematical Physics},
   author={Duminil-Copin, H and Goswami, S. and Raoufi, A.},
   year={2019},
   pages={891–921}
}

@article{G,
author = {Griffiths, R. B. },
title = {Correlations in Ising Ferromagnets. I},
journal = {Journal of Mathematical Physics},
volume = {8},
number = {3},
pages = {478-483},
year = {1967}
}

@article{FS81,
author = {Fr{\"o}hlich, J{\"u}rg and Spencer, Thomas},
year = {1981},
pages = {527-602},
title = {The Kosterlitz-Thouless transition in two-dimensional Abelian spin systems and the Coulomb gas},
volume = {81},
journal = {Communications in Mathematical Physics},
}

@article{Bricmont.Kupiainen.88,
author = {Bricmont, Jean and Kupiainen, Antti},
year = {1988},
pages = {539-572},
title = {Phase transition in the 3d random field Ising model},
volume = {116},
journal = {Communications in Mathematical Physics},
}

@article{Cassandro.Picco.09,
  title={Phase Transition in the 1d Random Field Ising Model with long range interaction},
  author={Cassandro, Marzio and Orlandi, Enza and Picco, Pierre},
  journal={Communications in Mathematical Physics},
  volume={288},
  pages={731--744},
  year={2009},
}

@article{Pirogov.Sinai.75,
  title={Phase diagrams of classical lattice systems},
  author={Pirogov, S. A. and Sinai, Y. G.},
  journal={Teoreticheskaya i Matematicheskaya Fizika},
  volume={25},
  pages={358--369},
  year={1975},
}

@article{Zahradnik.84,
  title={An alternate version of Pirogov-Sinai theory},
  author={Zahradn{\'\i}k, M.},
  journal={Communications in Mathematical Physics},
  volume={93},
  pages={559--581},
  year={1984},
}

@article{Park.88.I,
  title={Extension of Pirogov-Sinai theory of phase transitions to infinite range interactions I. Cluster expansion},
  author={Park, Y. M.},
  journal={Communications in Mathematical Physics},
  volume={114},
  pages={187--218},
  year={1988},
}

@article{Park.88.II,
  title={Extension of Pirogov-Sinai theory of phase transitions to infinite range interactions. II. Phase diagram},
  author={Park, Y. M.},
  journal={Communications in Mathematical Physics},
  volume={114},
  pages={219--241},
  year={1988},
}

@article{Frohlich.Spencer.82,
  title={The Phase Transition in the one-dimensional Ising model with $1/r^2$ interaction energy},
  author={Fr{\"o}hlich, J. and Spencer, T.},
  journal={Communications in Mathematical Physics},
  volume={84},
  pages={87--101},
  year={1982},
}

@article{Cassandro.Merola.Picco.Rozikov.14,
  title={One-dimensional Ising models with long range interactions: cluster expansion, phase-separating point},
  author={Cassandro, M. and Merola, I. and Picco, P. and Rozikov, U.},
  journal={Communications in Mathematical Physics},
  volume={327},
  pages={951--991},
  year={2014},
}

@article{Cassandro.Merola.Picco.17,
  title={Phase separation for the long range one-dimensional Ising model},
  author={Cassandro, M. and Merola, I. and Picco, P.},
  journal={Journal of Statistical Physics},
  volume={167},
  pages={351--382},
  year={2017},
}

@article{Imbrie.82,
  title={Decay of correlations in the one dimensional Ising model with {${J_{ij}=|i-j|^{-2}}$}},
  author={Imbrie, J. Z},
  journal={Communications in Mathematical Physics},
  volume={85},
  pages={491--515},
  year={1982},
}

@article{Imbrie.Newman.88,
  title={An intermediate phase with slow decay of correlations in one dimensional {${1/|x- y|^2}$} percolation, Ising and Potts models},
  author={Imbrie, J. Z. and Newman, C. M.},
  journal={Communications in Mathematical Physics},
  volume={118},
  pages={303--336},
  year={1988},
}

@article{Bissacot.Endo.18,
  title={Contour methods for long-range Ising models: weakening nearest-neighbor interactions and adding decaying fields},
  author={Bissacot, R. and Endo, E. O and van Enter, A. C. D. and Kimura, B. and Ruszel, W. M.},
  journal={Annales Henri Poincar{\'e}},
  volume={19},
  pages={2557--2574},
  year={2018},
}

@article{Cassandro.05,
  title={Geometry of contours and Peierls estimates in d= 1 Ising models with long range interactions},
  author={Cassandro, M. and Ferrari, P. A. and Merola, I. and Presutti, E.},
  journal={Journal of Mathematical Physics},
  volume={46},
  pages={053305},
  year={2005},
}

@article{Klein.Masooman.97,
  title={Taming Griffiths' Singularities in Long Range Random Ising Models},
  author={Klein, A. and Masooman, S.},
  journal={Communications in Mathematical Physics},
  pages={497--512},
  year={1997},
}

@misc{corsini,
      title={Phase Transitions on 1d Long-Range Ising Models with Decaying Fields: A Direct Proof via Contours}, 
      author={Lucas Affonso and Rodrigo Bissacot and Henrique Corsini and Kelvyn Welsch},
      year={2024},
      eprint={2412.07098},
      archivePrefix={arXiv},
      primaryClass={math-ph}
}

@unpublished{newman_spohn_shiba_relation,
  author    = {C. M. Newman and H. Spohn},
  title     = {The Shiba Relation for the Spin-Boson Model and Asymptotic Decay in Ferromagnetic Ising Models},
  note      = {\emph{Unpublished manuscript}},
  year      = {1999}
}

@article{Littin2017,
  title = {Quasi-additive estimates on the Hamiltonian for the one-dimensional long range Ising model},
  volume = {58},
  ISSN = {1089-7658},
  url = {http://dx.doi.org/10.1063/1.4994034},
  DOI = {10.1063/1.4994034},
  number = {7},
  journal = {Journal of Mathematical Physics},
  publisher = {AIP Publishing},
  author = {Littin,  Jorge and Picco,  Pierre},
  year = {2017},
}

@article{Lebowitz1972,
  title = {Bounds on the correlations and analyticity properties of ferromagnetic Ising spin systems},
  volume = {28},
  ISSN = {1432-0916},
  url = {http://dx.doi.org/10.1007/BF01645632},
  DOI = {10.1007/bf01645632},
  number = {4},
  journal = {Communications in Mathematical Physics},
  publisher = {Springer Science and Business Media LLC},
  author = {Lebowitz,  Joel L.},
  year = {1972},
  pages = {313–321}
}

@book{FV-Book,
place={Cambridge},
title={Statistical Mechanics of Lattice Systems: A Concrete Mathematical Introduction},
publisher={Cambridge University Press},
author={Friedli, S. and Velenik, Y.},
year={2017}
}

@incollection{GMM,
  doi = {10.1007/bfb0112759},
  url = {https://doi.org/10.1007/bfb0112759},
  publisher = {Springer Berlin Heidelberg},
year={1973},
  pages = {162--204},
  author = {G. Gallavotti and A. Martin-L\"{o}f and S. Miracle-Sol{\'{e}}},
  title = {Some problems connected with the description of coexisting phases at low temperatures in the Ising model},
  booktitle = {Statistical Mechanics and Mathematical Problems}
}

@article{Minlos,
  doi = {10.1070/sm1967v002n03abeh002345},
  url = {https://doi.org/10.1070/sm1967v002n03abeh002345},
  year = {1967},
  publisher = {{IOP} Publishing},
  volume = {2},
  number = {3},
  pages = {335--395},
  author = {R A Minlos and Ja G Sinai},
  title = {{The} {phenomenon} {of} {\textquotedblleft}{phase} {separation}{\textquotedblright} {at} {low} {temperatures} {in} {some} {lattice} {models} {of} a {gas}. I},
  journal = {Mathematics of the {USSR}-Sbornik}
}

@article{Fernndez2007,
  title = {Cluster {E}xpansion for {A}bstract {P}olymer {M}odels. {N}ew {B}ounds from an {O}ld {A}pproach},
  volume = {274},
  ISSN = {1432-0916},
  url = {http://dx.doi.org/10.1007/s00220-007-0279-2},
  DOI = {10.1007/s00220-007-0279-2},
  number = {1},
  journal = {Communications in Mathematical Physics},
  publisher = {Springer Science and Business Media LLC},
  author = {Fern\'{a}ndez,  Roberto and Procacci,  Aldo},
  year = {2007},
  pages = {123-140}
}

@article{Koteck1986,
  title = {Cluster expansion for abstract polymer models},
  volume = {103},
  ISSN = {1432-0916},
  url = {http://dx.doi.org/10.1007/BF01211762},
  DOI = {10.1007/bf01211762},
  number = {3},
  journal = {Communications in Mathematical Physics},
  publisher = {Springer Science and Business Media LLC},
  author = {Koteck\'{y},  R. and Preiss,  D.},
  year = {1986},
 pages = {491-498}
}

@article{Jansen2022,
  title = {Cluster {E}xpansions: {N}ecessary and {S}ufficient {C}onvergence {C}onditions},
  volume = {189},
  ISSN = {1572-9613},
  url = {http://dx.doi.org/10.1007/s10955-022-02992-6},
  DOI = {10.1007/s10955-022-02992-6},
  number = {3},
  journal = {Journal of Statistical Physics},
  publisher = {Springer Science and Business Media LLC},
  author = {Jansen,  Sabine and Kolesnikov,  Leonid},
  year = {2022},
}

@misc{procacci2023cluster,
      title={Cluster expansion methods in rigorous statistical mechanics}, 
      author={Aldo Procacci},
      year={2023},
      eprint={2308.06380},
      archivePrefix={arXiv},
      primaryClass={math-ph}
}

@article{Penrose1963,
  title = {Convergence of Fugacity Expansions for Fluids and Lattice Gases},
  volume = {4},
  ISSN = {1089-7658},
  url = {http://dx.doi.org/10.1063/1.1703906},
  DOI = {10.1063/1.1703906},
  number = {10},
  journal = {Journal of Mathematical Physics},
  publisher = {AIP Publishing},
  author = {Penrose,  Oliver},
  year = {1963},
  pages = {1312-1320}
}

@incollection {Ott+Velenik-2023,
    AUTHOR = {Ott, Sebastien and Velenik, Yvan},
     TITLE = {Asymptotics of correlations in the {I}sing model: a brief
              survey},
 BOOKTITLE = {Topics in statistical mechanics},
    SERIES = {Panor. Synth\`eses},
    VOLUME = {59},
     PAGES = {157--184},
 PUBLISHER = {Soc. Math. France, Paris},
      YEAR = {2023},
      ISBN = {978-2-85629-970-8},
   MRCLASS = {82B20},
  MRNUMBER = {4619313},
}

@article{Aoun2021,
  title = {Sharp asymptotics of correlation functions in the subcritical long-range random-cluster and Potts models},
  volume = {26},
  ISSN = {1083-589X},
  url = {http://dx.doi.org/10.1214/21-ECP390},
  DOI = {10.1214/21-ecp390},
  journal = {Electronic Communications in Probability},
  publisher = {Institute of Mathematical Statistics},
  author = {Aoun,  Yacine},
  year = {2021},
}

@article{Aoun2023,
  title = {On the Two-Point Function of the Ising Model with Infinite-Range Interactions},
  volume = {190},
  ISSN = {1572-9613},
  url = {http://dx.doi.org/10.1007/s10955-023-03175-7},
  DOI = {10.1007/s10955-023-03175-7},
  number = {11},
  journal = {Journal of Statistical Physics},
  publisher = {Springer Science and Business Media LLC},
  author = {Aoun,  Yacine and Khettabi,  Kamil},
  year = {2023},
}

@article {Aoun+Ott+Velenik-2024,
    AUTHOR = {Aoun, Yacine and Ott, S\'ebastien and Velenik, Yvan},
     TITLE = {Ornstein-{Z}ernike behavior for {I}sing models with
              infinite-range interactions},
   JOURNAL = {Ann. Inst. Henri Poincar\'e{} Probab. Stat.},
  FJOURNAL = {Annales de l'Institut Henri Poincar\'e{} Probabilit\'es et
              Statistiques},
    VOLUME = {60},
      YEAR = {2024},
    NUMBER = {1},
     PAGES = {167--207},
      ISSN = {0246-0203,1778-7017},
   MRCLASS = {60K35 (82B20 82D60)},
  MRNUMBER = {4718378},
       DOI = {10.1214/22-aihp1345},
       URL = {https://doi.org/10.1214/22-aihp1345},
}

@article{Lebowitz1968,
  title = {Analytic and clustering properties of thermodynamic functions and distribution functions for classical lattice and continuum systems},
  volume = {11},
  ISSN = {1432-0916},
  url = {http://dx.doi.org/10.1007/BF01645899},
  DOI = {10.1007/bf01645899},
  number = {2},
  journal = {Communications in Mathematical Physics},
  publisher = {Springer Science and Business Media LLC},
  author = {Lebowitz,  J. L. and Penrose,  O.},
  year = {1968},
  pages = {99–124}
}

@article{Pfister1991LargeDA,
  title={Large deviations and phase separation in the two-dimensional Ising model},
  author={C. E. Pfister},
  journal={Helvetica Physica Acta},
  year={1991},
  volume={64},
  pages={953-1054},
  url={https://api.semanticscholar.org/CorpusID:119812398}
}

@article{Iagolnitzer1977,
  title = {Decay of correlations for slowly decreasing potentials},
  volume = {16},
  ISSN = {0556-2791},
  url = {http://dx.doi.org/10.1103/PhysRevA.16.1700},
  DOI = {10.1103/physreva.16.1700},
  number = {4},
  journal = {Physical Review A},
  publisher = {American Physical Society (APS)},
  author = {Iagolnitzer,  D. and Souillard,  B.},
  year = {1977},
  pages = {1700–1704}
}

@article{Aizenman1988,
  title = {Critical exponents for long-range interactions},
  volume = {16},
  ISSN = {1573-0530},
  number = {1},
  journal = {Letters in Mathematical Physics},
  publisher = {Springer Science and Business Media LLC},
  author = {Aizenman,  Michael and Fernández,  Roberto},
  year = {1988},
  pages = {39–49}
}

\end{document}